\documentclass{article}
\RequirePackage{etex} %
\usepackage[margin=1in]{geometry}
\usepackage{setspace}
\usepackage{amssymb}%
\usepackage{pifont}%

\usepackage{multirow}
\usepackage{colortbl}
\usepackage{threeparttable}
\usepackage{acronym}

\usepackage{microtype}
\usepackage{graphicx}
\usepackage{subfigure}
\usepackage{booktabs}
\usepackage{hyperref}

\usepackage{natbib}

\usepackage{bm}
\usepackage{amsmath}
\usepackage{amssymb}
\usepackage{mathtools}
\usepackage{amsthm}
\usepackage{amsfonts}

\usepackage{aliascnt} %

\theoremstyle{plain}
\newtheorem{theorem}{Theorem}[section]

\newaliascnt{proposition}{theorem}
\newtheorem{proposition}[proposition]{Proposition}
\aliascntresetthe{proposition}

\newaliascnt{lemma}{theorem}
\newtheorem{lemma}[lemma]{Lemma}
\aliascntresetthe{lemma}

\newaliascnt{corollary}{theorem}
\newtheorem{corollary}[corollary]{Corollary}
\aliascntresetthe{corollary}

\theoremstyle{definition}

\newaliascnt{definition}{theorem}
\newtheorem{definition}[definition]{Definition}
\aliascntresetthe{definition}

\newaliascnt{assumption}{theorem}
\newtheorem{assumption}[assumption]{Assumption}
\aliascntresetthe{assumption}

\newaliascnt{remark}{theorem}
\newtheorem{remark}[remark]{Remark}
\aliascntresetthe{remark}

\newaliascnt{fact}{theorem}
\newtheorem{fact}[fact]{Fact}
\aliascntresetthe{fact}

\usepackage{cleveref}
\crefname{theorem}{theorem}{theorems}
\Crefname{theorem}{Theorem}{Theorems}
\crefname{lemma}{lemma}{lemmas}
\Crefname{lemma}{Lemma}{Lemmas}
\crefname{proposition}{proposition}{propositions}
\Crefname{proposition}{Proposition}{Propositions}
\crefname{corollary}{corollary}{corollaries}
\Crefname{corollary}{Corollary}{Corollaries}
\crefname{definition}{definition}{definitions}
\Crefname{definition}{Definition}{Definitions}
\crefname{assumption}{assumption}{assumptions}
\Crefname{assumption}{Assumption}{Assumptions}
\crefname{remark}{remark}{remarks}
\Crefname{remark}{Remark}{Remarks}
\crefname{fact}{fact}{facts}
\Crefname{fact}{Fact}{Facts}

\usepackage{todonotes}

\usepackage[algo2e,lined,boxed,commentsnumbered,ruled]{algorithm2e}
\usepackage{setspace}
\RestyleAlgo{ruled} 

\usepackage{enumitem}
\usepackage{nicefrac}       %
\usepackage{textcomp}

\usepackage{comment}

\usepackage{autonum}
\definecolor{mygray}{gray}{0.8}

\DeclareMathOperator*{\argmax}{argmax}

\newcommand{\lrset}[1]{\left\{ #1 \right\}}
\newcommand{\lrp}[1]{\left( #1 \right)}
\newcommand{\lrs}[1]{\left[ #1 \right]}

\newcommand{\KLinf}{\operatorname{KL_{inf}}}
\newcommand{\KL}{\operatorname{KL}}

\newcommand{\Exp}[2]{\mathbb{E}_{#1}\lrs{#2}}

\newcommand{\R}{\mathbb{R}}

\newcommand{\calP}{\mathcal P}
\newcommand{\calF}{\mathcal F}
\newcommand{\calG}{\mathcal G}
\newcommand{\calC}{\mathcal C}
\newcommand{\calJ}{\mathcal J}
\newcommand{\calQ}{\mathcal Q}
\newcommand{\calX}{\mathcal X}
\newcommand{\calD}{\mathcal D}

\newcommand{\calB}{\mathcal B}
\newcommand{\calL}{\mathcal L}

\newcommand{\calK}{\mathcal{K}}
\newcommand{\calT}{\mathcal T}
\newcommand{\calM}{\mathcal M}

\newcommand{\calY}{\mathcal Y}

\newcommand{\avec}{\mathbf{a}}

\newcommand{\cvec}{\mathbf{c}}
\newcommand{\dvec}{\mathbf{d}}

\newcommand{\ivec}{\mathbf{i}}

\newcommand{\mvec}{\mathbf{m}}

\newcommand{\svec}{\mathbf{s}}

\newcommand{\uvec}{\mathbf{u}}
\newcommand{\vvec}{\mathbf{v}}

\newcommand{\xvec}{\mathbf{x}}
\newcommand{\yvec}{\mathbf{y}}
\newcommand{\zvec}{\mathbf{z}}

\newcommand{\qtext}[1]{\quad \text{#1} \quad}

\usepackage{xspace}

\newcommand{\Qepsilon}{Q^{(\epsilon)}}
\newcommand{\Qtilde}{\widetilde{Q}}
\newcommand{\qtilde}{\widetilde{q}}

\newcommand{\iid}{i.i.d.\xspace}
\newcommand{\boldlambda}{\boldsymbol{\lambda}}
\newcommand{\boldmu}{\boldsymbol{\mu}}

\newcommand{\boldrho}{\boldsymbol{\rho}}
\DeclareMathOperator{\supp}{supp}

\newcommand{\Vertices}{\mathrm{Vert}}
\newcommand{\dualdomain}{\calL_{\boldmu}}
\newcommand{\epsdualdomain}{\calL_{\boldmu}^{(\epsilon)}}
\newcommand{\finitesupport}{\mathbb{X}}
\newcommand{\ftilde}{\widetilde{f}}

\newcommand{\fdivdualdomain}{\calL_{\boldmu, f}}
\newcommand{\epsfdivdualdomain}{\calL_{\boldmu, f}^{(\epsilon)}}
\newcommand{\lsc}{\text{lsc}\xspace}

\title{Dual Representation of Minimum Divergence Under Integral Constraints}
\date{}
\author{Shubhanshu Shekhar \\ EECS Department \\ University of Michigan, Ann Arbor \\ \url{shubhan@umich.edu} \and 
Shubhada Agrawal \\ ECE Department \\ Indian Institute of Science, Bangalore \\ \url{shubhada@iisc.ac.in}}

\begin{document}
\maketitle

\begin{abstract}
Minimum divergence problems under integral constraints appear throughout statistics and probability, including sequential inference, bandit theory, and distributionally robust optimization. In many such settings, dual representations are the key step that convert information-theoretic lower bounds into computationally tractable (and often near-optimal) algorithms.  In this paper, we present a general two-stage recipe for deriving dual representations of constrained minimum divergence~(in the second argument) for distributions supported on $[0,1]^K$. The first stage derives a dual representation for finitely-supported distributions using classical finite-dimensional convex duality techniques, while the second establishes an abstract interchange argument that lifts this discretized dual to arbitrary distributions. 

We begin with the simplest case of mean-constrained minimum relative entropy, commonly called $\KLinf$, and generalize an existing argument from multi-armed bandits literature for $K=1$ to arbitrary dimensions. Our main contribution is to significantly expand the scope of this approach  to a broad class of $f$-divergences~(beyond relative entropy) and to general integral constraint functionals~(beyond the mean constraint). Finally, we illustrate the statistical implications of our results by constructing optimal procedures for sequential testing, estimation, and change detection with observations in $[0,1]^K$.  
\end{abstract}

\setcounter{tocdepth}{2}
\tableofcontents
\section{Introduction}
We study constrained minimum-divergence problems for probability measures supported on a compact domain, which for concreteness we set to $\calX = [0,1]^K$ for an integer $K \geq 1$.  Let $\calP(\calX)$ denote the set of all probability measures on $(\calX, \calB_{\calX})$ with $\calB_{\calX}$ denoting the Borel sigma-field on $\calX$, and let $D:\calP(\calX) \times \calP(\calX) \to [0, \infty]$ denote a divergence measure. Consider a continuous constraint function $g:\calX \to \R^J$~(for $J \geq 1$), and let $\calC \subset \R^J$ be a closed and convex set. We are interested in studying the minimum divergence term 
\begin{align}
    I(P, g, \calC) = \inf_{Q \in \calQ(g, \calC)} D(P, Q), \qtext{where} \calQ \equiv \calQ(g, \calC) = \{Q \in \calP(\calX): \mathbb{E}_Q[g(X)] \in \calC \}.  \label{eq:general-min-divergence}
\end{align}
At a high-level, $I(P, g, \calC)$ quantifies how far a distribution $P$ is from a class of distributions satisfying an integral constraint. Such quantities appear as the intrinsic ``hardness'' parameters governing both the instance-dependent information-theoretic lower bounds and achievability results in several sequential decision-making problems, including multi-armed bandits~(MABs) \citep{lai1985asymptotically,burnetas1996optimal}, sequential testing \citep{agrawal2025stopping}, estimation, change detection, and distributionally robust optimization~\citep{bayraksan2015data, hu2013kullback}. %

The simplest instance of $I(P, g, \calC)$ is the mean-constrained minimum relative entropy functional, often referred to as $\KLinf$. For $P \in \calP([0,1])$, and a target mean value $\mu \in [0,1]$, it is defined as 
\begin{align}
    \KLinf(P, \mu) = \inf \{\KL(P, Q): Q \in \calP([0,1]), \; \mathbb{E}_{Q}[X] = \mu \}, \qtext{where} \KL(P, Q) = \int \log \lrp{\frac{dP}{dQ}} dP  
\end{align}
for $P \ll Q$, and $\KL(P, Q) = + \infty$ otherwise. 
This quantity governs the fundamental limits of sequential inference problems with bounded observations. 
Concretely, let $\{X_i: i \geq 1\}$ denote an \iid sequence drawn from a distribution $P$ supported on $\calX = [0,1]$ with an unknown mean $\mu_P$. Fix a candidate mean value $\mu \in (0,1)$, and consider the null hypothesis $H_0: \mu_P=\mu$ for a given $\mu \in [0,1]$. For a prespecified $\alpha \in (0, 1)$, suppose our goal is to construct a level-$\alpha$ power one test $\tau_\alpha$ for this problem~\citep{darling1968some},  which is a random stopping time satisfying $\mathbb{P}_{H_0}(\tau_\alpha < \infty) \leq \alpha$~(the ``level-$\alpha$ property under the null) and $\mathbb{P}_{H_1}(\tau_\alpha < \infty) = 1$~(the ``power-one'' property under the alternative). 
 A standard argument based on the data processing inequality for randomly stopped processes shows that any such level-$\alpha$ test $\tau_\alpha$~(uniformly over all null distributions) for this problem must satisfy the following fundamental lower bound \citep{garivier2019explore}
\begin{align}
    \mathbb{E}[\tau_\alpha] \geq \frac{\log(1/\alpha)}{\KLinf(P, \mu)}, \qtext{whenever} \mu_P \neq \mu. \label{eq:testing-lower-bound-1d}
\end{align}
This result has an intuitive interpretation. When the null is false, the maximum rate at which any level-$\alpha$ test can collect evidence against the null is controlled by the relative entropy between $P$ and the closest candidate in the class of null distributions $\calQ_{\mu} = \{Q: \mathbb{E}_Q[X] = \mu\}$. 
Moreover, the lower bound in~\eqref{eq:testing-lower-bound-1d} also suggests a constructive strategy that matches this lower bound asymptotically as $\alpha\downarrow 0$. In particular, the test 
\begin{align}
    \tau_\alpha = \inf \{n \geq 1: \KLinf(\widehat{P}_n, \mu) \geq f(n, \alpha) \},  \qtext{for a function} f(n, \alpha) \asymp \frac{\log(n/\alpha)}{n},
\end{align}
can be shown to satisfy $\lim_{\alpha\rightarrow 0}\frac{\mathbb{E}[\tau_{\alpha}]}{\log(1/\alpha)} = \frac{1}{\KLinf(P,\mu)}$, which matches the lower bound in the limit of $\alpha \to 0$ \citep{agrawal2020optimal,jourdan2022top}. A crucial reason such procedures are computationally feasible is that $\KLinf$ admits an explicit dual representation. For instance, in the bounded one-dimensional setting,~\citet{honda2010asymptotically} derived the following dual: 
\begin{align}
\KLinf(P, \mu) = \inf_{Q \in \calQ_\mu} \KL(P, Q) = \sup_{\lambda:  \sup_{x \in [0,1]}\lambda(x-\mu) \leq 1} \mathbb{E}_P[\log(1 - \lambda(X-\mu)]. \label{eq:klinf-one-dim}
\end{align}
Since this is a convex program over a compact one-dimensional domain~(even though the primal problem involves minimization over an infinite-dimensional space of probability measures with mean $\mu$), it can be solved efficiently using off-the-shelf solvers. This brief discussion illustrates how $\KLinf$ characterizes the fundamental lower bound, and its dual representation enables the construction of a computationally feasible method that nearly matches it. A parallel argument also applies to several other problems, including testing statistics beyond the mean (e.g., quantiles, CVaR) and beyond bounded-support distributions \citep{agrawal2021regret,agrawal2021optimal}, as well as to the problems mentioned above.

\subsection{Related Work and Overview of Our Results}
Duality for divergence minimization under integral constraints is a classical topic in convex analysis and information theory. A substantial body of work studies  the dual representation of general $f$-divergence  projection problems, typically in the first argument, under finitely many moment-type constraints. These works use  Fenchel duality ideas for convex integral functionals; see, for example, \citet{broniatowski2006minimization, broniatowski2012divergences, borwein1991duality} and follow-up works.  In this setting, one typically considers the minimum divergence between a reference measure and a class of measures defined through finitely-many integral constraints, with the aim of  identifying conditions ensuring equality of the primal and dual problems. Another important objective is to provide a characterization of the optimizers (i.e., the divergence projection). A central issue in their approach is that of \emph{constraint qualification}. Classical strong-duality results often require the primal feasible set to have a nonempty relative interior. However, as noted by~\citet{borwein1992partially} for measure spaces endowed with weak topologies, the relative interiors can often be empty even for standard integral constraints. To address this,~\citet[Definition 2.3]{borwein1992partially} introduced the notion of quasi-relative interior, and~\citet[Theorem 1.1]{broniatowski2006minimization} used this framework to give sufficient conditions for ensuring strong duality and dual attainment for $f$-divergence minimization~(in the first argument) under integral constraints, along with an explicit form of the optimal projection under additional regularity assumptions. 

Our objective in this paper is complementary to the this line of work. First, we study  divergence minimization in the second argument motivated by some applications in anytime-valid inference. 
Second, while the prior results provide dual representations for the minimum divergence problems they consider, they are often implicit. In contrast, 
we develop an elementary and constructive discretization-based strategy for distributions supported on compact domains, which results in more explicit dual formulations. Specifically, we  approximate the original problem (over an infinite dimensional domain) with a sequence of discretized problems with finite-dimensional domains. The sufficient conditions  for strong duality in these intermediate problems can then be easily verified. We then transfer this intermediate dual to the original continuous problem via a careful limiting argument along a sequence of discretizations. Interestingly, our limiting argument has an information-theoretic flavor as it explicitly relies on  the data processing inequality~(DPI) and lower semicontinuity (\lsc)
in the weak topology of $f$-divergences. 
This two-stage discretization-based methodology is motivated by the argument developed by~\citet{honda2010asymptotically} for minimum relative entropy under mean constraints~(i.e., $\KLinf$) for distributions supported on $[0,1]$. 
We build upon it to develop a modular two-stage approach,  that (i) applies naturally to the case of distributions on higher-dimensional supports $[0,1]^K$, and (ii) extends to a much broader class of divergence measures~(beyond relative entropy) and integral constraints~(beyond the mean constraint).

\paragraph{Overview of our results.} As mentioned above, our general method for obtaining the dual representation follows the template of~\citet[\S~4.1]{honda2010asymptotically}. First, we derive the dual for distributions with finite support on $\calX = [0,1]^K$. In this case, we show that we can restrict our attention to a finite-dimensional subset of the primal domain, and the dual then follows by standard Lagrangian arguments. More specifically, this gives us an explicit concave dual objective function  over a compact finite-dimensional dual domain. The second step in our pipeline is to transfer this dual from finitely supported distributions to general distributions $P \in \calP(\calX)$ by constructing a sequence of discretization channels~(or stochastic transformations) $\{\calK_k:k \geq 1\}$. For each discretized problem, we can use the finitely supported dual, and we show that the limiting value of this is the dual of the original problem.  Our argument relies on DPI for relative entropy, its lower semicontinuity in the weak topology, and the concavity of the finitely-supported dual objective. Since these properties are also satisfied by a larger class of $f$-divergence measures, we show that the same pipeline also allows a direct derivation of their duals as well under mean constraint. Such quantities arise in important applications such as variational Bayesian inference, distributionally robust optimization, and generalized empirical likelihood methods. 

One new component in our derivation of the dual of mean-constrained $\KLinf$~(as compared to~\citet{honda2010asymptotically}) is that we construct discretization channels based on the idea of stochastic rounding~\citep{croci2022stochastic} that exactly preserve the mean constraint. We illustrate that this same technique also extends to a restricted class of constraints beyond the mean-constraint. However, when we move to arbitrary continuous constraint functions, exact preservation on discretization is  generally not possible. To address such situations, we obtain an abstract dual representation theorem that follows the same two-step recipe mentioned above, but now allows for approximately satisfied constraints in the intermediate sequence of discretized problems. This construction, detailed in~\Cref{sec:general-argument-approximation-constraints} in an abstract form, allows for extending the two-stage approach for studying the dual of the general minimum divergence term $I(P, g, \calC)$ introduced in~\eqref{eq:general-min-divergence}.
We show that we can obtain the dual representation of $I(P, g, \calC)$ in terms of the corresponding finite-support duals under some verifiable conditions such as (i) the DPI and \lsc properties of $D$, (ii) the continuity of $g$ and compactness of $\calX =[0,1]^K$, (iii) mild regularity of the finite-support dual objective and domain. As for mean-constrained problems, our proof is transparent and constructive, and relies on elementary tools from analysis.

\paragraph{Organization.}  We derive the dual of mean-constrained minimum relative entropy in~\Cref{sec:Klinf-mean-constrained}, presenting the result for finitely supported distributions in~\Cref{theorem:finite-support}, followed by the limiting argument in~\Cref{theorem:continuous-klinf}. We then discuss how to extend this to general $f$-divergences in~\Cref{subsec:other-divergences}, and go beyond the mean-constraint to general continuous constraints in~\Cref{subsec:other-constraints}. 
\Cref{sec:general-argument-approximation-constraints} contains the key technical result of this paper that extends this two-stage dual derivation to general divergences and constraints, and we apply this to derive the dual of $\KLinf$ with general constraints in~\Cref{subsec:other-constraints}.  
Finally, in~\Cref{sec:statistical-applications}, we apply the dual $\KLinf$ term to construct and analyze optimal procedures for sequential anytime-valid inference with $[0,1]^K$-valued observations. We defer all the proofs to the appendices. 

\section{Warmup: Dual $\KLinf$}
\label{sec:Klinf-mean-constrained}
Throughout this section, we will use $\calB \equiv \calB_\calX$ to denote the Borel sigma-algebra on the unit cube $\calX = [0,1]^K$ with $K \geq 1$ dimensions. For any probability measure $P$ on $(\calX, \calB)$, and a mean vector $\boldmu$ lying in the interior $\mathring{\calX}$, we are interested in obtaining a dual characterization of 
\begin{align}
    \KLinf(P, \boldmu) = \inf \{\KL(P, Q): Q \in \calP(\calX), \; \mathbb{E}_Q[X]=\boldmu\}. 
\end{align}
As mentioned in the introduction, we will use a two-step approach to obtain the dual motivated by~\citet{honda2010asymptotically}. First, we consider the simpler case of $P$ supported on a finite subset of $\calX$. In this simplified setting, we can use standard arguments for finite-dimensional problems to obtain an explicit dual representation. We state this formally below.

\begin{proposition}
    \label{theorem:finite-support} 
        Let $P \in \calP(\calX)$ be a distribution with a finite support $\finitesupport$, and $\mu \in \mathring{\calX}$.  Then, we have 
    \begin{align}
        \KLinf(P, \boldmu) \coloneqq \inf \left\{ \KL(P, Q):\; Q \in \calP(\calX), \; \mathbb{E}_Q[X] = \boldmu \right\} = \sup_{\boldlambda \in \calL_{\boldmu}} \mathbb{E}_P\left[ \log \lrp{1 - \boldlambda^T(X-\boldmu)}\right], \label{eq:klinf-mean-constrained-primal}
    \end{align}
    where $\calL_{\boldmu} = \{\boldlambda \in \mathbb{R}^K: 1 - \boldlambda^T(x - \boldmu) \geq 0, \; \forall x \in \calX\}$.
\end{proposition}
The proof of this result~(details in~\Cref{proof:finite-support}) proceeds in three steps. The first step is to observe that even though the domain of optimization of the primal problem in~\eqref{eq:klinf-mean-constrained-primal} is infinite dimensional, we can restrict our attention to a smaller finite-dimensional subspace without loss of optimality. Next, we verify that the finite-dimensional optimization problem resulting in Step 1 satisfies the sufficient conditions for strong duality to hold. Finally, to complete the argument, we use the classical Lagrangian duality theory for finite-dimensional convex programs along with a ``scaling trick'' to achieve the simplified form stated in~\Cref{theorem:finite-support}.

The second step in our approach is to consider a general $P \in \calP(\calX)$, and approximate its dual as a limiting value of a sequence of duals associated with discretized versions of $P$, denoted by $\{P_k: k \geq 1\}$. Since~\Cref{theorem:finite-support} is applicable to all such intermediate problems, the main task it essentially to justify an interchange of ``$\lim$'' and ``$\sup$'' as we show below. But, before presenting the formal argument, we need to introduce some notation. 
\begin{definition}
    \label{def:grid-and-vertices}
    For any $k \geq 1$, let $\Delta_k = 2^{-k}$ denote the ``mesh-size'', and let $G_k \subset \calB$ denote a dyadic grid consisting of cubes~(or cells) of sides $\Delta_k$. Formally, we introduce the term 
    \begin{align}
        G_k = \left\{ \prod_{j=1}^K J_{i_j}^k: \ivec= (i_1, \ldots, i_K) \in \{0, 1, \ldots, 2^k-1\}^K\right\}, \; \text{where} \;  
        J_{i}^k = \begin{cases}
            [i \Delta_k, \,(i+1)\Delta_k), & 0 \leq i \leq 2^k - 2, \\
            [1-\Delta_k, 1], & i=2^k - 1. 
        \end{cases}
    \end{align}
    For each $\xvec \in \calX$, we then define $E_k(\xvec)$ as the unique cell in $G_k$ containing $\xvec$ and let $\avec(\xvec) = (a_1(\xvec), \ldots, a_K(\xvec))$ denote the corner of $E_k(\xvec)$ that is minimal in coordinate-wise ordering. In other words, we can write 
    \begin{align}
        E_k(\xvec) = \prod_{j=1}^K [a_j(\xvec), \, a_j(\xvec) + \Delta_k \}, \qtext{where} ``\}" =
        \begin{cases}
            ``)", & \text{ if } a_j(\xvec)  < 1 -  \Delta_k, \\
            ``]", &  \text{ if } a_j(\xvec) = 1 - \Delta_k. 
        \end{cases} 
    \end{align}
    For any cell $E \in G_k$ with corner $\avec \in \{0, \ldots, 1-\Delta_k\}^K$, let $\Vertices(E)$ denote its associated \emph{vertex set} $\{\avec + s \Delta_k: \svec \in \{0,1\}^K\}$, and for any $\xvec \in \calX$, we use $\Vertices(\xvec) = \Vertices(E_k(\xvec))$. Finally, we introduce the term 
    \begin{align}
        V_k = \cup_{E \in G_k} \Vertices(E) = \cup_{\xvec \in \calX} \Vertices(\xvec) =  \{0, \Delta_k, \ldots, 1\}^K, 
    \end{align}
    denote the grid of all vertices of cells in $G_k$.
\end{definition}

 We now introduce a key idea of \emph{mean-preserving channel} that we will use throughout this section. This definition is motivated by the concept of \emph{stochastic rounding} used in finite-precision representation of real numbers~\citep{croci2022stochastic}. 
 
\begin{definition}[Mean-preserving channel]
    \label{def:mean-preserving-channel}
    For any $k \geq 1$, let $G_k$ denote the collection of dyadic cubes of sides $\Delta_k = 2^{-k}$, and let $V_k$ denote the associated set of vertices introduced in~Definition~\ref{def:grid-and-vertices}. 
    We define the mean-preserving discretization Markov kernel~(or channel) $\calK_k:2^{V_k} \times \calX \to [0,1]$ as follows: Consider any $\xvec = (x_1, \ldots, x_K) \in \calX$, and let  $E_k(\xvec)$ denote the unique cube in~$G_k$ containing $\xvec$, with minimal vertex $\avec(\xvec) = (a_1(\xvec),\ldots, a_K(\xvec)) \in \{0, \Delta_k, \ldots, 1-\Delta_k\}^K$. Then, the conditional distribution $\calK_k(\cdot|\xvec)$ is supported on $\Vertices(\xvec)$, and for any $\vvec = \avec(\xvec) + \svec \Delta_k$ with $\svec = (s_1, \ldots, s_K) \in \{0,1\}^K$, we have 
    \begin{align}
        &\calK_k(\{\vvec\}|\xvec)  = \prod_{j=1}^K\lrp{ s_j\lrp{ \frac{x_j - a_j(x_j)}{\Delta_k} }  + (1-s_j) \lrp{ 1 -\frac{x_j - a_j(x_j)}{\Delta_k} } }. 
    \end{align}
\end{definition}
In other words, let $Y$ be any random variable with distribution $Q \in \calP(\calX)$ and $\calX = [0,1]^K$, and let $Y_k$ denote the output after passing $Y$ through the channel $\calK_k$, with distribution $Q_k = Q \calK_k$; that is, $Q_k(E) = \int_{\calX} \calK_k(E \mid \xvec) dQ(\xvec)$. For  a realization of $Y = \yvec = (y_1, \ldots, y_K)$, we can write $Y_k = \avec(\yvec) + \Delta_k B$, where $B = (B_1, \ldots, B_K)$, and 
\begin{align}
    B_i \mid (Y = \yvec) \; \sim \; \mathrm{Bernoulli}\lrp{ \frac{y_j-a_j(y_j)}{\Delta_k}}, \qtext{and} B_i \perp B_j \mid (Y=\yvec). 
\end{align}
As a result, the above construction of $\calK_k$ ensures that 
\begin{align}
    \mathbb{E}[Y_k|Y] = Y \quad \text{almost surely}, \qtext{which implies} \mathbb{E}[Y_k] = \mathbb{E}[Y],  
\end{align}
illustrating the mean-preserving property of $\calK_k$ for every $k \geq 1$. Additionally, we also have the approximation result $\|Y_k - Y\|_\infty \leq \Delta_k$ almost surely.

Throughout this section, we will use $P_k = P \calK_k $ and $Q_k = Q \calK_k $ to represent the distributions obtained after passing $P$ and $Q$ through $\calK_k$. That is, for any $A \in \calB$, we have 
\begin{align}
   P_k(A) = \int_{\calX} \calK_k(A \mid \xvec) dP(\xvec), \qtext{and} Q_k(A) = \int_{\calX} \calK_k(A \mid \xvec) dQ(\xvec),
\end{align}
where $P, Q$ are probability measures on $(\calX, \calB)$. Before proceeding to the main results of this section, we present a simple but interesting consequence of the fact that $\calB = \sigma \lrp{\cup_{k=1}^\infty G_k}$. 
\begin{lemma}
    \label{lemma:convergence-of-KL} 
    Let $P$ and $Q$ denote two probability measures on $(\calX, \calB)$. 
    For any $k \geq 1$, let $\calK_k$ denote the mean preserving channel, and  let $P_k = \calK_k P$ and $Q_k = \calK_k Q$ denote the corresponding pushforward measures. Then, assuming that $P \ll Q$ and $\KL(P, Q) < \infty$,  we have the following: 
    \begin{align}
        \lim_{k \to \infty} \KL(P_k, Q_k) = \KL(P, Q). 
    \end{align}
\end{lemma}
\begin{proof}
We first observe that $P_k \Longrightarrow P$ and $Q_k \Longrightarrow Q$  as $k \to \infty$,  where $\Longrightarrow$ denotes weak convergence. To see this, let $f:\calX \to \mathbb{R}$ denote any bounded continuous function. Since the domain $\calX$ is compact, this also means that $f$ is uniformly continuous.  Hence, for every $\epsilon>0$, there exists a $k_\epsilon$, such that for all $k \geq k_\epsilon$, we have $\sup_{\|x-x'\|\leq 2^{-k}}|f(\xvec)-f(\xvec')|\leq \epsilon$. This leads to the following inequalities with $X_k \sim P_k =  P\calK_k$, and with $k \ge k_\epsilon$: 
\begin{align}
    |\mathbb{E}[f(X_k)] - \mathbb{E}[f(X)]| \leq \mathbb{E}\left[ \mathbb{E}\left[ |f(X_k) - f(X)| \mid  X \right] \right] \leq \sup_{\xvec, \xvec': \|\xvec-\xvec'\|_\infty \leq 2^{-k}} |f(\xvec) - f(\xvec')| \leq \epsilon. 
\end{align}
In other words, for every bounded continuous $f$, we have $\lim_{k \to \infty} \mathbb{E}_{P_k}[f(X_k)] = \mathbb{E}_{P}[f(X)]$,
which means that $P_k \Longrightarrow P$. An exact same argument implies that $Q_k \Rightarrow Q$.  

Having established the weak convergence of $P_k$ and $Q_k$ to $P$ and $Q$ respectively, we note that  the joint lower semicontinuity under weak-convergence of relative entropy  implies 
\begin{align}
    \liminf_{k \to \infty} \KL(P_k,  Q_k)  \geq \KL(\lim_{k \to \infty} P_k, \lim_{k \to \infty} Q_k) = \KL(P, Q). 
\end{align}
On the other hand, for any fixed $k$, the data processing inequality~\citep[Theorem 2.17]{polyanskiy2025information} for relative entropy implies that 
\begin{align}
    \KL(P_k, Q_k) \leq \KL(P, Q), \qtext{hence} \limsup_{k \to \infty} \KL(P_k, Q_k) \leq \KL(P, Q). 
\end{align}
Combining the previous two displays, we obtain 
\begin{align}
    \KL(P, Q) \leq \liminf_{k \to \infty} \KL(P_k, Q_k) \leq \limsup_{k \to \infty} \KL(P_k, Q_k) \leq \KL(P, Q). 
\end{align}
This completes the proof. 
\end{proof}

    For any $k \geq 1$, let $H_k$ denote the functional $\boldlambda \mapsto \mathbb{E}_{P_k}[ \log (1 - \boldlambda^T(X_k-\mu))]$,  with $P_k = P \calK_k$.  Then, we know from the dual formulation derived in~\Cref{theorem:finite-support} that 
    \begin{align}
        \KLinf(P_k, \boldmu) = \sup_{\boldlambda \in \calL_{\boldmu}} H_k(\boldlambda), \qtext{where} \calL_{\boldmu} \coloneqq \lrset{\boldlambda \in \R^K: 1 - \boldlambda^T(\xvec-\mu) \geq 0, \; \forall \xvec \in \calX = [0,1]^K}.
    \end{align}
    Our main objective in this section, motivated by~\Cref{lemma:convergence-of-KL}, is to show that a similar dual expression also holds for arbitrary $P$ on $(\calX, \calB)$. That is,  let $H:\calL_\mu \to \R$ denote the functional $\boldlambda \mapsto \mathbb{E}_{P}[ \log (1 - \boldlambda^T(X-\mu))]$, and we want to establish that $\KLinf(P, \boldmu) = \sup_{\boldlambda \in \calL_{\boldmu}} H(\boldlambda)$. This is presented in our next result. 
\begin{theorem}
    \label{theorem:continuous-klinf} The mean-constrained divergence term $\KLinf(P, \boldmu)$ is equal to the limit of $\KLinf$ of a sequence of discretized versions of $P$, denoted by $\{P_k: k \geq 1\}$, constructed by passing $P$ through the sequence of mean-preserving channels $\calK_k$ introduced in~Definition~\ref{def:mean-preserving-channel}. In particular, we have the following chain: 
\begin{align}
    \KLinf(P, \boldmu) &= \inf_{Q: \mathbb{E}_Q[X]=\boldmu} \KL(P, Q) && \text{(by definition})\nonumber \\
   & {=} \lim_{k \to \infty}\;\inf_{Q: \mathbb{E}_{Q}[X]=\boldmu}  \KL(P_k, Q) && (\text{by~\Cref{lemma:interchange-1}}) \label{eq:interchange-1}\\
   &= \lim_{k \to \infty}\;\sup_{\boldlambda \in \calL_{\boldmu}} H_k(\boldlambda) && (\text{by~\Cref{theorem:finite-support}})  \nonumber \\ 
   & {=} \sup_{\boldlambda \in \calL_{\boldmu}} \lim_{k \to \infty} H_k(\boldlambda) = \sup_{\boldlambda \in \calL_{\boldmu}} H(\boldlambda). && (\text{by~\Cref{lemma:interchange-2}}) \label{eq:interchange-2}
\end{align}
Recall that $H_k(\boldlambda) = \mathbb{E}_{P_k}[\log(1-\boldlambda^T(X_k-\boldmu))]$, and $H(\boldlambda) = \mathbb{E}_{P}[\log(1-\boldlambda^T(X-\boldmu))]$. 
\end{theorem}
To prove~\Cref{theorem:continuous-klinf}, we need to justify the two ``interchange operations'' in~\eqref{eq:interchange-1} and~\eqref{eq:interchange-2}. We present the justification of these steps in~\Cref{lemma:interchange-1} and~\Cref{lemma:interchange-2} respectively in~\Cref{proof:continuous-klinf}.  
Together,~Proposition~\ref{theorem:finite-support} and~\Cref{theorem:continuous-klinf} generalize the argument developed by~\citet{honda2010asymptotically} to arbitrary $K \geq 1$. An interesting new component of our proof is the use of the mean-preserving channel~(\Cref{def:mean-preserving-channel}) that allows an exact satisfaction of the equality constraint on discretization. This also allows us to extend this argument to certain constraints beyond the mean-constraint as we discuss briefly in~Appendix~\ref{appendix:box-constraints}, but it breaks down for general continuous constraint functions, necessitating the development of the results of~\Cref{sec:general-argument-approximation-constraints}.  

\begin{remark}
    \label{remark:continuous-klinf-1}
    While stating~\Cref{theorem:continuous-klinf}, we work with a sequence of dyadic partitions $\{G_k: k \geq 1\}$ for simplicity. A similar argument goes through for the case of more general partitions $\{G_k: k \geq 1\}$ whose worst-case diameter $\sup_{E \in G_k} \mathrm{diam}(E) \to 0$, assuming we can define a mean-preserving kernel for the associated vertex sets. A sufficient condition for this is if each $E \in G_k$ can be contained in a prespecified axis-aligned hypercube contained in $\calX$. 
\end{remark}
Before proceeding to the general results in~\Cref{sec:general-argument-approximation-constraints}, we illustrate that our two-stage approach extends beyond relative entropy to a larger class of $f$-divergences. 

\subsection{Beyond Relative Entropy}
\label{subsec:other-divergences}
    The key ingredients in the proof of~\Cref{theorem:continuous-klinf} are (i) the data processing inequality for relative entropy, (ii) the lower semicontinuity of relative entropy under weak convergence on a compact space, and (iii) the uniform continuity of the integrand in the dual objective on $\calX$ on an interior of the domain. Since these conditions can be valid for a larger class of $f$-divergence measures  beyond just relative entropy, our two-stage approach for dual derivation can also be implemented on  this larger class. 

    Let $f:(0, \infty) \to \mathbb{R}$ be a convex, \lsc, function with $f(1)=0$ and $f(0) \coloneqq \lim_{t \to 0} f(t)$, and let $\ftilde$ denote its perspective\citep[\S~2.3.3]{boyd2004convex}; that is
    \begin{align}
        \ftilde(w) = w f(1/w), \quad \text{for } w \geq 0. 
    \end{align}
    Consider any pair of probability measures $P, Q$  on $\calX$, and  decompose $Q$ into $Q_{\mathrm{ac}} + Q_\perp$ with $Q_{\mathrm{ac}} \ll P$~(i.e., $Q_{\mathrm{ac}}$ is absolutely continuous with respect to $P$), and $Q_\perp$ denoting the singular component. Then, the $f$-divergence between $P$ and $Q$ is defined as  
    \begin{align}
        D_f(P \parallel Q) \coloneqq \mathbb{E}_P\left[ \ftilde(W) \right] + f(0) Q_\perp(\calX), \qtext{where} W \coloneqq \frac{dQ_{\mathrm{ac}}}{dP}. 
    \end{align}
    Assume throughout that $\ftilde$ is continuously differentiable and strictly convex, and define the function 
    \begin{align}
        \Phi(r) = \inf_{w \geq 0} \lrp{\ftilde(w) + r w} = -\ftilde^{*}(-r), \qtext{and} U_f = \{r \in \mathbb{R}: \Phi(r)>-\infty\}. 
    \end{align}
    With these definitions, we can state an analog of~\Cref{theorem:continuous-klinf} for the distance of a point $P$ to a set of distributions with mean $\boldmu$ in terms of $f$-divergence $D_f(P \parallel Q)$. Using the perspective $\ftilde$ allows us to write the objective in terms of an expectation in the fixed distribution $P$~(instead of the variable of optimization $Q$). 
    \begin{theorem}
        \label{theorem:f-divergence} Let $P$ denote any distribution on $(\calX, \calB)$, and let $\boldmu \in \mathring{\calX}$ denote any point in the interior of the domain. Then, we have 
        \begin{align}
            D_f^{\inf}(P, \boldmu) \coloneqq \inf_{Q:\mathbb{E}_Q[X]=\boldmu} D_f(P \parallel Q) = \sup_{(\boldlambda, \gamma) \in \calL_{\boldmu, f}} \; \left\{ \mathbb{E}_P [\Phi(\gamma - \boldlambda^T X)] - (\gamma - \boldlambda^T \boldmu) \right\}, 
        \end{align}
        where the dual feasible domain $\calL_{\boldmu, f}$ is defined as 
        \begin{align}
            \fdivdualdomain = \left\{ (\gamma, \boldlambda) \in \mathbb{R}^{K+1} \;:\; \gamma + f(0) \geq \sup_{\boldrho \in \calX} \boldlambda^T \boldrho, \text{ and } \gamma - \boldlambda^T \xvec \in U_f, \;\; \forall \xvec \in \calX \right \}. 
        \end{align}
    \end{theorem}
    \emph{Proof outline.}
        The proof of this result follows the exact pipeline we developed for the case of relative entropy: We first establish the result for the case of $P$ with finite support~(\Cref{theorem:finite-support}) using strong duality for finite-dimensional convex programs, and then extend it to arbitrary $P$ using careful limiting arguments, and appealing to data processing and lower semicontinuity properties of $f$-divergences~\citep[Chapter 7]{polyanskiy2025information}, along with the uniform continuity of the dual objective~(\Cref{theorem:continuous-klinf}). We present the details in~\Cref{proof:f-divergence}. \hfill \qedsymbol

    \begin{remark}
        \label{remark:sanity-check-KL} For relative entropy, we have $f(u) = u \log u$, which implies that $\ftilde(w) = -\log w$ for $w>0$~(and equal to $+\infty$ at $w=0$). Then, we have $\ftilde^{*}(t) = - 1 - \log(-t)$ on $t<0$, which implies that 
        \begin{align}
            \Phi(r) = \inf_{w \geq 0} \left\{ \ftilde(w) + rw \right\} = - \ftilde^{*}(-r) = 1 + \log r, \qtext{for} u \in U = (0, \infty). 
        \end{align}
        Thus,~\Cref{theorem:f-divergence} implies the dual form 
        \begin{align}
            \KLinf(P, \boldmu) = \sup_{\gamma - \boldlambda^T \xvec \geq 0, \; \forall x\in \calX} \left\{\mathbb{E}_P[\log(\gamma - \boldlambda^TX)] + 1 - \gamma + \boldlambda^T \boldmu \right\}. 
        \end{align}
        This is exactly the expression for $\KLinf(P, \boldmu)$ we obtained in~\eqref{eq:KLinf-derivation-no-scaling} while proving~\Cref{theorem:finite-support}.  Using a scaling trick, we can eliminate the dual variable $\gamma$, and obtain the familiar expression of~\Cref{theorem:finite-support}. 
    \end{remark}
    We end this section by specializing~\Cref{theorem:f-divergence} for (squared) Hellinger and Chi-squared divergences.
    \begin{corollary}
        \label{corollary:f-divergence} Let $\dualdomain$ denote the set $\{\boldlambda \in \R^K: 1 - \boldlambda^T(\xvec - \boldmu) \geq  0, \; \forall \xvec \in \calX\}$. Then, we have the following: 
        \begin{align}
            D_{\mathrm{Hel}}^{\inf}(P, \boldmu) &= \sup_{\boldlambda \in \dualdomain} \left(2 - 2 \sqrt{\mathbb{E}_P\left[ \frac{1}{1 - \boldlambda^T(X-\boldmu)} \right]}  \right), && \left( f(u) = (\sqrt{u}-1)^2 \right), \\
            D_{\chi^2}^{\inf}(P, \boldmu)&= \sup_{\boldlambda \in \dualdomain}\left[ \lrp{\mathbb{E}_P\left[\sqrt{1-\boldlambda^T(X-\boldmu)}\right]}^2 - 1 \right],  && \left( f(u) = (u-1)^2 \right).
        \end{align}
    \end{corollary}
    In both instances considered in~\Cref{corollary:f-divergence}, we use~\Cref{theorem:f-divergence} to derive the dual expression and then employ the ``scaling trick'' as in the proof of~\Cref{theorem:finite-support} to obtain the final forms stated above. The details are in~\Cref{proof:corollary-f-divergence}. 

\section{The General Limiting Argument}
\label{sec:general-argument-approximation-constraints}
The extension from finitely supported distributions to arbitrary distributions on $\calX = [0,1]^K$ in~\Cref{theorem:continuous-klinf} and~\Cref{theorem:f-divergence} strongly rely on the exact constraint satisfaction for discretized problems, which was ensured due to the properties of the mean-preserving discretization channel~(\Cref{def:mean-preserving-channel}). In general, as we go beyond mean constraints, the earlier arguments are no longer applicable, and instead we have to also account for approximate constraint satisfaction for discretized problems. We present the details of this argument in an abstract setting with a general divergence measure and constraint function, and then apply it to obtain a dual of $\KLinf$ with integral constraints in~\Cref{subsec:other-constraints}. 

As before, we work with the compact domain $\calX = [0,1]^K$ and consider an arbitrary divergence measure $D:\calP(\calX) \times \calP(\calX) \to [0, \infty]$. For some continuous constraint function $g:\calX \to \R^J$ for $J \geq 1$, and a closed and convex subset $\calC \subset g(\calX) \subset \R^J$, our goal is to obtain a dual representation of the following: 
\begin{align}
    I(P, g, \calC) = \inf\{D(P, Q): Q \in \calP(\calX), \; \mathbb{E}_Q[g(X)] \in \calC \}. \label{eq:general-divergence-to-C}
\end{align}
In order to state and prove the abstract limiting argument, we present a set of assumptions on the various components starting with the divergence measure $D$.
\begin{assumption}[Properties of Divergence Measure]
    \label{assump:general-divergence} We assume that the divergence measure $D:\calP(\calX) \times \calP(\calX) \to [0,\infty]$ satisfies the following two conditions: 
    \begin{align}
        &\text{If } P_k \Longrightarrow P, \; Q_k \Longrightarrow Q, \qtext{then} \liminf_{k \to \infty} D(P_k, Q_k) \geq D(P, Q). && \text{(Weak \lsc)} \\
        &\text{If } \calK \text{ is a Markov kernel}, \quad \text{then} \; D(P\calK, Q\calK) \leq D(P, Q). && \text{(DPI)} 
    \end{align}
\end{assumption}
Next, we formally state the conditions on $(g, \calC)$ that characterize the general constraints in~\eqref{eq:general-divergence-to-C}. 
\begin{assumption}[Continuity of Constraint]
    \label{assump:general-constraint} The constraint function $g:\calX \to \R^J$ for $J \geq 1$ is (uniformly) continuous with a modulus of continuity 
    \begin{align}
        \omega_g(\Delta) = \sup_{\xvec, \xvec' \in \calX: \|\xvec-\xvec'\|_\infty \leq \Delta} \, \|g(\xvec) - g(\xvec')\|_\infty, \qtext{with} \lim_{\Delta \downarrow 0} \omega_g(\Delta) = 0.  
    \end{align}
    The constraint set $\calC$ is a closed and convex subset of $\R^J$. 
\end{assumption}

In the case of $\KLinf$ in the previous section, we worked with a particular class of ``mean-preserving channels'' introduced in~\Cref{def:mean-preserving-channel}. However, that particular constraint-preserving property breaks down when we move to the integral constraints represented by $(g,\calC)$. Instead, we work with arbitrary discretization channels that approximately preserve the constraints as we discuss below. 
\begin{assumption}[Discretization-Channels]
    \label{assump:general-discretization} Let $\{\Delta_k: k \geq 1\}$ denote a positive sequence converging to $0$, and for each $k \geq 1$, let $V_k$ denote a $\Delta_k$-cover of the domain $\calX$ in terms of $\|\cdot\|_\infty$-norm. Let $\calK_k$ denote a discretization-channel, such that for any $\xvec \in \calX$, the distribution $\calK_k(\cdot \mid \xvec)$ is supported on $V_k \cap B_\infty(\xvec, \Delta_k)$. 
    For any $\calX$ valued random variable $X \sim P$, we then use $X_k \sim P_k$ to denote the output after passing $X$ through the channel $\calK_k$. 
    By construction, we have $\|X_k - X\|_\infty \leq \Delta_k$ almost surely, which means that 
    \begin{align}
        \|g(X_k) - g(X)\|_\infty \stackrel{a.s.}{\leq} \sup_{\xvec, \xvec': \|\xvec-\xvec'\|_\infty \leq \Delta} \|g(\xvec) - g(\xvec')\|_\infty = w_g(\Delta_k) \eqcolon \eta_k \downarrow 0, \qtext{as} k \to \infty. 
    \end{align}
\end{assumption}
This assumption says that, unlike our discussion of $\KLinf$ in~\Cref{sec:Klinf-mean-constrained}, a feasible $Q$ in the definition of~$I(P, g, \calC)$ in~\eqref{eq:general-divergence-to-C} may not necessarily remain feasible for the same constraint set after passing through a discretization channel $\calK_k$. Instead, we must enlarge the constraint set $\calC$ appropriately; that is, we define 
\begin{align}
    I_k \equiv I_k(P, g, \calC_k) = \inf \{D(P_k, Q): Q \in \calP(V_k), \; \mathbb{E}_Q[g(X)] \in \calC_k\}, \qtext{with} \calC_k = \calC + B_\infty(\boldsymbol{0}, \eta_k). \label{eq:general-divergence-to-C-k}
\end{align}
Thus, $I_k$ is the relaxed finite-support approximation of the original problem, with support restricted to $V_k$ and the constraint set enlarged to $\calC_k$. 
Now, let $(\R^{\dvec}, \|\cdot\|_2)$ denote the ambient space in which all our dual variables live. For any $k \geq 1$, let $\Theta_k \subset \R^{\dvec}$ denote the dual domain associated with~\eqref{eq:general-divergence-to-C-k}, and $H_k: \Theta_k \times \calP(V_k) \to \R \cup \{-\infty\}$ denote the dual objective, such that 
\begin{align}
        I_k =   \sup_{\theta \in \Theta_k}  H_k(\theta, P_k) = \sup_{\theta \in \Theta_k} \lrp{\mathbb{E}_{P_k}[\psi_k(X, \theta)] + b_k(\theta) }, 
\end{align}
for some measurable $\psi_k:\calX \times \Theta_k \to \R$ and $b_k: \Theta_k \to \R \cup \{-\infty\}$. In practice, $\psi_k$ is the $\xvec$-dependent part of the objective while $b_k$ collects the remaining $\theta$-dependent terms. 
In order for our argument to work, we need certain regularity conditions on the dual domains and objective functions, that we state next. 
\begin{assumption}[Dual Objective Functions]
    \label{assump:general-dual-function} 
    For each $k \geq 1$, we assume that the domain $\Theta_k$ is nonempty, convex, and compact, with a point $\theta_0 \in \mathring{\Theta}_k$ with $H_k(\theta_0, P_k) > -\infty$ for all $k \geq 1$. 
    For any $t \in (0, 1)$, let $\Theta_k^{(t)}$ denote the ``retraction'' of the domain $\Theta_k$; that is, $\Theta_k^{(t)} = t \theta_0 + (1-t) \Theta_k$, let $L_t < \infty$ denote a positive constant, and $\omega_t: [0, \infty) \to [0, \infty)$ denote a non-decreasing function. With these terms, suppose the following statements are true: 
    \begin{align}
        &|\psi_k(\xvec, \theta) - \psi_k(\xvec', \theta)| \leq L_t \|\xvec - \xvec'\|_\infty, \quad \forall \xvec, \xvec', \; \forall \theta \in \Theta_k^{(t)}, \; \forall k \geq 1. && (\text{Uniform Lipschitz in $\xvec$}) \\
        &|H_k(\theta, P_k) - H_k(\theta', P_k)| \leq \omega_t(\|\theta-\theta'\|_2), \quad \forall \theta, \theta' \in \Theta_k^{(t)}. && (\text{Uniform Continuity in $\theta$}) \\
        & \max \bigg\{ \sup_{k} \sup_{\theta \in \Theta_k^{(t)}} |b_k(\theta)|, \, \sup_{k} \sup_{\xvec, \theta \in \calX \times \Theta_k^{(t)}} |\psi_k(\xvec, \theta)| \bigg\}< \infty. && (\text{Uniform Boundedness}) 
    \end{align}
    Additionally, we assume throughout that $H_k$ is concave over its domain $\Theta_k$, for all $k \geq 1$.  
\end{assumption}

The sets $\Theta_k^{(t)}$ play the role of the ``interior-dual-domain'' $\epsdualdomain$ that we introduced while proving~\Cref{theorem:continuous-klinf} for relative entropy, and restrict the analysis to the points away from the boundary where the dual objective may diverge. The three conditions above then formalize the requirement that the family of dual objectives are uniformly ``well-behaved'' for all $k \geq 1$.

The final assumption, that is new for this particular result, arises from the fact that we are now working with approximate constraints in the intermediate problems. 
\begin{assumption}[Dual Limits]
    \label{assump:general-dual-limits}
    Suppose that there exists a nonempty, convex, and compact $\Theta \subset \R^{\dvec}$ with $\theta_0 \in \mathring{\Theta}$, such that  for some vanishing positive sequence $\{s_k: k \geq 1\}$, we have 
    \begin{align}
        d_H(\Theta_k, \Theta) \coloneqq \max \big\{ \sup_{\theta \in \Theta_k} \|\theta - \Theta\|_2, \, \sup_{\theta \in \Theta}  \|\theta - \Theta_k\|_2 \big\} = s_k \downarrow 0. 
    \end{align}
    As before,  for any $t\in (0,1)$, we can define the ``retraction'' $\Theta^{(t)} = t \theta_0 + (1 - t) \Theta$, which we also assume is compact. 
    Next, let $\Pi_{\Theta_k}$%
    denote the (Euclidean) projection from $\R^{\dvec}$ to $\Theta_k$, and introduce the following ``identification maps'', $\tau_{k,t}$,
    \begin{align}
        &\tau_{k,t}: \Theta^{(t)} \to \Theta_k^{(t)}, \quad \text{such that } \tau_{k,t}(\theta) = t \theta_0 + (1-t) \Pi_{\Theta_k}\lrp{\frac{\theta - t \theta_0}{1-t} }, \quad \forall \theta \in \Theta^{(t)}. \label{eq:tau-k-t-def} %
    \end{align}
     Finally, we assume that for every $t \in (0,1)$, there exists a countable dense $\calD_t \subset \Theta^{(t)}$, such that  
     \begin{align}
         \text{for every } \theta \in \calD_t, \quad \lim_{k \to \infty} {F_{k,t}(\theta)}  \quad \text{exists, and is finite}, \qtext{where} F_{k,t}(\theta) \coloneqq H_k(\tau_{k,t}(\theta), P_k). 
     \end{align}
\end{assumption}
The conditions ensure that the sequence of discretized dual problems asymptotically approach a well-defined limit. In particular, the Hausdorff convergence $d_H(\Theta_k, \Theta) \to 0$ says that the dual  domains are `stable' and converge to a fixed compact limit $\Theta$. For each $t \in (0,1)$, the identification map $\tau_{k,t}$ maps elements of the retracted limiting domain $\Theta^{(t)}$ to the closest point in $\Theta_k^{(t)}$, allowing for an analysis of the varying dual objectives over a common domain, while avoiding boundary effects. Finally, the existence of the dense subset $\calD_t$ is sufficient for the existence of a unique limiting dual objective function.

With all these assumptions available, we can now state the main result of this section.

\begin{theorem}
    \label{theorem:general-approximate-constraint} For $\calX = [0,1]^K$, fix a distribution $P \in \calP(\calX)$, and a divergence $D: \calP(\calX) \times \calP(\calX) \to [0, \infty]$ satisfying~\Cref{assump:general-divergence}. For $\{\Delta_k\}_{k \geq 1} \downarrow 0$, and with $\{V_k\}_{k \geq 1}$ denoting $\Delta_k$-covers of $\calX$ under $\|\cdot\|_\infty$, let $\{\calK_k\}_{k \geq 1}$ denote a sequence of discretization channels as in~\Cref{assump:general-discretization}, producing $P_k = \calK_k P$ supported on $V_k$. For the continuous constraint function $g:\calX \to \R^J$, let $\calC_k = \calC + B_\infty(\boldsymbol{0}, \eta_k)$ denote the enlarged constraint set with $\eta_k = \omega_g(\Delta_k)$ from~\Cref{assump:general-constraint}. For any $k \geq 1$, assume that the following dual representation is valid.  
    \begin{align}
        I_k = \sup_{\theta \in \Theta_k} H_k(\theta, P_k)  = \sup_{\theta \in \Theta_k} \lrp{ \mathbb{E}_{P_k}[\psi_k(X, \theta)] + b_k(\theta)}, 
    \end{align}
    and the dual domains converge in Hausdorff metric to a limiting set $\Theta$, and the  technical conditions stated in~\Cref{assump:general-dual-function} and~\Cref{assump:general-dual-limits} hold. Then, we have the following conclusions: 
    \begin{enumerate}
        \item  $\lim_{k \to \infty} I_k = I(P, g, \calC)$. 
        \item For every $t \in (0, 1)$, there exists a unique continuous function $H^{(t)}(\cdot, P): \Theta^{(t)} \to \R$, such that 
        \begin{align}
            \sup_{\theta \in \Theta^{(t)}} \left\lvert F_{k,t}(\theta) - H^{(t)}(\theta, P) \right\rvert \stackrel{k \to \infty}{\longrightarrow} 0, \qtext{where recall that} F_{k,t}(\theta) = H_k(\tau_{k,t}(\theta), P_k).  
        \end{align}
        \item We get the following limiting dual representation 
        \begin{align}
            I(P, g, \calC) = \lim_{k \to \infty} \sup_{\theta \in \Theta_k} H_k(\theta, P_k) = \sup_{t \in (0,1)} \sup_{\theta \in \Theta^{(t)}} H^{(t)}(\theta, P). 
        \end{align}
        \item If, in addition, there exists an $H(\cdot, P): \Theta \to \R$, such that for all $t \in (0, 1)$, we have $H(\theta, P) = H^{(t)}(\theta, P)$ for all $\theta \in \Theta^{(t)}$, then we have 
        \begin{align}
            I(P, g, \calC) = \sup_{\theta \in \Theta} H(\theta, P). 
        \end{align}
    \end{enumerate}
\end{theorem}
The proof of this result proceeds broadly in two stages analogous to the proof of~\Cref{theorem:continuous-klinf}. We first show that the sequence of primal discretized problems converge to the original value; that is, $I_k \to I(P, g, \calC)$, by observing that for any feasible $Q$ for the original problem is transformed by the discretization channel to $Q_k$ that is feasible with the enlarged $\calC_k$. The desired convergence then follows from the DPI and weak lower semicontinuity conditions of~\Cref{assump:general-divergence}. We then turn to the more delicate part of the argument, and show the convergence of the dual of these discretized problems. One crucial complication, in comparison to~\Cref{theorem:continuous-klinf}, is that both the dual objectives $H_k$ and dual domains $\Theta_k$ may vary with $k$. To address this, for any fixed $t \in (0,1)$, we first retract or shrink the domain $\Theta_k$ to $\Theta_k^{(t)}$, pull back the corresponding objectives to a fixed limiting set $\Theta^{(t)}$ via the identification map $\tau_{k,t}$, and study the behavior of the resulting functions $F_{k,t}(\theta) = H_k(\tau_{k,t}(\theta), P_k)$. The regularity assumptions of~\Cref{assump:general-dual-function} imply that for each $t \in (0,1)$, this family of functions $\{F_{k,t}: k \geq 1\}$ is uniformly bounded and equicontinuous, which yields a continuous limit $H^{(t)}(\cdot, P)$ on $\Theta^{(t)}$. Finally, we use the concavity of each $H_k$ to pass from the retracted to the full domain, allowing us to identify $\lim_{k} \sup_{\theta \in \Theta_k} H_k(\theta, P_k)$ with $\sup_{t \in (0,1)} \sup_{\theta \in \Theta^{(t)}} H^{(t)}(\theta, P)$.
We present the details in~\Cref{proof:general-limiting-argument}.

\subsection{$\KLinf$ with General Constraints}
\label{subsec:other-constraints}
In this section, we discuss the extension of the constrained minimum relative entropy term beyond mean constraints. 
In particular, we will show how~\Cref{theorem:general-approximate-constraint} can be used to obtain a dual of  
\begin{align}
    \KLinf(P, g, \calC) \coloneqq \inf_{Q \in \calQ_{g,\calC}} \KL(P, Q), \qtext{where} \calQ_{g, \calC} = \left\{ Q: \mathbb{E}_Q[g(X)]  \in \calC \right\},  \label{eq:general-constraint-1}
\end{align}
where $g: \calX \to \R^J$ for some $J \geq 1$ is a continuous constraint function,  $\calC \subset \R^J$ is a closed and convex constraint set. 
The simplest generalization beyond mean constraint is when $g$ is an affine function of the form $g(\xvec) = A \xvec +\mathbf{b}$. Due to the linearity of expectation, for any distribution $Q$ on $(\calX, \calB)$ satisfying this constraint, for any $k \geq 1$ and $\calK_k$ denoting the mean-preserving channel from~\Cref{def:mean-preserving-channel}, we have $\mathbb{E}_{Q}[A X] = \mathbb{E}_{Q_k}[A X_k]$, where $Q_k = \calK_k Q$. Thus, the existing argument in~\Cref{theorem:continuous-klinf} generalizes naturally to such constraints. We build  upon this in~\Cref{appendix:box-constraints} and identify a broader class of nonlinear functions for which we can still build such channels to obtain the dual for $\KLinf(P, g, \calC)$. 

However, this constraint-preserving approach breaks down for more general constraints in which the coordinates of $X$ are coupled; for example, if $g(\xvec) = (\|\xvec\|, \|\xvec\|^2)$. This motivates returning to the abstract framework of~\Cref{sec:general-argument-approximation-constraints} where we allow approximate constraint satisfaction in the discretized problems, rather than seeking to construct an exact constraint-preserving discretization channel. The limiting argument developed in~\Cref{theorem:general-approximate-constraint} then allows us to establish the required duality. 
\begin{theorem}
    \label{theorem:general-constraint-KL} Consider a  Lipschitz continuous $g:\calX \to \R^J$, and let $\calC \subset \R^J$ denote a compact and convex constraint set. For a sequence $(\Delta_k)_{k \geq 1}$ converging to $0$, let $(V_k)_{k \geq 1}$ denote the $\Delta_k$-covers of $\calX$ and let $(\calK_k)_{k \geq 1}$ denote arbitrary discretization channels such that the distribution $\calK_k(\cdot \mid \xvec)$ is supported on $V_k \cap B_\infty(\xvec, \Delta_k)$ for all $k \geq 1$, $\xvec \in \calX$. Assume that $\operatorname{Conv}(g(\calX))$ has a nonempty interior in $\R^J$ and $\operatorname{int}\left(\operatorname{Conv}(g(\calX))\right) \cap \calC \not =  \emptyset$. Then, we have 
    \begin{align}
        \KLinf(P, g, \calC)  = \sup_{\substack{(\boldlambda, \gamma) \in \R^J \times \R \\ \gamma - \langle \boldlambda, g(\xvec)\rangle \geq 0, \, \forall \xvec \in \calX}} 
        \lrp{
        \mathbb{E}_P[\log \lrp{\gamma - \langle \boldlambda, g(X) \rangle }] + 1 - \gamma + \inf_{\cvec \in \calC} \langle \cvec, \boldlambda \rangle
        }. 
    \end{align}
\end{theorem}
The proof of this result is presented in~Appendix~\ref{proof:general-constraint-KL}, and it proceeds by verifying the five assumptions used by the general limiting argument of~\Cref{theorem:general-approximate-constraint}. The first three assumptions~(Assumptions~\ref{assump:general-divergence}--\ref{assump:general-discretization}) are easy to verify, so the nontrivial part of the proof involves identifying the dual objective and domain for the discretized problem, and verifying that they satisfy~\Cref{assump:general-dual-function} and~\Cref{assump:general-dual-limits}.  

\begin{remark}
    \label{remark:lipschitz-constaint} The statement of~\Cref{theorem:general-constraint-KL} requires the constraint function $g$ to be Lipschitz continuous. This additional condition  is  placed  mainly to simplify the verification of~\Cref{assump:general-dual-function} in the proof. Since $\calX =[0,1]^K$ is compact, we know from Stone-Weierstrass theorem, that every continuous $g$ can be uniformly approximated by a polynomials, which are Lipschitz continuous on compact sets. This suggests that the dual representation of~\Cref{theorem:general-constraint-KL} should extend to continuous $g$, but we do not pursue that extra approximation argument here. 
\end{remark}
\begin{remark}
    \label{remark:nonempty-interior} The extra assumption that $\operatorname{int}\operatorname{Conv}(g(\calX)) \cap \calC$ is nonempty in the statement of~\Cref{theorem:general-constraint-KL} serves two purposes: first, it immediately justifies the strong duality for the discretized problems, and second, it aids the verification of~\Cref{assump:general-dual-function} and~\Cref{assump:general-dual-limits} by allowing us to identify compact subsets of the dual domains that contain the optimizers.  As with the Lipschitz assumption on $g$, this nonempty interior condition is also not strictly necessary and can potentially be weakened at the cost of additional technical steps that we do not pursue here.  
\end{remark}

\section{Statistical Applications}
\label{sec:statistical-applications}

In this section, we show how the dual representation of $\KLinf$ can be used to construct optimal sequential anytime-valid inference  procedures for observations supported on $\calX = [0,1]^K$ for $K \geq 1$. 

\paragraph{Sequential Testing:} Suppose $\{X_n: n \geq 1\}$ denotes a sequence of \iid\ observations drawn from a distribution $P_X \in \calP(\calX)$, where $\calX = [0,1]^K$. Let $\boldmu_X$ denote the unknown mean $\mathbb{E}_{X \sim P_X}[X]$. For some fixed $\boldmu \in \mathring{\calX}$, we are interested in deciding between 
\begin{align}
    H_0: \boldmu_X = \boldmu, \qtext{versus} H_1: \boldmu_X \neq \boldmu. \label{eq:mean-testing-problem-def}
\end{align}
We want to design a power one, level-$\alpha$ sequential test, which is a specification of stopping time $\tau_\alpha$ satisfying 
\begin{align}
     \mathbb{P}_{H_0}\lrp{ \tau_\alpha < \infty }\leq \alpha, \qtext{and} 
    \mathbb{P}_{H_1}\lrp{\tau_\alpha < \infty } = 1. 
\end{align}
With $\calT(\boldmu, \alpha)$ denoting the class of all level-$\alpha$, power-one tests for the null stated in~\eqref{eq:mean-testing-problem-def}, we are interested in characterizing the term $\inf_{\tau_\alpha \in \calT(\boldmu, \alpha)} \; {\mathbb{E}_{P}[\tau_\alpha]}$ in the limit as $\alpha \to 0$, for every $P\in \mathcal P(\mathcal X)$ with mean not equal to $\boldmu$. First, we describe the construction of a level-$\alpha$ test based on the empirical $\KLinf$ term and then analyze its performance in~\Cref{theorem:sequential-testing}. 
\begin{definition}
    \label{def:sequential-test} Fix an $\alpha \in (0,1)$, and $\boldmu \in \mathring{\calX}$. Let $\{X_n: n \geq 1\} \overset{\iid}{\sim} P_X$. For any $n \geq 2$, let $\widehat{P}_n$ denote the empirical distribution based on $(X_1, \ldots, X_{n-1})$, and define the stopping time 
    \begin{align}
        \tau_{\alpha} \equiv \tau_\alpha(\boldmu) = \inf \{n \geq 2: (n-1)\KLinf(\widehat{P}_n, \boldmu) \geq  K \log(n) + \log(1/\alpha) + 1\}. 
    \end{align}
    Due to the dual representation of $\KLinf$, we have 
        $\KLinf(\widehat{P}_n, \boldmu) = \sup_{\boldlambda \in \dualdomain} \frac{1}{n-1} \sum_{i=1}^{n-1} \log (1 - \boldlambda^T(X_i - \boldmu))$, 
    which can be computed using off-the-shelf convex solvers in a computationally feasible manner. 
\end{definition}
We now state the main result of this section. 
\begin{proposition}
    \label{theorem:sequential-testing} 
    For any $P_X \in \calP(\calX)$ with mean $\boldmu_X \neq \boldmu \in (\mathring\calX)$, we have the following: 
    \begin{align}
      \lim_{\alpha \downarrow 0}\;  \inf_{\tau_\alpha \in \calT(\boldmu, \alpha)} \;  \frac{\mathbb{E}_{P_X}[\tau_\alpha]}{\log(1/\alpha)} = \frac{1}{\KLinf(P_X, \boldmu)}. 
    \end{align}
    Furthermore, the test introduced in~\Cref{def:sequential-test} lies in $\calT(\boldmu, \alpha)$, and achieves the infimum in the above expression for every $\alpha \in (0,1)$. 
\end{proposition}
The proof of this result is in~\Cref{proof:sequential-testing}. The lower bound follows from an application of the data processing inequality for randomly stopped processes. For the $\alpha$-correctness of the test from Definition~\ref{def:sequential-test}, we first show using the dual form of $\KLinf$ that, after adjusting it with a small cost, it is dominated by the log of a mixture martingale \citep[Lemma F.1]{agrawal2021optimal}, which exceeds the threshold $\log(1/\alpha)$ with probability at most $\alpha$. The proof for the sample complexity upper bound expresses the stopping event in terms of the hitting time of a simple random walk with positive drift, and also explicitly relies on the dual formulation of $\KLinf$.

\paragraph{Sequential Estimation via Confidence Sequences~(CSs):} Now suppose instead of testing whether the mean $\boldmu_X$ is equal to a given value $\boldmu$, we wish to construct a confidence sequence~(CS) for the unknown mean $\boldmu_X$. Formally, for an $\alpha \in (0,1)$, a level-$(1-\alpha)$ confidence sequence for $\mu_X$ is a sequence of subsets $\{C_n \subset \calX: n \geq 1\}$, such that each $C_n$  is $\sigma(X_1, \ldots, X_n)$ measurable, $C_0 = \calX$, and 
\begin{align}
    \mathbb{P}(\exists n \geq 1: \mu_X \not \in C_n) \leq \alpha \quad \iff \quad \mathbb{P}(\forall n \geq 1: \mu_X \in C_n) \geq 1-\alpha. 
\end{align}
Since we have designed an optimal sequential test for any $\boldmu$ in~\Cref{def:sequential-test},  we can construct a level-$(1-\alpha)$ confidence sequence~(CS) for the unknown mean via the usual inversion as  
\begin{align}
    &C_0 = \calX, \quad C_n= \left\{ \boldmu \in \mathring{\calX}: \; (n-1)\KLinf(\widehat{P}_n, \boldmu) < \log\lrp{\frac{n^K}{\alpha}} \right\} \quad \text{for } n \geq 2. \label{eq:conf-seq-def}
\end{align}
As an immediate consequence of the level-$\alpha$ property of the test introduced in~\Cref{def:sequential-test}, we have the following result. 
\begin{corollary}
    \label{corollary:conf-seq} Suppose the true mean $\boldmu_X$ lies in the interior $\mathring{\calX}$. Then, the CS defined in~\eqref{eq:conf-seq-def} satisfies the level-$(1-\alpha)$ property:~ 
        $\mathbb{P}(\exists n \geq 1: \boldmu_X \not \in C_n) \leq \alpha$. 
\end{corollary}
\begin{proof}
This result  follows naturally from~\Cref{theorem:sequential-testing}. In particular, we have 
\begin{align}
    \mathbb{P}\left( \exists n \geq 1: \boldmu_X \not \in C_n \right)  & =     \mathbb{P}\left( \exists n \geq 1: (n-1) \KLinf(\widehat{P}_n, \boldmu_X) \geq \log(n^K/\alpha) \right). 
\end{align}
The second term is simply the probability that the test introduced in~\Cref{def:sequential-test} stops at a finite time under the null. As we have proved in~\Cref{theorem:sequential-testing}, this is upper bounded by $\alpha$, which establishes the required level-$(1-\alpha)$ property of the CS constructed in~\eqref{eq:conf-seq-def}. 
\end{proof}
Note that computing each confidence set $C_n$ involves finding the level-set of a convex function $f(\boldmu) = \KLinf(\widehat{P}_n, \boldmu)$, which in general can be computationally infeasible. We leave the thorough exploration of methods for approximating these sets in a computationally feasible manner for future work. 

\paragraph{Sequential Change Detection:} The final application we consider is that of detecting changes in the mean vector of a stream of observations. Formally, let $P, Q$ denote two elements in $\calP(\calX)$, such that $\mathbb{E}_P[X] = \boldmu_0$ and $\mathbb{E}_Q[X] = \boldmu_1 \neq \boldmu_0$. Throughout, we assume that $\boldmu_0$ is known to us, and at some unknown time $T \in \mathbb{N} \cup \{\infty\}$, there is an abrupt change in distribution from $P$ to $Q$~(both $\boldmu_1$ and $Q$ are unknown). The goal then is to design a stopping time $N_\alpha$, such that  the average run length~(ARL)
$\mathbb{E}_{\infty, P}[N_\alpha] \geq \frac{1}{\alpha}$, where $\mathbb{E}_{\infty, P}[\cdot]$ denotes the expectation when there is no change and the observations are drawn according to $P$. Further, given this constraint, we also want to minimize the  ``detection delay'' when a change occurs, formally given by  
\begin{align}
J_L(N_\alpha, P, Q) = \sup_{T \in \mathbb{N}} \mathrm{ess}\sup \mathbb{E}_{T, P, Q}[(N_\alpha-T)^+ \mid X_1, \ldots, X_T]. 
\end{align}
In these expressions we use $\mathbb{E}_{T, P, Q}[\cdot]$ to denote the expectation when there is a change in distribution from $P$ to $Q$ at time $T \in \mathbb{N}$. 
We now use design a change detection scheme using  the test in~\Cref{def:sequential-test},  following the construction of~\citet{lorden1971procedures}. 
\begin{definition}
    \label{def:change-detection} Given a stream of $\calX = [0,1]^K$-valued observations $\{X_n: n \geq 1\}$,  a mean vector $\boldmu_0 \in \mathring{\calX}$, and a parameter $\alpha \in (0,1)$, let $\tau_\alpha^{(k)}$ denote the level-$\alpha$ test for the null $H_0: \mathbb{E}[X] = \boldmu_0$ based on the observations $\{X_n: n \geq k\}$ for $k \in \mathbb{N}$. Then, define the change-detection procedure $N_\alpha = \inf\limits_{k \geq 1} \tau_\alpha^{(k)}$, where 
    \begin{align}
        \tau^{(k)}_{\alpha} &\coloneqq \inf \left\{n \geq k:  (n-k+1) \KLinf\left( \frac{1}{n-k+1} \sum_{i=k}^n \delta_{X_i}, \, \boldmu_0\right) \geq \log \left( \frac{(n-k+1)^K}{\alpha} \right) \right\} 
    \end{align}
    denotes the level-$\alpha$ test from the previous discussion for the null $H_0: \mathbb{E}[X]=\boldmu_0$, based on the observations $\{X_n: n \geq k\}$ for $k \in \mathbb{N}$. 
\end{definition}
In words, the procedure defined above initiates a new power-one test, $\tau_k^{(\alpha)}$, in every round, and stops and declares a detection as soon as any of the initiated tests rejects the null. 
We now present the main result characterizing the behavior of this change detection scheme.  

\begin{proposition}
\label{prop:change-detection}    
Consider the change detection problem described above, where for some unknown $T \in \mathbb{N} \cup \{\infty\}$, we have $\{X_n: 1 \leq n \leq T\} \overset{\iid}{\sim} P \in \calP_0 \coloneqq \{P' \in \calP(\calX): \mathbb{E}_{P'}[X] = \boldmu_0\}$ for a known  $\boldmu_0$, and $\{X_n: n > T\} \overset{\iid}{\sim} Q \in \calP_1 \coloneqq \{P' \in \calP(\calX): \mathbb{E}_{P'}[X] \neq \boldmu_0\}$ with an unknown mean $\boldmu_1 \neq \boldmu_0$. 
Then, the change-detection procedure described in~Definition~\ref{def:change-detection} satisfies the following: 
\begin{align}
    \inf_{P \in \calP_0} \mathbb{E}_{\infty, P}[N_\alpha] \geq \frac{1}{\alpha}, \qtext{and} \limsup_{\alpha \downarrow 0}  \sup_{P\in \calP_0}\frac{J_L(N_\alpha, P, Q)}{\log(1/\alpha)} \leq \frac{1}{\KLinf(Q, \boldmu_0)}. 
\end{align}
Furthermore, let $\calC(\boldmu_0, \alpha)$ denote the class of all change detection procedures $N'_\alpha$ with $\inf_{P \in \calP_0} \mathbb{E}_{\infty, P}[N'_\alpha] \geq 1/\alpha$. Then, we have the following: 
\begin{align}
    \liminf_{\alpha \downarrow 0} \frac{L(Q, \boldmu_0, \alpha)}{\log(1/\alpha)} \geq \frac{1}{\KLinf(Q, \boldmu_0)}, \qtext{where} L(Q, \boldmu_0, \alpha) = \inf_{N'_\alpha \in \calC(\boldmu_0, \alpha)} \; \sup_{P\in \calP_0} J_L(N'_\alpha, P, Q). 
\end{align}
\end{proposition}
The proof of this result is in Appendix~\ref{proof:change-detection}. The implementation of the optimum change detection scheme $N_\alpha$ again relies on the ability to efficiently compute $\KLinf(\widehat{P}_n, \boldmu_0)$, which is facilitated by our dual result.

\section{Conclusion and Future Work}
\label{sec:conclusion}
    $\KLinf$ and more general constrained minimum divergences have emerged as objects of fundamental importance in  many statistical decision problems, including sequential inference and bandit learning.
    The primal definition of these terms often infinite-dimensional optimization problems over probability measures, making them ill-suited for direct computations and algorithm design. 
    The main contribution of this paper is an elementary two-stage recipe for deriving tractable dual representations of such constrained minimum-divergence problems for distributions supported on compact domains $\calX \subset \R^K$~(for concreteness, we worked with $\calX = [0,1]^K$). 
    
    In the first step, we derive the dual for the case of distributions with finite support in $\calX$, which can be achieved via classical convex duality theory for finite-dimensional problems. In the second step, we show how to pass from a discretized dual to the dual for arbitrary distributions by developing a continuity argument along a sequence of increasingly fine discretizations. The proof relies on standard information-theoretic arguments, appealing to the data-processing inequality and the weak lower-semicontinuity of relative entropy. This observation allows us to extend the same pipeline to a more general class of $f$-divergences, and then to general continuous constraint functionals. 
    Finally, we  illustrated how these dual representations allow for constructing optimal statistical procedures in sequential testing, estimation, and change-detection problems.

    Since our results rely strongly on the compactness of the domain $\calX = [0,1]^K$, a natural direction for future work is to extend our two-stage approach to the case of distributions supported on more general domains, such as $\mathbb{R}^K$. Another interesting direction is exploring the role of our dual minimum divergence terms in areas such as distributionally robust optimization, and in constructing nonasymptotically valid empirical likelihood confidence sets.

\bibliography{ref_arxiv}
\bibliographystyle{abbrvnat}

\appendix

\section{Additional Background}
In this section, we recall some key results that are used to prove our main results. 
\begin{fact}[Caratheodary's Theorem]
    \label{fact:caratheodary}
    Let $\calX \subset \mathbb{R}^K$ be a convex hull of the a set $S$. Then, any element of $\calX$ can be written as a convex combination of at most $K+1$ elements of $S$. 
\end{fact}

\begin{fact}[Prohorov's Theorem]
    \label{fact:prohorov} Let $(\calX, d)$  denote a complete separable metric space, and let $\calP(\calX)$ denote the collection of Borel probability  measures on $\calX$. We say that a family $\Pi \subset \calP(\calX)$ is tight, if for every $\epsilon >0$, there exists a compact set $K_\epsilon \subset \calX$, such that $\inf_{P \in \Pi} P(K_\epsilon) \geq 1-\epsilon$. Then, the following are equivalent: 
    \begin{enumerate}
        \item $\Pi \subset \calP(\calX)$ is tight. 
        \item $\Pi$ is relatively compact in the topology of weak convergence; that is, every sequence in $\Pi$ has a weakly convergent subsequence. 
    \end{enumerate}
\end{fact}

\begin{fact}[Arzel\'a-Ascoli]
    \label{fact:arzela-ascoli}
    Let $(\calX, d)$ be a compact metric space, and let $\{f_n: n \geq 1\}$ denote a sequence of real-valued continuous functions on $\calX$; that is, $f_n \in C(\calX, \mathbb{R})$. Suppose the following conditions hold: 
    \begin{itemize}
        \item Uniform Boundedness: $\sup_{n \geq 1} \, \sup_{ \xvec \in \calX} |f_n(\xvec)| < \infty$
        \item Equicontinuity: For every $\epsilon>0$, there exists a $\delta>0$, such that for all $n \geq 1$ and $\xvec, \xvec' \in \calX$, 
        \begin{align}
            d(\xvec, \xvec') < \delta \quad \implies \quad 
            |f_n(\xvec) - f_n(\xvec')| < \epsilon. 
        \end{align}
    \end{itemize}
    Then, there exists a subsequence $\{f_{n_j}: j \geq 1\}$ and a function $f \in C(\calX, \R)$ such that 
    \begin{align}
        \sup_{\xvec \in \calX} |f_{n_j}(\xvec) - f(\xvec)| \; \overset{n \to \infty}{\longrightarrow} \; 0. 
    \end{align}
    In other words, uniformly bounded and equicontinuous collection of functions contain a convergent subsequence. 
\end{fact}

\section{Deferred Proofs from~\Cref{sec:Klinf-mean-constrained}}

\subsection{Proof of~\Cref{theorem:finite-support}}
\label{proof:finite-support}

\paragraph{Step 1: Restriction to a finite-dimensional domain.} Let us first introduce some notation. Let $\calP_{\boldmu} \subset \calP(\calX)$ denote the collection of distributions on $\calX$ with mean $\boldmu$; that is, $\calP_{\boldmu} = \{Q \in \calP(\calX): \mathbb{E}_Q[X] = \boldmu\}$. 
\begin{lemma}
    \label{lemma:discrete-proof-1} For every $Q \in \calP_{\boldmu}$, there exists another $R \in \calP_{\boldmu}$, such that $\KL(P, R) \leq \KL(P, Q)$, and $\supp(R) \subset \finitesupport \cup \{0, 1\}^K$, where $\finitesupport = \supp(P)$. In other words, we can restrict our attention to a finite-dimensional subspace of the domain of the primal problem without losing optimality. 
\end{lemma}

\begin{proof}
    The idea is simple: given a $Q$ with mean vector $\boldmu$, we can decompose it into $Q_a + Q_s$, where $Q_a$ is the ``absolutely continuous'' part of $Q$ w.r.t. $P$, and $Q_s$ is the singular part of $Q$ supported on $\calX \setminus \finitesupport$. Suppose $\finitesupport = \{x_1, \ldots, x_m\}$ and let us denote $Q_a \equiv (q_1, \ldots, q_m)$, $P \equiv (p_1, \ldots, p_m)$, and $\rho = 1- \sum_{i=1}^m q_i = \mathbb{E}_{Q_s}[1] \geq 0$. If $\rho=0$, then $Q$ is supported on $\finitesupport$ and we can set $R=Q$, so for the rest of the proof, we consider the case of $\rho>0$. Now, observe that the objective function in this case only depends on $Q_a \equiv (q_1, \ldots, q_m)$: 
    \begin{align}
        \KL(P, Q) = \sum_{i=1}^m p_i \log (p_i/q_i). 
    \end{align}
    Since $Q$ is feasible, we have 
    \begin{align}
        \mathbb{E}_Q[X] = \sum_{i=1}^m q_i x_i + \mathbb{E}_{Q_s}[X] = \boldmu. 
    \end{align}
    Let us denote by $A$ the term $(\boldmu - \sum_{i=1}^m q_i x_i)/\rho$, and we can verify that $A \in \calX$,  since it is the mean value associated with a distribution supported on $\calX$. Hence, by Caratheodary's theorem, $A$ can be represented as a convex combination of $K+1$ corner points of the cube. Formally, let $\{v_1, \ldots, v_{2^K}\}$ denote the corner points of $\{0, 1\}^K$ of $\calX = [0,1]^K$. Then, there exists a probability distribution $R_s \in \calP(\{0,1\}^K)$ such that 
    \begin{align}
        A = \frac{1}{\rho}\lrp{\boldmu - \sum_{i=1}^m q_i x_i} = \sum_{j=1}^{2^K} r_j v_j, \qtext{with} R_s \equiv (r_1, \ldots, r_{2^K}). 
    \end{align}
    Finally, to conclude the proof, we define $R = Q_a + \rho R_s$, and observe that $\KL(P, R) = \KL(P, Q)$ if $\finitesupport \cap \{0, 1\}^K = \emptyset$, and otherwise $\KL(P, R) < \KL(P, Q)$.  In other words, for every feasible $Q \in \calP_{\boldmu}$, there exists another feasible $R$ supported on $\finitesupport \cup\{0,1\}^K$ whose objective value is no larger than $\KL(P, Q)$. This completes the proof. 
\end{proof}

\paragraph{Step 2: Verification of Strong Duality.}  
Assume that the vector $\boldmu$ lies in the interior of the domain $\calX$, denoted by $\in \mathring{\calX} = (0,1)^K$. Then, we need to show the existence of a point in the relative interior of the feasible set for the primal problem. More specifically, we need to show the existence of a point $Q$, such that 
\begin{align}
    Q(\xvec)>0, \; \forall \xvec \in S \coloneqq \finitesupport \cup \{0, 1\}^K, \quad \mathbb{E}_Q[X] = \boldmu, \qtext{and} \mathbb{E}_Q[1] = 1. 
\end{align}
To construct such a $Q$, we first consider $\mathbf{V} = (V_1, V_2, \ldots, V_K)$ with $V_i \sim \mathrm{Bernoulli}(\mu_i)$ for $i \in [K]$, and $V_i \perp V_j$ for all $i \neq j$. Let $R_{\boldmu}$ denote the distribution of $\bf{V}$, and observe that $\mathbb{E}[\bf{V}] = \boldmu$. Furthermore,  as $\boldmu \in \mathring{\calX}$, each $\mu_i \in (0, 1)$ and hence $R_{\boldmu}$ is supported on $\{0,1\}^K$. Now, for every $\xvec_i \in \finitesupport \setminus \{0,1\}^K$, by Caratheodary's theorem, there exists a distribution $R_i$, supported on at most $K+1$ points in $\{0,1\}^K$, such that $\mathbb{E}_{R_i}[X] = \xvec_i$. Then, we define 
\begin{align}
    Q = R_{\boldmu} + \sum_{\xvec_i \in \finitesupport \setminus \{0,1\}^K} \epsilon_i (\delta_{\xvec_i} - R_i). 
\end{align}
Here, $\epsilon_i>0$ are constants that are small enough to ensure that each coordinate of $Q$ is strictly positive at all $\xvec \in \finitesupport \cup \{0,1\}^K$. A sufficient condition is if for all $i$, we have $\epsilon_i < (\min_{\xvec \in \{0,1\}^K} R_{\boldmu}(\{\xvec\}))/|\finitesupport|$.

\paragraph{Step 3: Obtaining the dual via KKT conditions.} 
As before, we use  $S$ to denote the support set $\finitesupport \cup \{0, 1\}^K$, and let $\calM^+ \equiv \calM^+(S)$ denote the collection of non-negative measures supported  on $S$. Then,  
\begin{align}
    \KLinf(P,\boldmu) = \quad &\min\limits_{Q\in \mathcal M^+} \quad \KL(P,Q) \\ 
    & \text{s.t.} \qquad \boldmu - \Exp{Q}{X} = 0  &&:\quad \boldlambda \in \R^K \\
    & \text{~~~~} \qquad \Exp{Q}{\boldsymbol{1}} - 1 = 0 \ &&:\quad \gamma \in \R
\end{align}
where the mean equality constraint represents a coordinate-wise equality, $\boldsymbol{1}:\calX \to \{1\}$ denotes the function that maps every $\xvec\in \calX$ to $1$,  and the variables $\boldlambda, \gamma$ denote the dual variables that will be associated with the constraints in the sequel. For any $j \in [K]$, we will denote the $j^{th}$ coordinates of $\boldlambda$, $\boldmu$, and $\xvec \in \calX= [0,1]^K$ with $\lambda_j$, $\mu_j$ and $x_j$ respectively. Then, for any unsigned measure $Q\in \calM^+(S)$, the Lagrangian equals
\begin{align}
    L({\bm\lambda}, \gamma, P, Q) &\coloneqq  \sum\limits_{\xvec \in\finitesupport} P(\xvec)\log\frac{P(\xvec)}{Q(\xvec)} + \sum\limits_{j=1}^K \lambda_j\lrp{ \mu_j - \sum\limits_{\xvec} Q(\xvec) x_j  } + \gamma \lrp{ \sum_{\xvec} Q(\xvec) - 1 }\\
    &= \sum\limits_{\xvec \in\finitesupport} P(\xvec )\log\frac{P(\xvec)}{Q(\xvec)} + \sum\limits_{j=1}^K \lambda_j\lrp{ \mu_j - \sum\limits_{\xvec \in\finitesupport} Q(\xvec ) x_j  } + \gamma \lrp{ \sum_{\xvec \in\finitesupport} Q(\xvec) - 1 } \\
    &\qquad\qquad\qquad\qquad +  \sum\limits_{\xvec \notin\finitesupport} Q(\xvec) \lrp{ \gamma  - \sum\limits_{j=1}^K \lambda_j  x_j  },
\end{align}
with the understanding that $0\log 0 = 0$, and $\log \infty = \infty$.
The Lagrangian dual for $\KLinf$ then becomes 
\[ \mathcal D(P, \boldmu) = \max\limits_{\substack{\boldlambda \in \R^K \\ \gamma \in \R}} ~~ \min\limits_{\substack{Q\in \mathcal M^+}} ~~
 L({\bm \lambda}, \gamma, P,Q).\]
We now observe that we can restrict our attention to a smaller class of dual variables. 
\begin{lemma}
    \label{lemma:finite-dual-support}
    Let $\Lambda_P$ denote the set of dual variables for which the inner minimization in the display above is non-trivial; that is, $\Lambda_P \coloneqq \{(\boldlambda, \gamma): \min_{Q \in \calM^+} L(\boldlambda, \gamma, P, Q) > -\infty\}$. Then, we have 
\begin{align}
    \Lambda_P \subset \left\{ (\boldlambda, \gamma) \in R^K \times \R \;:\;  \min_{\xvec \in \calX}   \gamma -  \boldlambda^T \xvec \geq 0,  \;\text{and}\; \gamma - \boldlambda^T \xvec >0, \;\forall \xvec \in \finitesupport  \right\}. \label{eq:dual-feasible-1}
\end{align}
\end{lemma}

\begin{proof}
We prove this result by separately considering two cases. The first is of $\xvec \in \{0,1\}^K \setminus \finitesupport$. For such an $\xvec$, note that $L(\boldlambda, \gamma, P, \cdot)$ is linear in $Q(\xvec)$ for  fixed $(\boldlambda, \gamma, P)$.  In particular, let $Q(\xvec) = q_\xvec$, then then part of $L(\boldlambda, \gamma, P, Q)$ depending on $q_\xvec$ is $q_\xvec(\gamma - \boldlambda^T \xvec)$. If $c_\xvec = \gamma - \boldlambda^T \xvec < 0$, then we can make the inner minimization $\min_{Q} L(\boldlambda, \gamma, P, Q) = -\infty$, by taking a sequence $Q'$s assigning increasingly larger mass at $\xvec$. 

For all $x\in \finitesupport$, it turns out that the part of $L(\boldlambda, \gamma, P, Q)$ that depends on $q_\xvec = Q(\xvec)$ is $\phi(q_\xvec) = p_\xvec \log \lrp{ \tfrac{p_\xvec}{q_\xvec} } + (\gamma - \boldlambda^T \xvec) q_\xvec$. So, if $c_\xvec = \gamma - \boldlambda^T \xvec \leq 0$, then we can again make $L(\boldlambda, \gamma, P, \cdot)$ go to $-\infty$ by choosing a sequence of $Q'$s assigning increasing values of $q_\xvec$. So for such $\xvec \in \finitesupport$, we need the stronger condition that $\gamma - \boldlambda^T \xvec > 0$. 

Together, these two conditions imply that 
\begin{align}
    \gamma - \boldlambda^T \xvec \geq 0, \qtext{for all} \xvec \in \{0,1\}^K. 
\end{align}
Since every $\xvec \in \calX$ can be written as a convex combination of the corner points $\{0,1\}^K$, the above condition also implies that $\gamma - \boldlambda^T \xvec \geq 0$ for all $\xvec \in \calX$. This concludes the proof. 
\end{proof}

Now restricting our attention to $(\boldlambda, \gamma) \in \Lambda_P$, we can now characterize the value of the optimal $Q^* \equiv Q^*(\boldlambda, \gamma)$ for any such pair of dual variables. 
\begin{lemma}
    \label{lemma:finite-optimal-Q} For any $(\boldlambda, \gamma) \in \Lambda_P$ defined in~\eqref{eq:dual-feasible-1}, there exists an optimal $Q^* \equiv Q^*(\boldlambda, \gamma)$ satisfying the following conditions. 
    \begin{itemize}
        \item For any $\xvec \in \finitesupport$, we have 
        \begin{align}
            q^*_\xvec = \frac{p_\xvec}{c_\xvec}, \qtext{where} q^*_\xvec = Q^*(\xvec), \qtext{and} c_\xvec = \gamma - \boldlambda^T \xvec > 0. 
        \end{align}
        \item For any $\xvec \in \{0,1\}^K \setminus \finitesupport$, we must have $q^*_\xvec (\gamma - \boldlambda^T \xvec) = 0$. 
    \end{itemize}
\end{lemma}

\begin{proof}
For the first condition, recall that for $\xvec \in \finitesupport$, the objective depends on $q^*_\xvec$ only through the function $\phi(q^*_\xvec) = p_\xvec \log(p_\xvec/q^*_\xvec) + c_\xvec q^*_\xvec$, with $c_\xvec > 0$~(by the definition of   $\Lambda_P$). For $q^*_\xvec > 0$, we can check that 
\begin{align}
    \phi''(q^*_\xvec) =  \frac{p_\xvec}{(q^*_\xvec)^2} > 0  \quad \implies \quad 
    \phi(\cdot) \text{ is convex on } (0, \infty). 
\end{align}
Thus, the minimizer is attained at $q^*_\xvec$ such that $\phi'(q^*_\xvec) = -p_\xvec/q^*_\xvec +  c_\xvec = 0$; or $q^*_\xvec = p_\xvec/c_\xvec$ as claimed. 

For the second statement, we have already proved that for all $\xvec \in \calX$, we must have $c_\xvec = \gamma -  \boldlambda^T \xvec \geq 0$.  If $\xvec \not \in \finitesupport$, then the only dependence of $L(\boldlambda, \gamma, P, Q)$ on $q_\xvec = Q(\xvec)$ is through the linear function $c_\xvec q_\xvec$. If $c_\xvec = 0$, then $c_\xvec q_\xvec$ is trivially equal to $0$ for any feasible $Q$, while if $c_\xvec>0$, then the optimizing $Q^*$ must assign zero mass at all such $\xvec$. In either case, we must have $c_\xvec q^*_\xvec = 0$ as required. 
\end{proof}

On substituting the optimal form of $Q^* \equiv Q^*(\boldlambda, \gamma)$ for $\boldlambda$ and $\gamma$ in $\Lambda_P$, the dual becomes 
\begin{align}
    \mathcal D(P,\boldmu) = \max\limits_{(\boldlambda, \gamma) \in \Lambda_P} ~ \sum\limits_{\xvec \in\finitesupport} P(\xvec) \log\lrp{ \gamma - \boldlambda^T \xvec } + 1 - \lrp{ \gamma  - \boldlambda^T \boldmu}. \label{eq:KLinf-derivation-no-scaling}
\end{align}
Before proceeding, we note a simple fact about the term $\gamma - \boldlambda^T \boldmu$ in~\eqref{eq:KLinf-derivation-no-scaling}. 
\begin{lemma}
    \label{lemma:finite-positive-constant} For any $(\boldlambda, \gamma) \in \Lambda_P$, and $\boldmu \in \mathring{\calX}$, we must have $\gamma - \boldlambda^T \boldmu >0$. 
\end{lemma}
\begin{proof}
    Let us enumerate the elements of $\{0,1\}^K$ by $\{\vvec_1, \ldots, \vvec_{2^K}\}$, and observe that since $\boldmu \in \mathring{\calX}$, there exist strictly positive $\{w_i: i \in [2^K]\}$ with $\sum_{i=1}^{2^K} w_i = 1$, such that $\boldmu = \sum_{i=1}^{2^K} w_i \vvec_i$. This implies that 
    \begin{align}
        \gamma - \boldlambda^T \boldmu = \gamma - \boldlambda^T \lrp{\sum_{i=1}^{2^K}  w_i \vvec_i}  = \sum_{i=1}^{2^K} w_i \lrp{\gamma - \boldlambda^T \vvec_i} \geq 0,  
    \end{align}
    where the inequality is due to the fact that $(\boldlambda, \gamma) \in \Lambda_P$ defined in~\eqref{eq:dual-feasible-1}. Now, consider the possibility that $\gamma - \boldlambda^T \boldmu = 0$. Since each $w_i > 0$, it must mean that $\gamma - \boldlambda^T \vvec_i = 0$ for all $i \in [2^K]$. Since every $\xvec \in \finitesupport$ can be written as a convex combination of $\{\vvec_i: i \in [2^K]\}$, this leads to the conclusion that $\gamma - \boldlambda^T \xvec = 0$ for all $\xvec \in \finitesupport$. This contradicts the result of~Lemma~\ref{lemma:finite-dual-support}. Hence, we must have $\gamma - \boldlambda^T \boldmu > 0$. 
\end{proof}

To conclude the proof, observe that the constraints in~\eqref{eq:KLinf-derivation-no-scaling} are scale-invariant. In other words if $(\boldlambda, \gamma) \in \Lambda_P$, then so does $(c \boldlambda, c \gamma)$ for all $c >0$, and dual optimal does not change with scaling the dual variables.
\begin{align}
    \mathcal D(P,\boldmu) = \max\limits_{\substack{(\boldlambda, \gamma) \in \Lambda_P \\ c>0}} ~ \sum\limits_{\xvec \in\finitesupport} P(\xvec) \log\lrp{ \gamma -  \boldlambda^T\xvec } + 1 + \log c - c\lrp{ \gamma  -  \boldlambda^T \boldmu}. 
\end{align}
Optimizing over the variable $c$, with a fixed $(\boldlambda, \gamma)$, we see that 
\[ c^* \equiv c^*(\boldlambda, \gamma, \boldmu, P) = \frac{1}{\gamma - \boldlambda^T \boldmu } > 0. \]
The strict inequality follows form Lemma~\Cref{lemma:finite-positive-constant}.  On substituting $c^*$, we get
\begin{align}
    &\mathcal D(P,\boldmu) = \max\limits_{(\boldlambda, \gamma) \in \Lambda_P} ~ \sum\limits_{\xvec\in\finitesupport} P(\xvec) \log\lrp{\frac{ \gamma - \boldlambda^T \xvec }{\gamma - \boldlambda^T \boldmu }},\\
    \text{which implies}\quad 
    &\mathcal D(P,\boldmu) = \max\limits_{(\boldlambda, \gamma) \in \Lambda_P} ~ \sum\limits_{\xvec \in\finitesupport} P(\xvec) \log\lrp{1 - \frac{ \boldlambda^T (\xvec-\boldmu) }{\gamma - \boldlambda^T \boldmu }}. 
\end{align}

Renaming the dual variable as $\boldlambda \leftarrow \frac{\boldlambda}{\gamma - \boldlambda^T \boldmu}$, we get the required final expression: 
\begin{align}
\mathcal D(P,\boldmu) 
&= \max\limits_{\boldlambda \in \calL_{\boldmu}} ~~\sum\limits_{\xvec \in\finitesupport} P(\xvec) \log\lrp{ 1- \boldlambda^T(\xvec-\boldmu) }, \quad\text{where} \\
\calL_{\boldmu} &\coloneqq \{\boldlambda \in \R^K: \min_{\xvec \in \calX} 1- \boldlambda^T(\xvec-\boldmu) \geq 0 \}. 
\end{align}
This completes the proof. \hfill \qedsymbol

\subsection{Proof of~\Cref{theorem:continuous-klinf}}
\label{proof:continuous-klinf}
We will break down the proof of this result into several steps. The starting point is to establish a continuity property of $\KLinf$, and this is achieved by establishing the lower and upper semicontinuity separately. The \lsc property is inherited from the \lsc of relative entropy, while for showing the usc\ property, we rely on the data processing inequality.    
\begin{lemma}
    \label{lemma:interchange-1}
    For any two distributions $P$ and $Q$  in $\calP(\calX)$, let $P_k = P \calK_k $ and $Q_k = Q \calK_k$ denote their push-forward measures associated with the mean-preserving discretization channel $\calK_k$ from~\Cref{def:mean-preserving-channel}. Let $\boldmu \in \mathring{\calX}$ denote the mean vector associated with $Q$ lying in the interior of the domain $\calX$, and define 
    \begin{align}
        I_k \equiv I_k(P_k, \boldmu) = \inf_{Q: \mathbb{E}_Q[X]=\boldmu} \KL(P_k, Q_k). 
    \end{align}
    Then, we have $\lim_{k \to \infty} I_k = \KLinf(P, \boldmu)$. 
\end{lemma}

\emph{Proof of~\Cref{lemma:interchange-1}.} 
First we observe that the set of distributions $\calQ_1 = \{Q_k = Q \calK_k: \, Q \in \calP(\calX), \; \mathbb{E}_Q[X]=\boldmu\}$ coincides with $\calQ_2 = \{Q'_k \in \calP(V_k): \mathbb{E}_{Q'_k}[X]=\boldmu\}$. The inclusion $\calQ_1 \subset \calQ_2$ holds by definition, as each $Q_k \in \calP(V_k)$ with mean $\boldmu$ due to the mean-preserving property of $\calK_k$, while the other direction $\calQ_2 \subset \calQ_1$ holds trivially since $\calP(V_k) \subset \calP(\calX)$.  Hence, we have 
\begin{align}
    I_k = \inf_{Q \in \calQ_1} \KL(P \calK_k, Q \calK_k) = \inf_{Q_k \in \calQ_2} \KL(P \calK_k, Q_k). 
\end{align}
\paragraph{$\boldsymbol{\limsup_{k \to \infty} I_k \leq I}$.} This direction of the proof relies on the data processing inequality~(DPI) for relative entropy. In particular, consider any $P \in \calP(\calX)$, and $Q \in \calP(\calX)$ with $\mathbb{E}_Q[X]=\boldmu$. Then, for any $k \geq 1$, we have 
\begin{align}
\KL(P, Q) & \geq \KL( P \calK_k, Q \calK_k) && (\text{DPI for relative entropy}) \\
\implies \inf_{Q: \mathbb{E}_{Q}[X] = \boldmu} \KL(P, Q) & \geq \inf_{Q:\mathbb{E}_{Q}[X]=\boldmu} \KL(P \calK_k, Q \calK_k) = I_k. && (\text{Definition of $I_k$}) 
\end{align}
This leads us to the bound:  
\begin{align}
    I = \KLinf(P, \boldmu) = \inf_{Q: \mathbb{E}_Q[X]=\boldmu} \KL(P, Q) \geq  I_k. 
\end{align}
Since this inequality is true for all $k \geq 1$, it is preserved on taking a limsup over all $k$: that is, $I \geq \limsup_{k} I_k$. 

\paragraph{$\boldsymbol{\liminf_{k \to \infty} I_k \geq I}$.} To show the other direction, we start by considering a subsequence of $\mathbb{N}$ that achieves the $\liminf$: that is, let $\{k^\ell: \ell \geq 1\}$ denote a subsequence such that 
\begin{align}
    \liminf_{k \to \infty} I_k = \lim_{\ell \to \infty} I_{k^\ell}. 
\end{align}
Choose an arbitrary $\epsilon > 0$, and for every $\ell \geq 1$,  {select an $\epsilon$-suboptimal distribution $\Qepsilon_{k^\ell} \gg P_k$, such that} 
\begin{align}
 \KL(P_{k^\ell}, \Qepsilon_{k^\ell}) \geq    I_{k^\ell} = \inf_{Q: \mathbb{E}_{Q}[X]=\boldmu} \KL(P_{k^\ell},  Q \calK_{k^\ell} )  \geq \KL(P_{k^\ell}, \Qepsilon_{k^\ell}) - \epsilon. \label{eq:interchange-1-proof-2}
\end{align}
All elements of the resulting collection of probability measures on $(\calX, \calB)$, denoted by $\{\Qepsilon_{k^\ell}:\ell \geq 1\}$,  are supported on the compact set $\calX = [0,1]^K$. Hence, this collection of probability measures is trivially ``tight''. An application of  Prokhorov's theorem~\citep[Theorem 5.1]{billingsley1999convergence} then implies that there exists a distribution $\Qepsilon_\infty$ on $(\calX, \calB)$, and a subsequence $(k_j^\ell)_{j \geq 1}$ such that 
\begin{align}
    \Qepsilon_{k^{\ell}_j} \Longrightarrow \Qepsilon_\infty, \qtext{as} j \to \infty, 
\end{align}
where ``$\Longrightarrow$'' denotes weak convergence. 
{Since weak convergence of probability measures implies convergence of the bounded functions, and as the projection map $\xvec \mapsto x_j$ is bounded on $\calX = [0,1]^K$} for all $j \in [K]$, we observe that 
\begin{align}
\mathbb{E}_{\Qepsilon_\infty}[X] = \lim_{j \to \infty} \mathbb{E}_{\Qepsilon_{k_j^{\ell}}}[X] = \lim_{j \to \infty} \boldmu = \boldmu. \label{eq:interchange-1-proof-3}
\end{align}
Thus, $\Qepsilon_\infty$ is a valid probability measure for the primal definition of  $\KLinf(P, \boldmu)$. Furthermore, we can also verify that $ P\calK_k = P_k \Longrightarrow P$, which also implies that the subsequence $P_{k_j^\ell} \stackrel{j \to \infty}{\Longrightarrow} P$. Due to the joint lower semicontinuity of relative entropy, we then obtain 
\begin{align}
    \liminf_{j \to \infty} \KL(P_{k_j^\ell}, \Qepsilon_{k_j^\ell}) \geq \KL\lrp{\liminf_{j \to \infty} P_{k_j^\ell}, \liminf_{j \to \infty} \Qepsilon_{k_j^\ell}} = \KL(P, \Qepsilon_\infty). 
\end{align}

This fact coupled with the inequality~\eqref{eq:interchange-1-proof-2} leads to 
\begin{align}
    \liminf_{k \to \infty} I_k = \lim_{\ell \to \infty}  I_{k^\ell} = \lim_{j \to \infty}I_{k_j^\ell} \geq \liminf_{j \to \infty} \KL(P_{k_j^\ell}, \Qepsilon_{k^\ell_j}) - \epsilon  \geq \KL(P, \Qepsilon_\infty) - \epsilon. 
\end{align}

Finally, using the fact that $Q_\infty^{(\epsilon)}$ is a feasible distribution for the $\KLinf(P, \boldmu)$ definition according to~\eqref{eq:interchange-1-proof-3}, we have the following: 
\begin{align}
    \left( \liminf_{k \to \infty} I_k \right) + \epsilon \geq \KL(P, \Qepsilon_\infty) \geq \inf_{Q: \mathbb{E}_Q[X]=\boldmu} \KL(P, Q) = \KLinf(P, \mu) = I. 
\end{align}
Since $\epsilon>0$ is arbitrary, this implies the required $\liminf_{k \to \infty} I_k \geq I$. 
\hfill \qedsymbol
The previous lemma, combined with the dual representation of $\KLinf$ for finitely supported distributions~(Proposition~\ref{theorem:finite-support}) implies the following: 
\begin{align}
    \KLinf(P, \boldmu) \stackrel{\text{(Lemma\ref{lemma:interchange-1})}}{=} \lim_{k \to \infty}  \KLinf(P\calK_k, \boldmu) \stackrel{\text{(Prop.~\ref{theorem:finite-support})}}{=} \lim_{k \to \infty} \sup_{\boldlambda \in \dualdomain} H_k(\boldlambda). 
\end{align}
To complete the proof, we need to justify the interchange of $\lim$ and $\sup$ in the second equality above. 
\begin{lemma}
    \label{lemma:interchange-2}
    For any $\calX$-valued random variable $X \sim P$ and $k \geq 1$,  let $X_k \sim P_k$ denote the output obtained by  passing $X$ through the mean-preserving discretization channel $\calK_k$. For a fixed $\boldmu = (\mu_1,\ldots, \mu_K) \in \mathring{\calX} = (0, 1)^K$, introduce the terms 
    \begin{align}
        \calL_{\boldmu} \coloneqq \{ \boldlambda \in \R^K: \sup_{\xvec \in \calX} \boldlambda^T(\xvec - \boldmu) \leq 1 \}, & \qtext{and} h(x, \boldlambda)  \coloneqq \log(1 - \boldlambda^T(x - \boldmu)),  \\
    \quad 
        H_k(\boldlambda) \coloneqq \mathbb{E}_{P_k}[h(X_k, \boldlambda)], &\qtext{and} H(\boldlambda) \coloneqq \mathbb{E}_P[h(X, \boldlambda)]. 
    \end{align}
    Then, we have $\lim_{k \to \infty} \sup_{\boldlambda \in \calL_{\boldmu}} H_k(\boldlambda) = \sup_{\boldlambda\in \calL_{\boldmu}} H(\boldlambda)$. 
\end{lemma}

\emph{Proof of~\Cref{lemma:interchange-2}.}
The sets $\calX = [0,1]^K$, and $\calL_{\boldmu}$ with $\boldmu \in \mathring{\calX}$ are closed and bounded subsets of $\R^K$, and thus are compact. For some $\epsilon>0$, introduce the ``$\epsilon$-interior'' of $\dualdomain$, defined as
\begin{align}
    \epsdualdomain \coloneqq \{ \boldlambda \in \dualdomain: \sup_{\xvec\in \calX} \boldlambda^T(\xvec-\boldmu) \leq 1-\epsilon\}.  \label{eq:epsilon-interior}
\end{align}
Then $(\xvec, \boldlambda) \mapsto h(\xvec, \boldlambda)$ is uniformly continuous and bounded on the domain $\calX \times \epsdualdomain$, and satisfies 
\begin{align}
    \log \epsilon \leq h(\xvec, \boldlambda) \leq \log\lrp{1 + 4(1-\epsilon) \frac{\max\{ \|\boldmu\|_\infty, \, \|\boldsymbol{1}-\boldmu\|_\infty \} }{\min\{ \|\boldmu\|_\infty, \, \|\boldsymbol{1}-\boldmu\|_\infty \}} }, \qtext{for all} (\xvec, \boldlambda) \in \calX \times \epsdualdomain. 
\end{align}
The lower bound follows from the defining property of $\epsdualdomain$. To see why the upper bound must hold, note that $\boldmu \in \mathring{\calX}$, and hence there exists a $r_\infty < \min \{\|\boldmu\|_\infty, \|\boldsymbol{1}-\boldmu\|_\infty\}$ and $R_\infty> \max  \{\|\boldmu\|_\infty, \|\boldsymbol{1}-\boldmu\|_\infty\}$, such that $B_\infty(\boldmu, r_\infty) \subset \calX \subset B_\infty(\boldmu, R_\infty)$, where we use $B_\infty(\uvec, r)$ to denote $\{\yvec \in \mathbb{R}^K: \|\yvec-\uvec\|_\infty \leq r\}$.
A valid choice is $R_\infty = 2  \max  \{\|\boldmu\|_\infty, \|\boldsymbol{1}-\boldmu\|_\infty\}$ and $r_\infty = (1/2) \min   \{\|\boldmu\|_\infty, \|\boldsymbol{1}-\boldmu\|_\infty\}$. 
Then, for any $\boldlambda\in \epsdualdomain$, we must have 
\begin{align}
    r_\infty \|\boldlambda\|_1 = \sup_{\xvec \in B_\infty(\boldmu, r_\infty)} \boldlambda^T(\xvec - \boldmu) \leq \sup_{\xvec \in \calX} \boldlambda^T(\xvec - \boldmu)  \leq 1-\epsilon. \label{eq:eps-dual-bound-1}
\end{align}
Furthermore, since $\calX \subset B_\infty(\boldmu, R_\infty)$, we also have 
\begin{align}
    \inf_{\xvec \in B_\infty(\boldmu, R_\infty)} \boldlambda^T(\xvec - \boldmu) = - \|\boldlambda\|_1 R_\infty \geq - \frac{(1-\epsilon) R_\infty}{r_\infty} = -4(1-\epsilon)  \frac{\max  \{\|\boldmu\|_\infty, \|\boldsymbol{1}-\boldmu\|_\infty\}}{ \min  \{\|\boldmu\|_\infty, \|\boldsymbol{1}-\boldmu\|_\infty\}  } \eqcolon -C.  
\end{align}
Moreover, for any $\boldlambda \in \epsdualdomain$, the  map $\xvec \mapsto h(\xvec, \boldlambda)$ is {also Lipschitz continuous} on $\calX \times \epsdualdomain$, since $z \mapsto \log z$ is $1/\epsilon$-Lipschitz on $[\epsilon, \infty)$, and  
\begin{align}
    |h(\xvec, \boldlambda) - h(\xvec', \boldlambda)| \leq \frac{|\boldlambda^T(\xvec - \xvec')|}{\epsilon}  \leq \frac{\|\boldlambda\|_1}{\epsilon} \|\xvec-\xvec'\|_\infty \leq \frac{(1-\epsilon)}{\epsilon r_\infty} \|\xvec-\xvec'\|_\infty
    \coloneqq L_\epsilon \|\xvec-\xvec'\|_\infty, 
\end{align}
where the third inequality follows from the bound on $\|\boldlambda\|_1$ obtained in\eqref{eq:eps-dual-bound-1}. 
We state our result in terms of the $\|\cdot\|_\infty$ norm for concreteness, but the same argument should work with any other norm. Note that the assumption that $\boldmu$ lies in the interior of $\calX$ is crucial for the boundedness and Lipschitz continuity. 

\paragraph{Proof of $\boldsymbol{\limsup_{k \to \infty} \sup_{\boldlambda\in \dualdomain} H_k(\boldlambda) \leq \sup_{\boldlambda \in \dualdomain} H(\boldlambda)}$.} Using the properties of $\calK_k$, we know that $\mathbb{E}[X_k|X] = X$ for every realization of $X$, and hence 
\begin{align}
    H_k(\boldlambda) = \mathbb{E}[h(X_k, \boldlambda)] = \mathbb{E}[\mathbb{E}[h(X_k, \boldlambda) \mid X]] \stackrel{(i)}{\leq} \mathbb{E}[ h(\mathbb{E}[X_k|X], \boldlambda)] = \mathbb{E}[h(X, \boldlambda)] = H(\boldlambda),  \label{eq:interchange-2-proof-1}
\end{align}
where $(i)$ is due to the conditional Jensen inequality and the concavity of the $\log$ function.  Hence, on maximizing over $\dualdomain$, we get 
\begin{align}
    \sup_{\boldlambda \in \dualdomain} H_k(\boldlambda) \leq  \sup_{\boldlambda \in \dualdomain} H(\boldlambda) \quad \implies \quad   
    \limsup_{k \to \infty} \sup_{\boldlambda \in \dualdomain} H_k(\boldlambda) \leq  \sup_{\boldlambda \in \dualdomain} H(\boldlambda). \label{eq:interchange-2-proof-2}
\end{align}

\paragraph{Proof of $\boldsymbol{\liminf_{k \to \infty} \sup_{\boldlambda\in \dualdomain} H_k(\boldlambda) \geq \sup_{\boldlambda \in \dualdomain} H(\boldlambda)}$.} By the construction of $X_k$, we know that $\|X - X_k\|_\infty \leq \Delta_k = 2^{-k}$. This fact, combined with the Lipschitz continuity of~$h(\cdot, \boldlambda)$ on $\calX$ implies that for any $\delta>0$, there exists a finite $k_{\epsilon, \delta}$, such that for all $k \geq k_{\epsilon, \delta}$ and $\boldlambda\in \epsdualdomain$, we have 
\begin{align}
    |H(\boldlambda) - H_k(\boldlambda)| \leq \delta \quad \implies \quad \sup_{\boldlambda\in \epsdualdomain} H(\boldlambda) - \delta \leq \sup_{\boldlambda \in \epsdualdomain} H_k(\boldlambda) \leq \sup_{\boldlambda\in \epsdualdomain} H(\boldlambda) + \delta. 
\end{align}
Since $\delta>0$ is arbitrary, we can conclude that 
\begin{align}
    \lim_{k \to \infty} \sup_{\boldlambda \in \epsdualdomain} H_k(\boldlambda) = \sup_{\boldlambda \in \epsdualdomain} H(\boldlambda). \label{eq:interchange-2-proof-4}
\end{align}
It remains to relate the supremum of $H(\boldlambda)$ over the restricted domain $\epsdualdomain$ to the supremum over the entire $\dualdomain$. In fact, we will show that 
\begin{align}
    \sup_{\boldlambda \in \dualdomain} H(\boldlambda) = \sup_{\epsilon \in (0, 1)} \sup_{\boldlambda \in \epsdualdomain} H(\boldlambda).  \label{eq:interchange-2-proof-3}
\end{align}
\sloppy Before proving this statement, we show how it suffices to reach our required conclusion $\liminf_{k \to \infty} \sup_{\boldlambda \in \dualdomain} H_k(\boldlambda) \geq \sup_{\boldlambda \in \dualdomain} H(\boldlambda)$. To do this,  observe that 
\begin{align}
    \liminf_{k \to \infty} \sup_{\boldlambda \in \calL_{\boldmu}} H_k(\boldlambda) \geq \liminf_{k \to \infty} \sup_{\boldlambda \in \epsdualdomain} H_k(\boldlambda) = \sup_{\boldlambda \in \epsdualdomain} H(\boldlambda), 
\end{align}
where the equality holds due to~\eqref{eq:interchange-2-proof-4}. Taking a supremum over $\epsilon \in (0, 1)$ and on invoking the equality~\eqref{eq:interchange-2-proof-3}, we get the required 
\begin{align}
    \liminf_{k \to \infty} \sup_{\boldlambda \in \dualdomain} H_k(\boldlambda) \geq \sup_{\epsilon \in (0, 1)} \sup_{\boldlambda \in \epsdualdomain} H(\boldlambda) = \sup_{\boldlambda \in \dualdomain} H(\boldlambda). 
\end{align}
This completes the proof of~\Cref{lemma:interchange-2}.

\paragraph{Justification of~\eqref{eq:interchange-2-proof-3}.} We begin by observing that $\epsdualdomain = (1-\epsilon)\dualdomain = \{(1-\epsilon) \boldlambda: \boldlambda \in \dualdomain\}$, which is  a consequence of the linearity of the constraint defining $\dualdomain$ and $\epsdualdomain$.  In particular, if $\boldlambda \in \dualdomain$, then 
\begin{align}
    \sup_{\xvec \in \calX} \boldlambda^T(\xvec - \boldmu) \leq 1 \quad \iff \quad 
    \sup_{\xvec \in \calX} (1-\epsilon)\boldlambda^T(\xvec - \boldmu) \leq 1-\epsilon, 
\end{align}
which implies that $(\boldlambda \in \dualdomain) \; \iff \; \big( (1-\epsilon)\boldlambda \in \epsdualdomain\big)$. 
Now, observe that $H(\boldlambda) = \mathbb{E}_P[\log(1 - \boldlambda^T(X-\boldmu)]$ is concave on the interior of the dual domain $\dualdomain$, since $\boldlambda \mapsto \log(1 + \boldlambda^T(X-\boldmu))$ is concave, which is preserved under expectation. Also  at $\boldsymbol{0} \in \dualdomain$, we have $H(\boldsymbol{0}) = 0$. Now, introduce the terms 
\begin{align}
    s = \sup_{\boldlambda \in \dualdomain} H(\boldlambda), \qtext{and} s_\epsilon = \sup_{\boldlambda \in \epsdualdomain} H(\boldlambda), \qtext{and observe that}  \sup_{\epsilon \in (0,1)} s_\epsilon \leq s.  \label{eq:interchange-2-proof-9}
\end{align}
To show the other direction, fix an arbitrary $\delta >0$, and let $\boldlambda_\delta \in \dualdomain$ be an element such that $H(\boldlambda_\delta) \geq s - \delta = \sup_{\boldlambda \in \dualdomain} H(\boldlambda) - \delta$. For any $\epsilon \in (0, 1)$, we then have 
\begin{align}
    s_\epsilon \geq H\big( (1-\epsilon) \boldlambda_\delta \big) = H\big( (1-\epsilon) \boldlambda_\delta + \epsilon \boldsymbol{0}\big) \geq (1-\epsilon) H(\boldlambda_\delta) + 0 \geq (1-\epsilon) (s - \delta). 
\end{align}
This implies the following chain: 
\begin{align}
    \sup_{\epsilon \in (0,1)} s_\epsilon \geq \liminf_{ \epsilon \to 0} s_\epsilon  \geq \liminf_{\epsilon \to 0} (1-\epsilon) (s-\delta) = s-\delta, \quad \implies \quad 
    \sup_{\epsilon \in (0,1)} s_\epsilon \geq s, \label{eq:interchange-2-proof-10}
\end{align}
since $\delta>0$ was arbitrary. 
Together,~\eqref{eq:interchange-2-proof-9} and~\eqref{eq:interchange-2-proof-10} imply the required equality.  \hfill \qedsymbol

\subsection{$\KLinf$ with box constraints}
\label{appendix:box-constraints}
An important aspect of our proof of~\Cref{theorem:continuous-klinf} was the usage of mean-preserving channels that ensured that the constraints of the discretized problems were exactly preserved. In this section, we observe that there exist a larger class of constraints (beyond the mean) for which the same idea still works. 
These include  constraints that involve monotone coordinate-wise transforms, such as $\mathbb{E}[|X|^j] \preccurlyeq \boldmu$, or bijective transforms from $\calX$ to $\calX$. 
\begin{assumption}
    \label{assump:general-constraint-1}
    The constraint function $g:\calX \to \mathbb{R}^J$, for some $J \in [K]$, is a continuous function that satisfies the following two properties: 
    \begin{itemize}
        \item The range of $g$ is a box; that is, $\calY \coloneqq g(\calX) =  \prod_{j=1}^J [a_j, b_j]$ for $a_j < b_j \in \mathbb{R}$ for all $j \in [J]$. 
        \item There exists a continuous surjective~(i.e., onto) selection function $s:\calY \to \calX$, such that $g(s(\yvec)) = \yvec$ for all $\yvec \in \calY$.  
\end{itemize}
\end{assumption}
The box structure of the range of $g$ allows us to discretize the space $\calY$ using the mean-preserving channel~(Definition~\ref{def:mean-preserving-channel}), while the existence of selection map allows us to pull the resulting discrete points in $\calY$ back to the domain $\calX$. 
A sufficient condition for these conditions to hold is if $g$ is a continuous bijective map from $\calX$ to $\calX$ with a continuous inverse~(that is, $g:\calX \to \calX$ is a homeomorphism). 
Next, suppose that we have a dual formulation of $\KLinf(P, g)$ for the case of $P$ supported on a finite subset of $\calX$ of the form 
\begin{align}
    \KLinf(P, g, \calC) = \inf_{Q \in \calQ_{g, \calC}} \KL(P, Q) = \sup_{\boldlambda \in \calL_g} H_g(P, \boldlambda),  \label{eq:finite-klinf-general-constraint}
\end{align}
for some domain $\calL_g$ and dual objective $H_g$. Then, we can establish an analog of~\Cref{theorem:continuous-klinf} for this more general case as stated next. 

\begin{proposition}
    \label{prop:bijective-constraints} Consider a continuous $g:\calX \to \calY \coloneqq g(\calX)$ satisfying~\Cref{assump:general-constraint-1}, where $\calX = [0,1]^K$ and $\calY \subset \mathbb{R}^J$ for some $J \in [K]$.  Then, for every $k \geq 1$, there exists a finite subset $V_k \subset \calX$, and a constraint-preserving channel $\calK_k:2^{V_k} \times \calX \to [0,1]$, such that $\int g(x) d(Q\calK_k)(x) = \int g(x) dQ(x)$ for all $Q \in \calP(\calX)$. Furthermore, we have the following: 
    \begin{align}
        \KLinf(P, g, \calC) = \lim_{k \to \infty} \KLinf(P \calK_k, g, \calC) = \lim_{k \to \infty} \sup_{\boldlambda \in \calL_g} H_g( P \calK_k, \boldlambda), \label{eq:KLinf-general-constraint}
    \end{align}
    where the dual formulation stated in~\eqref{eq:finite-klinf-general-constraint} is assumed to hold for finitely supported distributions. 
\end{proposition}
\begin{proof}
The key idea behind this result is to use the ``box assumption" to construct a mean-preserving kernel $\calJ_k$ in the feature space $g(\calX)$ with uniform grids, and then pull it back via the inverse map $s$ to obtain the required $(V_k, \calK_k)$. 

In particular, let $Y = g(X)$ be a $\calY$-valued random variable. By the assumption that $\calY = g(\calX) = \prod_{i=1}^K [a_i, b_i]$, we can construct a dyadic grid $\calG_k$ with an associated set of vertices $U_k \subset \calY$. Using these, we can then define a mean-preserving channel $\calJ_k:2^{U_k} \times \calY \to [0,1]$ as described in~\Cref{def:mean-preserving-channel}. Let $Y_k$ denote the discretized version of $Y$ after passing through $\calJ_k$. Then, it follows that $\mathbb{E}[Y_k] = \mathbb{E}[\mathbb{E}[Y_k \mid Y] ] = \mathbb{E}[Y] = \mathbb{E}[g(X)]$. Finally, by the existence of a continuous selection function $s:\calY \to \calX$, we can construct a discrete subset $V_k = \{s(\yvec): \yvec \in U_k\}$, and the channel $\calK_k$ by  pulling back $\calJ_k$ through $s$ as $\calK_k(E \mid \xvec) = \calJ_k(s^{-1}(E) \mid g(\xvec))$.  
Since $s$ is onto and continuous, and $\cup_k U_k$ is dense in $\calY$, it follows that $\cup_k V_k = \cup_k s(U_k)$ is also dense in the original domain $\calX$. 
Thus, we have constructed a sequence of constraint-preserving channels $\calK_k:2^{V_k} \times \calX \to [0,1]$, which discretize $X$ to $X_k$ supported on finite domains $V_k$, such that $\cup_k V_k$ is dense in $\calX$. With these properties, we can repeat all the arguments of~\Cref{theorem:continuous-klinf} with this more general constraint to obtain the result stated in~\eqref{eq:KLinf-general-constraint}. We omit details to avoid repetition. 
\end{proof}

\subsection{Proof of~\Cref{theorem:f-divergence}} 
\label{proof:f-divergence} 

As in the case of $\KLinf$, we prove this result in two steps. First, in~Appendix~\ref{proof:f-divergence-finite-support}, we obtain the dual for finitely supported distributions, and in~Appendix~\ref{proof:f-divergence-extension}, we extend it to general distributions on $\calX = [0,1]^K$ following an argument that parallels the proof of~\Cref{theorem:continuous-klinf}. 

    \subsubsection{$P$ with finite support}
    \label{proof:f-divergence-finite-support} 
    As for relative entropy, we first consider the case of finitely supported $P$ on $\finitesupport = \{\xvec_1, \ldots, \xvec_m\}$ with mean $\boldmu \in \mathring{\calX}$. Any  $Q \in \calP(\calX)$ with finite $D_f(P \parallel Q)$ can be represented by $\{q_i = Q(\{\xvec_i\}): i \in [m]\}$ and $\Qtilde$ which is supported in $\calX \setminus \finitesupport$. Let us denote by $\qtilde$ the mass assigned by $Q$~(equivalently $\Qtilde$) to $\calX \setminus \finitesupport$, and define 
    \begin{align}
        \boldrho = \begin{cases}
            (1/\qtilde) \int \xvec d\Qtilde(\xvec), & \text{ if } \qtilde >0, \\
            \boldsymbol{0}, & \text{ if } \qtilde =0. 
        \end{cases}
    \end{align}
    Since $P$ is supported on $\finitesupport$, it follows that 
    \begin{align}
        D_f(P \parallel Q) = \int f \lrp{\frac{dP}{dQ}} dQ  = \sum_{i=1}^m p_i \ftilde\lrp{\frac{q_i}{p_i}} + \qtilde f(0), \qtext{with} p_i \coloneqq P(\{\xvec_i\}), \; \forall i \in [m].
    \end{align}
    The constraints only depend on $Q$ through $(\qtilde, \boldrho)$; that is, 
    \begin{align}
        &\sum_{i=1}^m q_i + \qtilde= 1, \; \text{for}\; q_i \geq 0, \; \qtilde \geq 0, && (\text{probability constraint}) \\
        &\sum_{i=1}^m q_i \xvec_i + \qtilde \boldrho = \boldmu,  \; \text{for} \; \boldrho \in \calX. && (\text{mean constraint})
    \end{align}
    We first recall that in this problem, strong duality holds by the exact construction used in the proof of~\Cref{theorem:finite-support}. 
    Now, we introduce the dual variables $\boldlambda \in \R^K$ and $\gamma \in \mathbb{R}$, and write the Lagrangian as 
    \begin{align}
        L(\{q_i\}, \qtilde, \boldrho; \boldlambda, \gamma) & = \sum_{i=1}^m p_i \ftilde\lrp{\frac{q_i}{p_i} }  + \qtilde f(0) +  \sum_{i=1}^m (\gamma - \boldlambda^T \xvec_i) q_i + (\gamma - \boldlambda^T \boldrho) \qtilde - (\gamma - \boldlambda^T \boldmu). 
    \end{align}
    For any feasible $(\boldlambda, \gamma)$, we can define the dual objective by minimizing $L$ over $(\{q_i\}, \qtilde, \boldrho)$: 
    \begin{align}
        g(\boldlambda, \gamma) \coloneqq \inf_{\{q_i \geq 0 \}, \qtilde \geq 0, \boldrho \in \calX}  L(\{q_i\}, \qtilde, \boldrho; \boldlambda, \gamma) = \inf_{\{q_i\}} \lrp{  \inf_{\qtilde, \boldrho} L(\{q_i\}, \qtilde, \boldrho; \boldlambda, \gamma)}. 
    \end{align}
    Observe that the only $(\qtilde, \boldrho)$-dependent term in $L$ is $(f(0) + \gamma - \boldlambda^T \rho) \qtilde$, and we claim that this is equal to zero  or $-\infty$. To see that, let us introduce $A(\boldlambda) = \sup_{ \boldrho \in \calX}\,\boldlambda^T \boldrho$, and consider two cases: 
    \begin{itemize}
        \item If $f(0) + \gamma < A(\boldlambda)$, then there exists a $\boldrho \in \calX$, such that $f(0) + \gamma - \boldlambda^T \boldrho < 0$. For such $(\gamma, \boldlambda)$ pairs, the inner infimum can be made arbitrarily small by taking $\qtilde \uparrow \infty$.
        \item If $f(0) + \gamma > A(\boldlambda)$, then for every $\boldrho \in \calX$, the minimum is achieved by taking $\qtilde=0$. 
        \item Hence, if $\gamma = A(\boldlambda)$, then the minimizing $\boldrho \in \argmax_{\yvec \in \calX} \boldlambda^T \yvec$~(the supremum is achieved due to the compactness of the domain $\calX$), and the optimizing $\qtilde$ may be positive.  
    \end{itemize}
    These points summarize that the $\inf_{\qtilde, \boldrho} (f(0) + \gamma - \boldlambda^T \boldrho) \qtilde$ term is either equal to $-\infty$~(if $f(0) + \gamma < A(\boldlambda)$), or zero. If we restrict our attention to dual variables for which $g$ is nontrivial~(i.e., $>-\infty$), then this term contributes $0$ to $g(\boldlambda, \gamma)$. 

    We next consider the optimization over $\{q_i\}$. Let $r_i = \gamma - \boldlambda^T \xvec_i$, and observe that for any fixed $(\gamma, \boldlambda)$, the objective is separable in each $q_i$. Considering each  $q_i$-dependent term in $L$ individually, we get  
    \begin{align}
        \inf_{q_i \geq 1} \left\{ p_i \ftilde\lrp{\frac{q_i}{p_i}} + r_i q_i \right\} = p_i \inf_{w \geq 0} \left\{ \ftilde(w) + r_i w \right\} \eqcolon p_i \Phi(r_i). 
    \end{align}
    By definition this value is finite only if $r_i = \gamma - \boldlambda^T \xvec_i \in U_f = \{r \in \R: \Phi(r)> -\infty\}$. Putting together these components, we can conclude that for all feasible $(\boldlambda, \gamma)$, the dual objective is either $-\infty$, or defined as 
    \begin{align}
        g(\boldlambda, \gamma) = \sum_{i=1}^m p_i \Phi(\gamma - \boldlambda^T \xvec_i) - (\gamma - \boldlambda^T \boldmu). 
    \end{align}
    By strong duality, this means that the mean-constrained $f$-divergence for finitely supported $P$ is equal to: 
    \begin{align}
        &D_f^{\inf}(P, \boldmu) = \sup_{\boldlambda, \gamma \in \calL_{\finitesupport}}\; \mathbb{E}_P[\Phi(\gamma - \boldlambda^T X)] - (\gamma - \boldlambda^T \boldmu),   \\
        \text{with} \quad & \calL_{\finitesupport} = \left\{ (\boldlambda, \gamma) \in \R^{K+1}: \gamma + f(0) \geq \sup_{\boldrho \in \calX} \boldlambda^T \boldrho, \; \text{and} \; \gamma - \boldlambda^T \xvec \in U_f, \; \forall \xvec \in \finitesupport \right\}. 
    \end{align}
This completes the derivation of the dual of mean-constrained divergence for finitely supported $P$. \hfill \qedsymbol

\subsubsection{$U_f$ is an interval}
\label{proof:Uf-is-an-interval}
Before completing the proof of~\Cref{theorem:f-divergence}, we first show a useful result that says that $U_f$~(that is, the set of points at which $\Phi(\cdot)$ is greater than $-\infty$) is an interval. 
\begin{lemma}
    \label{lemma:Uf-is-an-interval}
    Let $\Phi(r) = \inf_{w \geq 0} \lrp{\ftilde(w) + r w} = -\ftilde^*(-r)$,/and let $U_f = \{r \in \R: \Phi(r) > -\infty\}$. Then, $U_f$ is half-line of the from $(r_{\min}, \infty)$ or $[-r_{\min}, \infty)$ for some $r_{\min} \in [-\infty, \infty)$. Additionally, if $\ftilde$ is continuously differentiable and strictly convex, then $r_{\min}= - \lim_{w \to \infty} \ftilde'(w)$. 
\end{lemma}

\begin{proof}
    Since $\ftilde^*$ is a convex function, its level set $ \mathrm{dom}(\ftilde^*) = \{r: \ftilde^*(r) < \infty\}$ is convex, which means that 
    \begin{align}
        U_f = \{r \in \R: \ftilde^*(-r) < \infty\} =   -\mathrm{dom}(\ftilde^*), 
    \end{align}
    is also convex. Since convex subsets of the real line are intervals, we can conclude that $U_f$ must also be an interval. 

    Additionally, note that $\Phi(r)$ is non-decreasing in $r$. That is, for $r_2 > r_1$, then $\ftilde(w) + r_2 w \geq \ftilde(w) + r_1 w$ since $w \geq 0$, which implies that $\Phi(r_2) \geq \Phi(r_1)$. Therefore the interval $U_f$ must be unbounded from above, and of the form $(r_{\min}, \infty)$ or $[r_{\min}, \infty)$ as claimed. 

    Finally, under the additional assumption that $\ftilde$ is $C^1$ and strictly convex, which means that $\ftilde'$ is non-decreasing on its domain, and define $b = \lim_{w \to \infty} \ftilde'(w)$.  Then, it follows that $\ftilde^*(s) < \infty$ for all $s < b$, and $\ftilde^*(s) = + \infty$ for $s > b$~(if $b < \infty$). This immediately implies that $\Phi(r) > - \infty$ for all $r > r_{\min} = - b$. 
    \begin{itemize}
        \item If $s > b$, then choose $w_0$ large enough to ensure that $\ftilde'(w_0) \leq (b+s)/2 < s$. Hence, by the non-decreasing property of $\ftilde'(\cdot)$, we have the following for $w \geq w_0$: 
        \begin{align}
            \ftilde(w) \leq \ftilde(w_0) + \ftilde'(w)(w - w_0) \leq \ftilde(w_0) + \frac{b+s}{2} (w-w_0). 
        \end{align}
        This leads to
        \begin{align}
            \ftilde(s) = \sup_{w \geq 0} \lrp{s w - \ftilde(w)} \geq \sup_{w \geq 0} \lrp{ s w - \ftilde(w_0) - \frac{b+s}{2} (w-w_0)} =  \sup_{w \geq 0} \lrp{\frac{s-b}{2} w + \text{constant}} = + \infty. 
        \end{align}
        So no $s > b$ can belong to $\mathrm{dom}(\ftilde^*)$. 
        \item For $s < b$, a similar argument shows that $s \in \mathrm{dom}(\ftilde^*)$. To see why, choose a $w_1$ large enough to ensure that $\ftilde'(w_1) \geq (s+b)/2 > s$. By the convexity, we have the bound 
        \begin{align}
            \ftilde(w) \geq \ftilde(w_1) + \ftilde'(w_1)(w-w_1). 
        \end{align}
        Hence, for $w \geq w_1$, we have 
        \begin{align}
            sw - \ftilde(w) \leq s w - \ftilde(w_1) + \ftilde'(w_1)(w-w_1) = (s - \ftilde'(w_1)) w + \text{constant}. 
        \end{align}
        Since $s - \ftilde'(w_1) < 0$, this is decreasing in $w$ for all $w \geq w_1$, and thus the supremum of $sw - \ftilde(w)$ over $w \geq 0$ is achieved at some point in $[0, w_1]$, and that this value is finite. Hence, for all $s<b$, we have $\ftilde^*(s) < \infty$, or equivalently $s \in \mathrm{dom}(\ftilde^*)$. 
    \end{itemize}
    As $U_f = - \mathrm{dom}(\ftilde^*)$, this completes the proof of~\Cref{lemma:Uf-is-an-interval}. 
\end{proof}

\begin{remark}[$U_f$ instantiations]
    \label{remark:Uf-examples} 
    Using~\Cref{lemma:Uf-is-an-interval}, we can instantiate the evaluation of $U_f$ for the three $f$-divergences discussed in the main text. 
    \begin{itemize}
        \item For relative entropy, we have $\ftilde(w) = - \log w$, $\ftilde'(w) = -1/w$ which converges to $0$ as $w \to \infty$. Hence, $b = 0$, and thus $r_{\min} = - b = 0$, and $U_f = (0, \infty)$. We can exclude the end-point $r_{\min}=0$, because $\Phi(0) = \inf_{w \geq 0} - \log w = -\infty$. 

        \item For squared Hellinger distance, we have $\ftilde(w) = (1-\sqrt{w})^2$, $\ftilde'(w)= 1 - 1/\sqrt{w}$. Hence $b = \lim_{ w \to \infty} \ftilde'(w) = 1$, which means that $r_{\min} = -1$, and $U_f = (-1, \infty)$. We can exclude the end-point $r_{\min}=-1$, because $\Phi(-1) = \inf_{w \geq 0} (1-\sqrt{w})^2 - w = -\infty$. 
        
        \item For $\chi^2$-divergence, we have $\ftilde(w) = w + 1/w - 2$, $\ftilde'(w) = 1 - 1/w^2$, which implies that $b = \lim_{w \to \infty} \ftilde'(w) = 1$, Again, this implies that $r_{\min} = -1$, and in this case we have $U_f = [-1, \infty)$. 
    \end{itemize}
\end{remark}

\subsubsection{Extension to Arbitrary $P$}
\label{proof:f-divergence-extension}
The extension to arbitrary $P$ supported on $\calX = [0,1]$ can be obtained following the exact same arguments used in~\Cref{theorem:continuous-klinf} for the case of relative entropy.  We present an outline below, omitting the details to avoid repetition. 

As before, we work with a sequence of mean-preserving discretization channels $\{\calK_k: k \geq 1\}$ supported on finite $\Delta_k$-covering sets $\{V_k: k \geq 1\}$ of $\calX$. We will also assume that each $V_k$ contains the end points $\{0,1\}^K$. While this condition is not strictly necessary, it greatly simplifies the proof, because, the dual domain $\calL_k$ then becomes independent of $k$, and can be written as 
\begin{align}
    \fdivdualdomain = \{(\boldlambda, \gamma) \in \R^{K+1}: \gamma + f(0) \geq \beta(\boldlambda), \; \gamma - \beta(\boldlambda) \in U_f \}, \qtext{where} \beta(\boldlambda) = \sup_{\xvec \in \calX} \boldlambda^T \xvec.  
\end{align}
This is due to the fact that $\sup_{\xvec \in V_k} \boldlambda^T\xvec = \sup_{\xvec \in \calX} \boldlambda^T\xvec =   \sup_{\xvec \in \{0,1\}^K} \boldlambda^T\xvec$  under the assumption that $\{0,1\}^K \subset V_k$ for all $k \geq 1$. 
To complete the proof, we need to repeat the two steps involved in the proof of~\Cref{theorem:continuous-klinf}, represented by~\Cref{lemma:interchange-1} and~\Cref{lemma:interchange-2}. In particular, let 
\begin{align}
    I_{f,k} = \inf_{\substack{Q \in \calP(V_k) \\ \mathbb{E}_Q[X] = \boldmu }} D_f(P \calK_k \parallel Q), \qtext{and}
    I_f = \inf_{\substack{Q \in \calP(\calX) \\ \mathbb{E}_Q[X] = \boldmu }} D_f(P \parallel Q). 
\end{align}
Then, the first step is to show that $\lim_{k \to \infty} I_{f,k} = I_f$. The proof is identical to that of~\Cref{lemma:interchange-1}. Let $\calQ$ and $\calQ_k$ denote the optimization domains in the definitions of $I_f$ and $I_{f,k}$ respectively. Then, due to the mean-preserving property of $\calK_k$, we observe that $Q \in \calQ$ implies $Q \calK_k \in \calQ_k$. Hence, by the data processing inequality, we get $I_{f,k} = \inf_{Q_k \in \calQ_k} D_f(P \calK_k \parallel Q_k) \leq \inf_{Q \in \calQ} D_f( P \calK_k \parallel Q \calK_k) \leq \inf_{Q \in \calQ} D_f(P \parallel Q)$. Hence, we have $\limsup_{k \to \infty} I_{f,k} \leq I_f$. For the other direction, we can again use the Prokhorov + lower-semicontinuity argument of~\Cref{lemma:interchange-1}. We omit the details to avoid repetition. 

To conclude the proof, it remains to establish the continuity of the dual. For any $\theta \equiv (\boldlambda, \gamma) \in \fdivdualdomain$, define 
\begin{align}
    H_k(\theta) \coloneqq \mathbb{E}_{P_k}[\Phi(\gamma - \boldlambda^T X)] - (\gamma - \boldlambda^T \boldmu), \quad 
    H(\theta) \coloneqq \mathbb{E}_{P}[\Phi(\gamma - \boldlambda^T X)] - (\gamma - \boldlambda^T \boldmu). 
\end{align}
By the finite-support dual derivation, we have 
\begin{align}
    I_{f,k} = \sup_{\theta \in \fdivdualdomain} H_k(\theta), \qtext{for all} k \geq1. 
\end{align}
To pass to the limit, we argue exactly as in~\Cref{lemma:interchange-2} for $\KLinf$, replacing the $\log$ with $\Phi$. Then, by the mean-preservation property $\mathbb{E}[X_k \mid X] = X$, and the concavity of $\Phi(\cdot)$, we get $H_k(\theta) \leq H(\theta)$ for all $\theta \in \fdivdualdomain$, which leads to 
\begin{align}
    \limsup_{k \to \infty} \lrp{ \sup_{\theta \in \fdivdualdomain} H_k(\theta) } \leq \sup_{\theta \in \fdivdualdomain} H(\theta). 
\end{align}
For the $\liminf$ direction, we again restrict to $\epsfdivdualdomain \subset \fdivdualdomain$~(defined as in the $\KLinf$ proof) on which $\Phi(\gamma - \boldlambda^T \xvec)$ is uniformly Lipschitz in $\xvec$. This fact combined with $\|X_k - X\|_\infty \overset{a.s.}{\longrightarrow} 0$ implies the uniform convergence $\sup_{\theta \in \epsfdivdualdomain} \left\lvert H_k(\theta) - H(\theta) \right\rvert \to 0$. Finally, taking~$\epsilon \downarrow 0$, and invoking the inequality~\eqref{eq:interchange-2-proof-3} used in the proof of~\Cref{lemma:interchange-2} yields the required $\lim_{k \to \infty} \sup_{\theta \in \fdivdualdomain} H_k(\theta) = \sup_{\theta \in \fdivdualdomain} H(\theta)$. This concludes the proof.

\subsection{Details of \Cref{corollary:f-divergence}}
\label{proof:corollary-f-divergence}

\subsubsection{Hellinger Divergence}
\label{proof:corollary-Hellinger}
Hellinger divergence is associated with $f(u) = (\sqrt{u} - 1)^2$, so the associated perspective function $\ftilde(w) = w f(1/w) = w (\sqrt{1/w} - 1)^2 = w \lrp{\frac 1 w + 1 - 2 \frac 1 {\sqrt{w}} } = (\sqrt{w} -1)^2$. This leads to the following definition of $\Phi(r)$: 
\begin{align}
    \Phi(r) = \inf_{w \geq 0} \lrp{1 - 2 \sqrt{w} + w + rw } = \frac{r}{r+1} = 1 - \frac{1}{r+1}.  
\end{align}
This implies that the domain of $\Phi$ is $U_f = (-1, \infty)$. Plugging this expression in the general dual derived in~\Cref{theorem:f-divergence}, and using $c = \gamma - \boldlambda^T \boldmu$ and $Z_{\boldlambda} = \boldlambda^T(X-\boldmu)$, we get 
\begin{align}
    \mathbb{E}_P[\Phi(\gamma - \boldlambda^T X)] - \gamma - \boldlambda^T \boldmu = 1 - \mathbb{E}_P\left[ \frac{1}{\gamma - \boldlambda^T X + 1} \right] - c = 1 - c - \mathbb{E}_P\left[ \frac{1}{1 + c + Z_{\boldlambda}} \right].
\end{align}
Now, let $\boldlambda_0$ denote the normalized parameter $\boldlambda/(1+c)$, and observe that 
\begin{align}
    1 + c - Z_{\boldlambda} = (1+c) \lrp{ 1 - \boldlambda_0^T(X - \boldmu)}, 
\end{align}
which implies that 
\begin{align}
    1 - c - \mathbb{E}\left[ \frac{1}{1 + c - Z_{\boldlambda}} \right] = 2 - \left\{ (1+c) + \frac{1}{1+c} \underbrace{\mathbb{E}\left[ \frac{1}{1 - \boldlambda_0^T(X-\boldmu)} \right]}_{\coloneqq A(\boldlambda_0)} \right\}. 
\end{align}
Now, for a fixed $\boldlambda_0$, the term above is maximized~(over $d = 1+c$) at $d^* = \sqrt{A(\boldlambda_0)}$ yielding the value $2 - 2 \sqrt{A(\boldlambda_0)}$. The feasibility of $(\gamma, \boldlambda)$ directly translates into $1-\boldlambda_0^T(\xvec - \boldmu) > 0$ for all $x \in \calX$; that is, $\boldlambda_0 \in \calL_{\boldmu} = \{\boldlambda: 1-\boldlambda^T(\xvec - \boldmu) \geq 0, \; \forall \xvec \in \calX\}$.  Hence, we get the required 
\begin{align}
 D_{\mathrm{Hel}}^{\inf}(P, \boldmu)  = \sup_{\boldlambda \in \calL_{\boldmu}}  \lrp{ 2 -  2\sqrt{\mathbb{E}_P\left[ \frac{1}{1 - \boldlambda^T(X-\boldmu)  }  \right] }}. 
\end{align}

\subsubsection{Chi-Squared Divergence}
\label{proof:corollary-Chi-squared}
    In this case, we have $f(u) = (u-1)^2$, which leads to $\ftilde(w) = w + \frac 1 w - 2$ for $w>0$. Then, 
    \begin{align}
        \Phi(r) = \inf_{w > 0} \lrp{ w + \frac 1 w - 2 + r w} = 2 \lrp{\sqrt{1+r} - 1},  \quad U_f = [-1, \infty). 
    \end{align}
    With $c = \gamma - \boldlambda^T \boldmu$ and $Z_{\boldlambda} = \boldlambda^T(X-\boldmu)$, we have 
    \begin{align}
        \mathbb{E}_P[\Phi(\gamma - \boldlambda^TX)]  - \gamma  + \boldlambda^T \boldmu = 2 \mathbb{E}_P[\sqrt{1 + c - Z_{\boldlambda}}] - 2 - c. 
    \end{align}
    Let us use $s$ to denote $1+c$, and set $\boldlambda_0 = \boldlambda/s$, which implies $\sqrt{1+c + Z_{\boldlambda}} = \sqrt{s} \sqrt{1 + \boldlambda_0^T(X-\boldmu)}$. Hence, for a fixed $\boldlambda_0$, the objective is 
    \begin{align}
        2 \sqrt{s} B(\boldlambda_0) - s - 1, \qtext{with} B(\boldlambda_0) \coloneqq \mathbb{E}_P[\sqrt{1- \boldlambda^T_0(X-\boldmu)}]. 
    \end{align}
    On maximizing over $s$, we get that $s^* = B(\boldlambda_0)^2$, which gives the objective $B(\boldlambda_0)^2 -1$. 

    The feasibility of the reparameterized variable requires $1-\boldlambda_0^T(\xvec-\boldmu) \geq 0$ for all $\xvec  \in \calX$, or $\boldlambda\in \calL_{\boldmu}$. Hence, we get the required dual formulation
    \begin{align}
        D_{\chi^2}^{\inf}(P, \boldmu)& = \sup_{\boldlambda \in \calL_{\boldmu}} \lrp{ \left[ \mathbb{E}_P \sqrt{1 - \boldlambda^T(X-\boldmu)} \right]^2 - 1 }. 
    \end{align}
    This completes the proof of~\Cref{corollary:f-divergence}. \hfill \qedsymbol

\section{Deferred Proofs from~\Cref{sec:general-argument-approximation-constraints}}
\label{proof:general-limiting-argument}

\subsection{Proof of~\Cref{theorem:general-approximate-constraint}}
As we mentioned after the statement of~\Cref{theorem:general-approximate-constraint}, its proof has two main parts. The first is to establish the primal continuity which says that $I_k$ converges to $I \equiv I(P, g, \calC)$. The second, and more involved, step is to show the existence of a limit of the dual representations of the discretized problems. A key challenge, that we did not face in proving~\Cref{theorem:continuous-klinf}, is that the dual objectives and domains for the discretized problems could change with $k \geq 1$. To address this, a key idea is to first shrink the domain from $\Theta_k$ to $\Theta_k^{(t)}$ for $t \in (0,1)$, and then transport the domain to the limiting set $\Theta^{(t)}$ via the map $\tau_{k,t}$ in~\Cref{assump:general-dual-function}, and study the limiting value of $F_{k,t}(\theta) \coloneqq H_k(\tau_{k,t}(\theta), P_k)$ which is now defined on $\Theta^{(t)}$ for all $k \geq 1$.  
Before going into the details, we need to introduce some notation. For any $t \in (0,1)$ and $k \geq 1$, introduce 
\begin{align}
    A_k(t) \coloneqq \sup_{\theta \in \Theta^{(t)}_k}, \quad 
    \widetilde{A}_k(t) \coloneqq \sup_{\theta \in \Theta^{(t)}} F_{k,t}(\theta), \quad 
    A(t) \coloneqq \sup_{\theta \in \Theta^{(t)}} H^{(t)}(\theta, P). \label{eq:A-terms}
\end{align}
The shrinking away from the boundary to define $\Theta^{(t)}_k$ and $\Theta^{(t)}$ allows us to avoid boundary singularities analogous to the $\KLinf$ case in~\Cref{theorem:continuous-klinf}. Once the limit is established for a fixed $t$, then we appeal to concavity to allow us to take $t \downarrow 0$. 
We now present the proof of~\Cref{theorem:general-approximate-constraint} in five Lemmas, starting with  the primal continuity step. 
\begin{lemma}
    \label{lemma:approx-contraint-primal-1} 
    Let $\calQ = \{Q \in \calP(\calX): \mathbb{E}_Q[g(X)] \in \calC\}$ and $\calQ_k = \{Q \in \calP(V_k): \mathbb{E}_Q[g(X)] \in \calC_k\}$ denote the primal domains. 
    Under~\Cref{assump:general-constraint} and~\ref{assump:general-discretization}, for any $Q \in \calQ$, we have  $Q_k = Q\calK_k \in \calQ_k$, which implies $\lim_{k \to \infty} I_k \;=\; I \equiv I(P, g, \calC)$.
\end{lemma}
This result completes the primal part of the proof by showing that the discretized feasible sets, $\calQ_k$, are rich enough to approximate the original feasible set $\calQ$ sufficiently well, while the lower semicontinuity of $D$ provides the lower bound. We now turn to the dual part of the argument. 

We know from the assumptions that the dual  of $I_k$ is characterized by the terms $(H_k, \Theta_k, \psi_k, b_k)$. Additionally, for any $t \in (0,1)$, we introduce the function 
\begin{align}
    F_{k,t}(\theta) \coloneqq H_k(\tau_{k,t}(\theta), P_k), \quad \theta \in \Theta^{(t)}. 
\end{align}
With this simple trick, for every $t \in (0,1)$, we now have a sequence of functions defined on the same  domain $\Theta^{(t)}$, which allows us to seek a limiting dual objective. Our next step is to identify this limit on the retracted domain $\Theta^{(t)}$. The regularity assumptions of~\Cref{assump:general-dual-function} imply that $\{F_{k,t}: k \geq 1\}$ is a uniformly bounded and equicontinuous family on the compact set $\Theta^{(t)}$, and thus we start off by identifying its limit. 
\begin{lemma}
    \label{lemma:approx-constraint-limiting-1} For a fixed  $t \in (0, 1)$, there exists a unique continuous function $H^{(t)}(\cdot, P)$ on $\Theta^{(t)}$ such that 
    \begin{align}
        \sup_{\theta \in \Theta^{(t)}} \; \left\lvert F_{k, t}(\theta) - H^{(t)}(\theta, P) \right\rvert \; \stackrel{k \to \infty}{\longrightarrow} \; 0. 
    \end{align}
\end{lemma}
This result follows from an application of Arzel\'a-Ascoli theorem, and it presents a candidate limiting dual objective function $H^{(t)}$ on the interior domain $\Theta^{(t)}$ whose uniqueness is then justified by the   ``dense subset'' $\calD_t$ assumption. 
The next step is to relate these fixed-$t$ limits back to the original dual values $\sup_{\theta \in \Theta_k} H_k(\theta, P_k)$. We organize this part of the proof into the following chain
\begin{align}
    \lim_{k \to \infty} \sup_{\theta \in \Theta_k} H_k(\theta, P_k) \stackrel{?}{=} \lim_{k \to \infty} \sup_{t \in (0,1)} A_k(t) \stackrel{?}{=} \sup_{t \in (0,1)} \lim_{k \to \infty} A_k(t) \stackrel{?}{=} \sup_{t \in (0,1)}  A(t),      \label{eq:approx-constraint-proof-1}
\end{align}
where $A_k(t)  \coloneqq \sup_{\theta \in \Theta_k^{(t)}}H_k(\theta, P_k)$ and $A(t) \coloneqq \sup_{\theta \in \Theta^{(t)}}H^{(t)}(\theta, P)$ were defined in~\eqref{eq:A-terms}. 

We first justify the passage from the full  domain $\Theta_k$ to its retracted versions $\Theta^{(t)}_k$~(that is, the first ``?'' above). This is the point in the proof where the concavity of $H_k$ enters, since for a concave objective, there is no loss of optimality by shrinking the domain towards the interior and then taking a supremum over $t$. 
\begin{lemma}
    \label{lemma:approx-constraint-interchange-1} Under the assumptions of~\Cref{theorem:general-approximate-constraint}, let $\widetilde{\Theta}$ be either $\Theta$ or some $\Theta_k$ for $k \geq 1$, and let $h:\widetilde{\Theta} \to \R$ be a concave function with $h(\theta_0)> -\infty$. Then, we have 
    \begin{align}
        \sup_{\theta \in \widetilde{\Theta}} h(\theta) \; = \; \sup_{t \in (0,1)} \sup_{\theta \in \widetilde{\Theta}} h(\theta). 
    \end{align} 
\end{lemma}
Applying this Lemma with $h = H_k(\cdot, P_k)$ yields the first equality in~\eqref{eq:approx-constraint-proof-1}. We next turn to the last equality in~\eqref{eq:approx-constraint-proof-1}, namely the identification of the fixed-$t$ limit $A(t)$ with the limit of the retracted dual suprema $A_k(t)$. For each fixed $t$, this comparison is obtained in two steps. The first is to compare $A(t)$ with the transported supremum $\widetilde{A}_k(t) \coloneqq \sup_{\theta \in \Theta^{(t)}} F_{k,t}(\theta)$, and then compare $\widetilde{A}_k(t)$ with the retracted supremum $A_k(t)$. 
\begin{lemma}
    \label{lemma:approx-constraint-interchange-2} Let $A_k(t) \coloneqq \sup_{\theta \in \Theta_k^{(t)}}H_k(\theta, P_k)$ and $A(t) \coloneqq \sup_{\theta \in \Theta^{(t)}}H^{(t)}(\theta, P)$. Then,~\Cref{lemma:approx-constraint-limiting-1} implies that $A_k(t) \to A(t)$ as $k \to \infty$, for each fixed $t$. If $\widetilde{A}_k(t) = \sup_{\theta \in \Theta^{(t)}} F_{k,t}(\theta) =  \sup_{\theta \in \Theta_k^{(t)}} H_k(\tau_{k,t}(\theta), P_k)$, then we also have 
    \begin{align}
        \lim_{k \to \infty} \, \left \lvert \widetilde{A}_k(t) - A_k(t) \right\rvert = 0 \quad \implies \quad \lim_{k \to \infty} \, \left \lvert {A}_k(t) - A(t) \right\rvert = 0
    \end{align}
\end{lemma}
This lemma shows that for each fixed $t$, the retracted dual value $A_k(t)$ converges to its limit $A(t)$. The only remaining issue is that the theorem requires a supremum over $t \in (0,1)$, and so we must justify interchanging the limit in $k$ with the supremum over $t$; that is, the middle ``?'' in~\eqref{eq:approx-constraint-proof-1}. 
\begin{lemma}
    \label{lemma:approx-constraint-unique-limit}
    With $A_k(t)$ and $A(t)$ as in~\Cref{lemma:approx-constraint-interchange-1}, we have 
    \begin{align}
        \lim_{k \to \infty} \sup_{t \in (0,1)} A_k(t) = \sup_{t \in (0,1)} \lim_{k \to \infty}  A_k(t)  =  \sup_{t \in (0,1)} A(t). 
    \end{align}
    
\end{lemma}
Combining~\Cref{lemma:approx-constraint-interchange-1},~\Cref{lemma:approx-constraint-interchange-2}, and~\Cref{lemma:approx-constraint-unique-limit}, we get 
\begin{align}
    I(P, g, \calC) = \sup_{t \in (0,1)} \sup_{\theta \in \Theta^{(t)}} H^{(t)}(\theta, P). 
\end{align}
This proves the limiting dual representation of a family of retracted domains $\{\Theta^{(t)}: t \in (0,1)\}$. To conclude~\Cref{theorem:general-approximate-constraint} from this, it remains to identify this collection of limiting functions with a single dual objective on all $\Theta$.  
    
Under the additional compatibility assumption that there exists a concave $H(\cdot, P): \Theta \to \R$ such that for all $t \in (0,1)$, we have $H(\theta, P) = H^{(t)}(\theta, P)$ for all $\theta \in \Theta^{(t)}$, we can apply~\Cref{lemma:approx-constraint-interchange-1} again to remove the parameter $t$ and recover the full dual representation 
    \begin{align}
        \sup_{t\in (0, 1)} \sup_{\theta \in \Theta^{(t)}} H^{(t)} (\theta, P) \; = \; \sup_{\theta \in \Theta} H(\theta, P), \qtext{which implies}  I(P, g, \calC) = \sup_{\theta \in \Theta} H(\theta, P). 
    \end{align}
This concludes the proof of~\Cref{theorem:general-approximate-constraint}. \hfill \qedsymbol  

\subsubsection{Proof of~\Cref{lemma:approx-contraint-primal-1}}
\label{proof:approx-constraint-primal-1}
Let us recall the terms 
\begin{align}
    \calQ \coloneqq \{Q \in \calP(\calX): \mathbb{E}_Q[g(X)] \in \calC\}, \quad \calQ_k \coloneqq \{Q \in \calP(V_k): \mathbb{E}_Q[g(X)] \in \calC_k \}, 
\end{align}
where $\calC_k = \calC + B_\infty(\boldsymbol{0}, \eta_k)$. We start with the observation that if $Q \in \calQ$, then $Q_k = \calQ \calK_k$ lies in $\calQ_k$. This is simply due to the fact that $\calK_k(\cdot \mid X)$ is supported on $V_k$, and 
\begin{align}
    \left\lVert \mathbb{E}_{Q\calK_k}[g(X)] - \mathbb{E}_Q[g(X)] \right\rVert_\infty & = \left\lVert \mathbb{E}[ \mathbb{E}[g(X_k) \mid X]] - \mathbb{E}[g(X)] \right\rVert_\infty  \leq  \mathbb{E} \big[ \left\lVert  g(X_k)  - g(X) \right\rVert_\infty  \big]\leq \omega_g(\Delta_k) \leq \eta_k. 
\end{align}
Thus, $\mathbb{E}_{Q\calK_k}[g(X)] \in \calC_k$ as required, since it is at most $\eta_k$ away from  $\calC$ in $\|\cdot\|_\infty$ norm. 
This fact, along with the data-processing inequality~(DPI) assumption on $D(\cdot, \cdot)$ gives us one direction of the proof, since for any $Q \in \calQ$ and  $Q_k = \calK_k Q \in \calQ_k$
\begin{align}
    I_k = \inf_{Q' \in \calQ_k} D(P_k, Q') \leq D(P_k, Q_k) = D(\calK_k P, \calK_k Q) \leq D(P, Q), 
\end{align}
where the last inequality follows from the DPI condition of~\Cref{assump:general-divergence}. This implies that 
\begin{align}
    \limsup_{k \to \infty} I_k \leq \inf_{Q \in \calQ} D(P, Q) = I(P, g, \calC). 
    \label{eq:approx-constraint-proof-7}
\end{align}
For the other direction, we will rely on the lower semicontinuity of the divergence $D$. In particular, for each $k \geq 1$, and pick a $Q_k \in \calQ_k$, such that $I_k \geq D(P_k, Q_k) - \epsilon_k$ for some sequence $\{\epsilon_k\}_{k \geq 1} \downarrow 0$.
Now, let us select a subsequence $\{k^\ell: \ell \geq 1\}$ of the natural numbers that satisfy $\liminf_{k \to \infty} I_k = \lim_{\ell \to \infty} I_{k^\ell}$. 
Because $\calX$ is compact, the collection $\{Q_{k^\ell}: \ell \geq 1\}$ is tight, and hence it has a weakly convergent subsequence; that is, there exists a subsequence $\{k^\ell_j\}_{j \geq 1}$ such that $Q_{k^\ell_j} \Longrightarrow Q_\infty \in \calP(\calX)$. 
Also, because of the uniform approximation property of $\calK_k$, it follows that $P_k = P \calK_k$ also converges weakly to $P$ by the same bounded-continuous test function argument used earlier in~\Cref{lemma:convergence-of-KL}. We will drop the $\ell$-superscript from $k_j^\ell$, and simply refer to it is $k_j$ for the rest of the proof to simplify notation. 

Because the constraint function $g:\calX \to\R^J$ is continuous on the compact domain $\calX$, it is also bounded and continuous coordinate-wise. Then, the weak convergence of $Q_{k_j}$ also implies the convergence of the bounded functions 
\begin{align}
    \mathbb{E}_{Q_{k_j}}[g(X)] \; \longrightarrow \; \mathbb{E}_{Q_\infty}[g(X)] \quad \implies \quad \|\mathbb{E}_{Q_{k_j}}[g(X)] - \mathbb{E}_{Q_{\infty}}[g(X)] \|_\infty  \; \longrightarrow \; 0. 
\end{align}
Now, observe that 
\begin{align}
    \|\mathbb{E}_{Q_{\infty}}[g(X)] - \calC \|_\infty  &\leq \|\mathbb{E}_{Q_{k_j}}[g(X)] - \calC \|_\infty  + \|\mathbb{E}_{Q_{k_j}}[g(X)] - \mathbb{E}_{Q_{\infty}}[g(X)] \|_\infty \\
    & \leq \eta_{k_j} + + \|\mathbb{E}_{Q_{k_j}}[g(X)] - \mathbb{E}_{Q_{\infty}}[g(X)] \|_\infty.  
\end{align}
Together, the previous two displays imply that $\|\mathbb{E}_{Q_{\infty}}[g(X)] - \calC \|_\infty \to 0$ with $j \to \infty$. Since we have assumed that the set $\calC$ is closed, this means that $\mathbb{E}_{Q_{\infty}}[g(X)] \in \calC$. Or in other words, the limiting distribution $Q_\infty$ is feasible for the original problem, and it lies in $\calQ$. 

It now remains to exploit the weak lower semicontinuity of the divergence $D$, to get 
\begin{align}
    \liminf_{j \to \infty} D(P_{k_j}, Q_{k_j}) \; \geq \; D\lrp{ \liminf_{j \to \infty} P_{k_j}, \liminf_{j \to \infty} Q_{k_j} } = D(P, Q_\infty) \geq \inf_{Q \in \calQ} D(P, Q) = I(P, g, \calC). 
\end{align}
Since $\{k_j: j \geq 1\}$ is a subsequence (of a subsequence) that achieves the $\liminf$, we have the following: 
\begin{align}
    \liminf_{k \to \infty} I_k = \lim_{j \to \infty} I_{k_j} \geq \liminf_{j \to \infty} \lrp{D(P_{k_j}, Q_{k_j}) - \epsilon_{k_j}} \geq \lim_{j \to \infty} \lrp{I(P, g, \calC) - \epsilon_{k_j}} = I(P, g, \calC). \label{eq:approx-constraint-proof-8}
\end{align}
Taken together,~\eqref{eq:approx-constraint-proof-7} and~\eqref{eq:approx-constraint-proof-8} give us the required conclusion that $\lim_{k \to \infty} I_k = I$. 

\subsubsection{Proof of~\Cref{lemma:approx-constraint-limiting-1}}
\label{proof:approx-constraint-limiting-1}
We begin with the simple observation that functions $\tau_{k,t}$ introduced in~\eqref{eq:tau-k-t-def} in~\Cref{assump:general-dual-limits} is $1$-Lipschitz for all $k,t$. This is a direct consequence of the fact that the Euclidean projections are non-expansive, and hence 
\begin{align}
    \|\tau_{k,t}(\theta) - \tau_{k,t}(\theta')\|_2  &= (1-t) \left\lVert  \Pi_{\Theta_k}\lrp{\frac{\theta - t \theta_0}{1-t}} - \Pi_{\Theta_k}\lrp{\frac{\theta' - t \theta_0}{1-t}} \right\rVert_2 \\
    & \leq (1-t) \left \lVert \frac{\theta - \theta'}{1-t} \right\rVert_2 = \|\theta - \theta'\|_2. \label{eq:lipschitz-property-of-identification-map}
\end{align}

Now, for any $t \in (0,1)$, the collection of functions $\{F_{k,t}: k \geq 1\}$ satisfies the following two conditions: 
\begin{itemize}
    \item It is uniformly bounded on $\Theta^{(t)}$; that is, 
    \begin{align}
        \sup_{k}\, \sup_{\theta \in \Theta^{(t)}} |F_{k,t}(\theta)| &= \sup_{k}\, \sup_{\theta \in \Theta^{(t)}}\left\lvert \mathbb{E}_{P_k}[\psi_k(X, \tau_{k,t}(\theta))]  + b_k(\tau_{k,t}(\theta))\right\rvert  \\ 
        & \leq  \sup_{k}\, \sup_{\theta' \in \Theta_k^{(t)}}\left\lvert \mathbb{E}_{P_k}[\psi_k(X, \theta')]  + b_k(\theta')\right\rvert  < \infty,
    \end{align}
    by the uniform boundedness condition of~\Cref{assump:general-dual-function}. 
    \item This collection is also equicontinuous, since 
    \begin{align}
        \sup_{k} |F_{k,t}(\theta)  - F_{k,t}(\theta')| & = \sup_{k} |H_{k}(\tau_{k,t}(\theta))  - H_{k}(\tau_{k,t}(\theta'))| \leq \omega_t(\|\tau_{k,t}(\theta) - \tau_{k,t}(\theta')\|_2) 
         \leq \omega_t(\|\theta -\theta'\|_2), 
    \end{align}
    where the last inequality uses the 1-Lipschitz  property of $\tau_{k,t}$ established in~\eqref{eq:lipschitz-property-of-identification-map}. 
\end{itemize}
Thus, we know that for any $t \in (0,1)$, the dual set $\Theta^{(t)}$ is compact~(\Cref{assump:general-dual-limits}), and the collection of functions $\{F_{k,t}: k \geq 1\}$ on $\Theta^{(t)}$ is uniformly bounded and equicontinuous. Hence, by~Arzel\'a-Ascoli~(\Cref{fact:arzela-ascoli}), we know that there exists a subsequence $\{k_j: j \geq 1\}$, and a continuous function $H^{(t)}(\cdot, P)$, on $\Theta^{(t)}$, such that 
\begin{align}
    \sup_{\theta \in \Theta^{(t)}} \left\lvert F_{k_j,t}(\theta) - H^{(t)}(\theta, P) \right\rvert \; \stackrel{j \to \infty}{\longrightarrow} \; 0. 
\end{align}
This shows a subsequential convergence of the collection of functions. But the assumption that there exists a dense subset of $\Theta^{(t)}$ on which $F_{k,t}$ is convergent implies that every such subsequential limit agrees on a dense set $\calD_t$, and hence everywhere on $\Theta^{(t)}$ by continuity. Therefore, the whole sequence converges uniformly to a unique continuous function $H^{(t)}(\cdot, P)$.  This concludes the proof.  

\subsubsection{Proof of~\Cref{lemma:approx-constraint-interchange-1}}
\label{proof:approx-constraint-interchange-1}

Let $M$ denote the value $\sup_{\theta \in \widetilde{\Theta}} h(\theta) \in (-\infty, \infty]$. 

\paragraph{The ``$\leq$'' direction.} For any $t \in (0,1)$ and $\theta' \in \widetilde{\Theta}$,  let $\theta = \theta_0 t + (1-t) \theta'$ denote an arbitrary element of $\widetilde{\Theta}^{(t)}$. Since  $\widetilde{\Theta}$ is convex, it follows immediately that $\widetilde{\Theta}^{(t)} \subset \widetilde{\Theta}$, hence we have 
\begin{align}
    \sup_{\theta \in \widetilde{\Theta}^{(t)}} h(\theta) \leq  \sup_{\theta \in \widetilde{\Theta}} h(\theta) = M. 
\end{align}
This inequality is maintained on taking a supremum over all $t \in (0,1)$, proving the required inequality. 

\paragraph{The ``$\geq$'' direction.} This direction crucially relies on the concavity of the function $h$. We break the proof into two cases. First we consider the case of $M = \infty$. This means that for any $y \in \R$, there exists a $\theta_y \in \widetilde{\Theta}$, such that $h(\theta_y) \geq y$. Now, consider any $t \in (0,1)$, and observe that $\theta_{y,t} = t \theta_0 + (1-t) \theta_y$ lies in $\widetilde{\Theta}^{(t)}$, and 
\begin{align}
    h(\theta_{y,t}) = h(t \theta_0 + (1-t) \theta_y) \geq t h(\theta_0) + (1-t) h(\theta_y) \geq y - t \lrp{|h(\theta_0)| + |y|}
\end{align}
by an application of Jensen's inequality. Since $t \in (0,1)$ is arbitrary and for any $\epsilon > 0$, we can select $t_\epsilon < \epsilon/\lrp{|h(\theta_0)| + |y|}$, to ensure that $\sup_{\theta \in \widetilde{\Theta}^{(t_\epsilon)}} h(\theta)  \geq y - \epsilon$. As a result, for any $y \in \R$, we have 
 $\sup_{t \in (0,1)}   \sup_{\theta \in \widetilde{\Theta}^{(t)}} h(\theta)  \geq  y - \epsilon$. Since $\epsilon>0$ and $y \in \R$ were arbitrary, this implies that $\sup_{t \in (0,1)}   \sup_{\theta \in \widetilde{\Theta}^{(t)}} h(\theta) = \infty$ as required. 

 A similar argument works for the case of finite $M$ as well. In particular, for any $\epsilon > 0$, there exists a $\theta_\epsilon \in \widetilde{\Theta}$ such that $h(\theta_\epsilon) > M - \epsilon$. Now, for any $t \in (0, 1)$, we can define $\theta_{t,\epsilon} = t \theta_0 + (1-t) \theta_\epsilon$, such that 
 \begin{align}
     h(\theta_{t,\epsilon}) = h(t \theta_0 + (1-t) \theta_\epsilon) \geq t h(\theta_0) + (1-t) h(\theta_\epsilon) \geq M - \epsilon - t (M + |h(\theta_0)|). 
 \end{align}
For all $t < \epsilon/(M + |h(\theta_0)|)$, we then get that $h(\theta_{t,\epsilon}) \geq M - 2\epsilon$, which implies that 
\begin{align}
    \sup_{t\in (0,1)} \sup_{\theta \in \widetilde{\Theta}^{(t)}} h(\theta) \geq M - 2\epsilon. 
\end{align}
Since $\epsilon >0$ is arbitrary, the result follows. 

\subsubsection{Proof of~\Cref{lemma:approx-constraint-interchange-2}}
\label{proof:approx-constraint-interchange-2}
The goal of this lemma is to essentially make precise the  relation 
\begin{align}
  \underbrace{\sup_{\theta \in \Theta_k^{(t)} } H_k^{(t)}(\theta, P_k) }_{\coloneqq A_k(t)} \; \approx \;
  \underbrace{\sup_{\theta \in \Theta^{(t)} } F_{k,t}^{(t)}(\theta) }_{\coloneqq \widetilde{A}_k(t)} \; \approx \;
  \underbrace{\sup_{\theta \in \Theta^{(t)} } H^{(t)}(\theta, P) }_{\coloneqq A(t)}. \label{eq:approx-constraint-proof-2}
\end{align}
The first statement is due to the following observation:
\begin{align}
    |A(t) - \widetilde{A}_k(t)| & = \left\rvert \sup_{\theta \in \Theta_k^{(t)}} H(\theta, P) \; - \; \sup_{\theta \in \Theta^{(t)}} F_{k,t}(\theta)  \right\rvert 
     \leq \sup_{\theta \in \Theta^{(t)}} \left \lvert H^{(t)}(\theta, P) - F_{k,t}(\theta) \right\rvert,\label{eq:approx-constraint-proof-6}
\end{align}
which converges to $0$ by~\Cref{lemma:approx-constraint-limiting-1}. This establishes the $\widetilde{A}_k(t) \approx A(t)$ part of~\eqref{eq:approx-constraint-proof-2}.

The next step is to relate $\widetilde{A}_k(t)$ and $A_k(t)$. We start with the simple observation that $\{\tau_{k,t}(\theta): \theta \in \Theta^{(t)}\} \subset \Theta^{(t)}_k$, which means that 
\begin{align}
    \widetilde{A}_k(t) = \sup_{\theta \in \Theta^{(t)}} H_k(\tau_{k,t}(\theta), P_k) \leq  \sup_{\theta \in \Theta_k^{(t)}} H_k(\theta, P_k)  = A_k(t) \quad \implies A_k(t) - \widetilde{A}_k(t) \geq 0.  \label{eq:approx-constraint-proof-3}
\end{align}
This inequality is valid for all $k \geq 1$ and $t \in (0,1)$. 

Next, consider an arbitrary $\epsilon>0$, and let $\vartheta \in \Theta_k^{(t)}$ be a point such that $H_k(\vartheta, P_k) \geq A_k(t) - \epsilon$.  Since $\vartheta$ is an element of $\Theta_k^{(t)}$, by definition there must exist a $\eta \in \Theta_k$, such that $\vartheta = \theta_0 t + (1-t) \eta$. Next, we associate a point $\theta \in \Theta^{(t)}$ to the $\epsilon$-suboptimal point $\vartheta \in \Theta^{(t)}_k$ defined as 
\begin{align}
    \theta &= t \theta_0 + (1-t) \Pi_\Theta \lrp{\frac{\vartheta - t\theta_0}{1-t}} = t \theta_0 + (1-t) \Pi_\Theta(\eta). 
\end{align}
Here, $\Pi_{\Theta}$ denotes an $\ell_2$-projection  onto $\Theta$. 
Note that both $\vartheta = t\theta_0 + (1-t)\eta \in \Theta_k^{(t)}$ and $\theta = t \theta_0 + (1-t) \Pi_{\Theta}(\eta) \in \Theta^{(t)}$ are defined using the same point, $\eta \in \Theta_k$. Now, let us compare $\vartheta \in \Theta^{(t)}_k$ with $\tau_{k,t}(\theta) \in \Theta^{(t)}_k$, where 
\begin{align}
    \tau_{k,t}(\theta) = t \theta_0 + (1-t) \Pi_{\Theta_k}\lrp{\frac{\theta - t \theta_0}{1-t}} = t \theta_0 + (1-t) \Pi_{\Theta_k}\lrp{\Pi_\Theta(\eta)}
\end{align}
This means that for $\eta \in \Theta_k$, $\vartheta = t\theta_0 + (1-t)\eta \in \Theta^{(t)}_k$, and $\theta = t \theta_0 + (1-t) \Pi_{\Theta}(\eta) \in \Theta^{(t)}$,  we have 
\begin{align}
    \|\tau_{k,t}(\theta) - \vartheta\|_2   = (1-t) \|\Pi_{\Theta_k}(\Pi_{\Theta}(\eta)) - \eta\|_2 \leq (1-t)\|\Pi_{\Theta}(\eta) - \eta\|_2 \leq (1-t) s_k. 
\end{align}
The first inequality above relies on the fact that $\eta \in \Theta_k$, and for any $\zvec$, $\|\Pi_{\Theta_k}(\zvec) - \eta\|_2 \leq \|\zvec - \eta\|_2$. The last inequality is by~\Cref{assump:general-dual-limits}, where $\{s_k\}_{k \geq 1}$ denotes a sequence converging to $0$. Thus, we can conclude
\begin{align}
    |H_k(\tau_{k,t}(\theta, P_k) - H_k(\vartheta, P_k)| \;\leq\; \omega_t\lrp{(1-t) s_k} \quad \implies \quad H_k(\vartheta, P_k) \leq  F_{k,t}(\theta) + \omega_t\lrp{s_k}, 
\end{align}
since $\omega_t(\cdot)$ is non-decreasing. 
Now, recall that $\vartheta \in \Theta_k^{(t)}$ was an $\epsilon$-sub-optimal point, which means that 
\begin{align}
    A_k(t) - \epsilon \;\leq\; F_{k,t}(\theta) + \omega_t(s_k) \;\leq \;  \widetilde{A}_k(t) + \omega_t(s_k). \label{eq:approx-constaint-proof-4}
\end{align}
Since $\epsilon>0$ was arbitrary, it follows that $A_k(t) \leq \widetilde{A}_k(t) + \omega_t(s_k)$, and  together with~\eqref{eq:approx-constraint-proof-3}, this implies that  
\begin{align}
    0 \leq A_k(t) - \widetilde{A}_k(t) \leq \omega_t(s_k) \;   \quad \implies \quad 
    \lim_{k \to \infty} \left \lvert A_k(t) - \widetilde{A}_k(t) \right \rvert = 0. \label{eq:approx-constraint-proof-5}
\end{align}
Now, combining the above statement with~\eqref{eq:approx-constraint-proof-6}, we get the required 
\begin{align}
    \lim_{k \to \infty} \left\lvert A_k(t) - A(t) \right\rvert \leq \lim_{k \to \infty} \left( \left\lvert A_k(t) - \widetilde{A}_k(t) \right\rvert  +  \left\lvert \widetilde{A}_k(t) - A(t) \right\rvert  \right) = 0. 
\end{align}
This completes the proof.

\subsubsection{Proof of~\Cref{lemma:approx-constraint-unique-limit}}
\label{proof:approx-constraint-unique-limit}

To start off the proof, we introduce the notation 
\begin{align}
    S_k \coloneqq \sup_{t \in (0,1)} A_k(t), \qtext{and} S \coloneqq \sup_{t \in (0,1)} A(t). 
\end{align}
Our goal is to show that $S_k \to S$, which implies that 
\begin{align}
    \lim_{k \to \infty}  \sup_{t \in (0,1)} A_k(t) = \sup_{t \in (0,1)} A(t) = \sup_{t \in (0,1)} \lim_{k \to \infty} A_k(t), 
\end{align}
justifying the interchange of $\lim_{k \to \infty}$ and $\sup_{t \in (0,1)}$ that we need to prove~\Cref{theorem:general-approximate-constraint}. 

\paragraph{The proof of $\liminf_{k \to \infty} S_k \geq S$.} This is the easy direction of the proof. Fix any $t \in (0,1)$, and for each $k \geq 1$, we have 
\begin{align}
    S_k = \sup_{u \in (0,1)} A_k(u)  \geq A_k(t) \quad \implies \quad \liminf_{k \to \infty} S_k \geq \liminf_{k \to \infty} A_k(t) = A(t),  
\end{align}
where the last equality uses~\Cref{lemma:approx-constraint-interchange-2}. Taking a supremum over $t$ gives the required inequality. 

\paragraph{The proof of $\limsup_{k \to \infty} S_k \leq S$.} This is the nontrivial part of the proof, and uses the concavity of the dual objective for each $k$, and the uniform boundedness assumption. In particular, we assume that there exist constants $L>-\infty$ and $U < \infty$, such that for all sufficiently large $k$, we have 
\begin{align}
    S_k \leq U, \qtext{and} H_k(\theta_0, P_k) \geq L. 
\end{align}
The first inequality is justified by the fact that due to~\Cref{lemma:approx-constraint-interchange-1}, $S_k = \sup_{\theta \in \Theta_k} H_k(\theta, P_k) = I_k$, which by~\Cref{lemma:approx-contraint-primal-1} converges to $I(P, \calC, g) < \infty$. Hence  we may set $U = I(P, \calC, g) + 1$, and there must exist as finite $k'$, such that for all $k \geq k'$, we have $I_k \leq U$. Also, the existence of an $L> -\infty$ is a consequence of the stronger uniform boundedness assumption on $H_k$ in~\Cref{assump:general-dual-function}. 

Now, let us fix a $t\in (0,1)$ and an $\epsilon >0$, and select a $\vartheta_{k,\epsilon} \in \Theta_k$, such that $H_k(\vartheta_{k,\epsilon}, P_k) \geq S_k - \epsilon$. On shrinking this parameter, we get $\vartheta^{(t)}_{k, \epsilon} = t \theta_0 + (1-t) \vartheta_{k,\epsilon}$ which is a member of $\Theta^{(t)}_k$. Since $H_k(\cdot, P_k)$ is concave by assumption, we have 
\begin{align}
    A_k(t) = \sup_{\vartheta \in \Theta_k^{(t)}} H_k(\vartheta, P_k) \geq H_k(\vartheta^{(t)}_{k,\epsilon}, P_k) \geq t H_k(\theta_0) + (1-t) H_k(\vartheta_{k, \epsilon}, P_k) \geq t L + (1-t) (S_k - \epsilon). 
\end{align}
On simplifying this implies that 
\begin{align}
    (1-t)(S_k- \epsilon)   \leq A_k(t) - t L \quad \implies \quad S_k - A_k(t) \leq \frac{t}{1-t} \lrp{A_k(t)- L} + \epsilon.  
\end{align}
Since $\epsilon>0$ was arbitrary, and using  $A_k(t) \leq \sup_{t} A_k(t) = S_k \leq U$, and $H_k(\theta_0, P_k) \geq L$, we get
\begin{align}
    S_k - A_k(t) \leq \frac{t(U-L)}{1-t} \quad \implies \limsup_{k \to \infty} S_k \leq \lim_{k \to \infty} A_k(t) + \frac{t (U-L)}{1-t} = A(t) + \frac{t(U-L)}{1-t},
\end{align}
where we used the fact that $A_k(t) \to A(t)$, and that $U$ and $L$ are independent of $k$. Next, we bound $A(t)$ with a supremum over all $t' \in (0,1)$, to get 
\begin{align}
    \limsup_{k \to \infty} S_k \; \leq \; A(t) + \frac{t(U-L)}{1-t} \leq \lrp{\sup_{t' \in (0,1)} A(t')} + \frac{t(U-L)}{1-t} = I(P, g, \calC) + \frac{t(U-L)}{1-t}. 
\end{align}
So this inequality is true for all $t \in (0,1)$, and we get the required conclusion by taking $t \downarrow 0$ since $U, L$ are independent of $t$. 

\paragraph{Unique Limit.} Finally, if there exists a unique function $H(\cdot, P): \Theta \to \R \cup \{\pm \infty\}$, such that it agrees with $H^{(t)}(\cdot, P)$ on $\Theta^{(t)}$ for all $t$, then by~\Cref{lemma:approx-constraint-interchange-1}, we get that 
\begin{align}
    \sup_{t \in (0,1)} \sup_{\theta \in \Theta^{(t)}} H^{(t)}(\theta, P) =  \sup_{t \in (0,1)} \sup_{\theta \in \Theta^{(t)}} H(\theta, P) = \sup_{\theta \in \Theta} H(\theta, P), 
\end{align}
as required. This completes the proof of~\Cref{lemma:approx-constraint-unique-limit}.

\subsection{Proof of~\Cref{theorem:general-constraint-KL}}
\label{proof:general-constraint-KL}
To prove this result, we will verify the assumptions required by~\Cref{theorem:general-approximate-constraint}. In particular, note that relative entropy satisfies DPI and weak lower-semicontinuity~(\Cref{assump:general-divergence}), the constraint function $g$ and the set $\calC$ satisfy~\Cref{assump:general-constraint},  and the discretization channels are assumed to satisfy~\Cref{assump:general-discretization}. The nontrivial steps lie in verifying~\Cref{assump:general-dual-function} and~\Cref{assump:general-dual-limits}, which we will break down into four steps. First, for each discretized problem $I_k$, we derive the corresponding dual using finite-dimensional  convex duality. Second, using the existence of an interior point in the constraint set, we show the crucial fact that the dual maximizers can be restricted to compact convex sets uniformly in $k$. Third, we verify the conditions required by~\Cref{assump:general-dual-function} on these restricted domains, and finally, we verify~\Cref{assump:general-dual-limits} by proving the convergence of the truncated dual domains and of the dual objectives, which allows us to invoke~\Cref{theorem:general-approximate-constraint} to complete the proof. 

Throughout the proof, we will use $\omega_g$ to denote the modulus of continuity of $g$
\begin{align}
\omega_g(\delta)\coloneqq \sup\{\|g(x)-g(x')\|_\infty:\ \|x-x'\|_\infty\le \delta\}, \qtext{satisfying} \lim_{\delta \to 0} \omega_g(\delta)=0. 
\end{align}
Since we have assumed that $g$ is Lipschitz~(say with constant $L_g$), we can simply take $\omega_g(\delta) \leq L_g  \delta$. 
Let $\mathcal M=\mathrm{conv}(g(\mathcal X))$, and observe that $\mathcal M$ is compact and convex, and by assumption we know that  there exists an $\mvec^\circ\in \calC \cap \mathring{\calM}$.
Next, recall that for each $k\ge 1$, $V_k\subset \calX$ is a $\Delta_k$-cover of $\calX$ under $\|\cdot\|_\infty$,
with $\Delta_k\downarrow 0$, and set $\eta_k\coloneqq \omega_g(\Delta_k)$ and define the inflated constraint sets
$\calC_k\coloneqq \calC+B_\infty(0,\eta_k)$. Let $P_k \coloneqq P\mathcal K_k$, so that $P_k$ is supported on $V_k$. Moreover, by uniform continuity of bounded continuous test functions on compact $\mathcal X$ and the fact that $\calK_k(\cdot\mid \xvec)$ is supported on $B_\infty(\xvec,\Delta_k)$, we have $P_k\Longrightarrow P$. Finally, introduce the term $G_\infty\coloneqq \sup_{\xvec\in\mathcal X}\|g(\xvec)\|_\infty<\infty$.

The starting point of the proof is to obtain the dual for the discretized version of the problem. For any $k\ge 1$, denote the
minimum divergence value by
\begin{align}
I_k \coloneqq  \inf\Big\{\mathrm{KL}(P_k,Q):\ Q\in\mathcal P(V_k),\ \mathbb E_Q[g(X)]\in \calC_k\Big\}.
\end{align}
Since there exists an $\mvec^\circ \in \calC \cap \mathring{\calM}$, by the compactness of $g(\calX)$ we can show the existence of $Q^\circ \in \calP(\calX)$ such that $\mathbb{E}_{Q^\circ}[g(X)] = \mvec^\circ$. On discretizing $Q^\circ$ with $\calK_k$, we observe that $\|\mathbb{E}_{Q^\circ \calK_k}[g(X)] - \mvec^\circ\|_\infty \leq \eta_k \coloneqq \omega_g(\Delta_k) \leq L_g \Delta_k$, it follows that  $\mathbb{E}_{Q^\circ \calK_k}[g(X)]$ lies in the interior of $\calC_k$. This also implies that for small enough $\epsilon$, the distribution $Q_{k, \epsilon} = (1-\epsilon) Q^\circ \calK_k + \epsilon P \calK_k$ has strictly positive mass on the support of $P_k$ and $\mathbb{E}_{Q_{k,\epsilon}}[g(X)]$ also lies in the interior of $\calC_k$. This $Q_{k,\epsilon}$ is a strictly feasible point for the discretized primal problem, and hence strong duality and dual attainment hold  by Slater's criterion.

Define the restricted dual domain
\begin{align}
\widetilde\Theta_k
\coloneqq \Big\{(\boldlambda,\gamma)\in\mathbb R^J\times\mathbb R:\ \gamma-\langle \boldlambda,g(\vvec)\rangle>0\ \ \forall \vvec\in V_k\Big\},
\end{align}
and define the support function $\sigma_{\calC_k}(\boldlambda)\coloneqq \sup_{\cvec\in \calC_k}\langle \boldlambda,\cvec\rangle$, and note  that $\inf_{\cvec\in \calC_k}\langle \boldlambda,\cvec \rangle \;=\; -\,\sigma_{\calC_k}(-\boldlambda)$. Repeating the argument of~\Cref{theorem:finite-support}, we obtain the following dual representation of $I_k$:
\begin{align}
I_k = \sup_{(\boldlambda,\gamma)\in \widetilde\Theta_k} \Big\{\mathbb E_{P_k}\big[\log(\gamma-\langle\boldlambda,g(X)\rangle)\big] +1-\gamma +\inf_{\cvec \in \calC_k}\langle \boldlambda,\cvec\rangle \Big\}.
\end{align}
Moreover, the supremum is attained by some maximizer $(\boldlambda_k^\star,\gamma_k^\star)\in \widetilde\Theta_k$. Note also that
$(\boldsymbol{0},1)\in \widetilde\Theta_k$ and achieves value $0$, so $I_k\ge 0$, and every maximizer satisfies
\begin{align}
\mathbb E_{P_k}\big[\log(\gamma_k^\star-\langle\boldlambda_k^\star,g(X)\rangle)\big] +1-\gamma_k^\star +\inf_{\cvec\in \calC_k}\langle \boldlambda_k^\star,\cvec \rangle \;\ge\;0. \label{eq:opt-nonneg}
\end{align}

We now state a key observation that allows us to restrict attention to compact subsets of the dual domain without losing optimality.
It is this step that utilizes the non-empty intersection of $\calC$ and the  interior of $\calM = \mathrm{conv}(g(\mathcal X))$.

\begin{lemma}\label{lemma:KL-general-limit-1}
There exists a constant $R<\infty$ such that for all $k$ large enough, there exists a pair of dual optimizers
$(\boldlambda_k^\star,\gamma_k^\star)$ satisfying $\|(\boldlambda_k^\star,\gamma_k^\star)\|_2\le R$.
\end{lemma}

\begin{proof}
Since $\mvec^{\circ}\in \mathring{\calM}$, where $\calM\coloneqq \mathrm{conv}(g(\calX))$, we have that for every $\uvec \in \R^J$ with $\|\uvec\|_1=1$, $\langle \uvec,\mvec^{\circ}\rangle < \max_{\mvec\in \calM}\langle \uvec,\mvec\rangle$.
The map $\uvec\mapsto \langle \uvec,\mvec^{\circ}\rangle-\max_{\mvec\in \calM}\langle \uvec,\mvec\rangle$ is continuous on the compact
set $\{\uvec:\|\uvec\|_1=1\}$, hence there exists a uniform margin
\begin{align}
\nu \;\coloneqq \;\inf_{\|\uvec\|_1=1}\Big(\max_{\mvec\in \calM}\langle \uvec,\mvec\rangle-\langle \uvec,\mvec^{\circ}\rangle\Big)\;>\;0.
\end{align}

We will now use this ``margin'' property to derive a uniform bound on dual maximizers. Fix a $k\ge 1$, and select any
$(\boldlambda,\gamma)\in \widetilde\Theta_k$. If $r\coloneqq \|\boldlambda\|_1$ is equal to $0$, there is nothing to prove, so assume $r>0$ and set
$\uvec\coloneqq \boldlambda/r$ so that $\|\uvec\|_1=1$. Define the following terms:
\begin{align}
M(\uvec)\coloneqq \max_{\xvec\in\calX}\langle \uvec,g(\xvec)\rangle, \qquad M_k(\uvec)\coloneqq \max_{\vvec\in V_k}\langle \uvec,g(\vvec)\rangle, \qquad \gamma'\coloneqq \gamma-r\,M_k(\uvec).
\end{align}
Feasibility of $(\boldlambda,\gamma)$ implies $\gamma'>0$, and observe that for any $\vvec\in V_k$, we have  $\gamma-\langle \boldlambda,g(\vvec)\rangle
=\gamma' + r\big(M_k(\uvec)-\langle \uvec,g(\vvec)\rangle\big)$.
Since $\|\uvec\|_1=1$ and $\|g(\cdot)\|_\infty\le G_\infty$, this implies $\langle \uvec,g(\vvec)\rangle\in[-G_\infty,G_\infty]$, and hence
\begin{align}
0\le M_k(\uvec)-\langle \uvec,g(\vvec)\rangle\le 2G_\infty \quad\text{and}\quad \gamma-\langle \boldlambda,g(\vvec)\rangle\le \gamma'+2G_\infty r.
\end{align}
Together these bounds imply that $ \mathbb E_{P_k}\!\big[\log(\gamma-\langle \boldlambda,g(X)\rangle)\big]\le \log(\gamma'+2G_\infty r)$. 
Moreover, $\inf_{\cvec\in\calC_k}\langle \boldlambda,\cvec\rangle=r\inf_{\cvec\in\calC_k}\langle \uvec,\cvec\rangle$ and
$\gamma=\gamma'+rM_k(\uvec)$, so the dual objective satisfies
\[
\mathbb E_{P_k}\!\big[\log(\gamma-\langle \boldlambda,g(X)\rangle)\big] + 1-\gamma + \inf_{\cvec\in\calC_k}\langle \boldlambda,\cvec\rangle
\le
\log(\gamma'+2G_\infty r) + 1-\gamma' - r\Big(M_k(\uvec)-\inf_{\cvec\in\calC_k}\langle \uvec,\cvec\rangle\Big).
\label{eq:KL-general-proof-1}
\]

We will now obtain a lower bound on the gap term in \eqref{eq:KL-general-proof-1}. Since $\mvec^{\circ}\in\calC\subseteq \calC_k$, we have $\inf_{\cvec\in\calC_k}\langle \uvec,\cvec\rangle\le \langle \uvec,\mvec^{\circ}\rangle$, hence
$ M_k(\uvec)-\inf_{\cvec\in\calC_k}\langle \uvec,\cvec\rangle \ge M_k(\uvec)-\langle \uvec,\mvec^{\circ}\rangle$. 
Also, because $V_k$ is a $\Delta_k$-cover of $\calX$ and $\eta_k=\omega_g(\Delta_k)$, we have
$ M_k(\uvec)=\max_{\vvec\in V_k}\langle \uvec,g(\vvec)\rangle
\ge \max_{\xvec\in\calX}\langle \uvec,g(\xvec)\rangle-\eta_k
= M(\uvec)-\eta_k$. 
Combining these two inequalities gives us  $M_k(\uvec)-\inf_{\cvec\in\calC_k}\langle \uvec,\cvec\rangle \ge M(\uvec)-\langle \uvec,\mvec^{\circ}\rangle-\eta_k \ge \nu-\eta_k$. Since $\eta_k \downarrow 0$, for all sufficiently large $k$, we have $\eta_k\le \nu/2$, hence $M_k(\uvec)-\inf_{\cvec\in\calC_k}\langle \uvec,\cvec\rangle\ge \nu/2$. Plugging this into \eqref{eq:KL-general-proof-1} yields~(for all sufficiently large $k$),
\begin{align}
\mathbb E_{P_k}\!\big[\log(\gamma-\langle \boldlambda,g(X)\rangle)\big] + 1-\gamma + \inf_{\cvec\in\calC_k}\langle \boldlambda,\cvec\rangle \le \log(\gamma'+2G_\infty r) + 1-\gamma' - (\nu/2)\,r.
\end{align}
Using $\log y\le y-1$ with $y=(\gamma'+2G_\infty r)/(1+2G_\infty r)$ gives
$ \sup_{\gamma'>0}\{\log(\gamma'+2G_\infty r)+1-\gamma'\}\le \log(1+2G_\infty r)+1$, 
and thus for all sufficiently large $k$, we have
\begin{align}
\mathbb E_{P_k}\!\big[\log(\gamma-\langle \boldlambda,g(X)\rangle)\big] + 1-\gamma + \inf_{\cvec\in\calC_k}\langle \boldlambda,\cvec\rangle
\le \log(1+2G_\infty r)+1-(\nu/2)\,r.
\end{align}
The right-hand side tends to $-\infty$ as $r\to\infty$. Since any maximizer $(\boldlambda_k^\star,\gamma_k^\star)$ satisfies \eqref{eq:opt-nonneg}, it follows that $\|\boldlambda_k^\star\|_1$ is uniformly bounded for all sufficiently large $k$. Let us call the bound $R_{\boldlambda}$.

Finally, feasibility gives $\gamma_k^\star>\max_{\vvec\in V_k}\langle \boldlambda_k^\star,g(\vvec)\rangle\ge
-\|\boldlambda_k^\star\|_1 G_\infty\ge -R_{\boldlambda}G_\infty$, giving a uniform lower bound on $\gamma_k^\star$.
Moreover, since $\calC$ is compact and $\eta_k\to 0$, the sets $\calC_k$ are uniformly bounded; thus
$\inf_{\cvec\in\calC_k}\langle \boldlambda_k^\star,\cvec\rangle\le \|\boldlambda_k^\star\|_1\sup_{\cvec\in\calC_k}\|\cvec\|_\infty$
is uniformly bounded above. Using again that the term $1-\gamma$ drives the objective to $-\infty$ as $\gamma\to\infty$,
while \eqref{eq:opt-nonneg} must hold at a maximizer, we obtain a uniform upper bound $\gamma_k^\star\le R_\gamma$ for all sufficiently
large $k$. Enlarging constants if needed, there exists $R<\infty$ such that for all sufficiently large $k$,
\[
\|(\boldlambda_k^\star,\gamma_k^\star)\|_2\le R. \label{eq:theta-star-bound}
\]
This completes the proof.
\end{proof}

We will now verify that conditions required in~\Cref{assump:general-dual-function} are satisfied. 
To do that, we first note that by Lemma~\ref{lemma:KL-general-limit-1}, for all sufficiently large $k$, we can restrict the attention to the following restricted dual domain without loss of optimality: 
\begin{align}
\Theta_k \coloneqq \overline{\widetilde\Theta}_k\cap\{\theta:\ \|\theta\|_2\le R\},\qquad \theta=(\boldlambda,\gamma),
\end{align}
and note that $\theta_0\coloneqq(\boldsymbol{0},1)\in\mathring{\Theta}_k$ and $H_k(\theta_0,P_k)=0 > -\infty$.

\begin{lemma}
\label{lemma:general-KL-2} 
Define, for $\xvec\in\calX$ and $\theta=(\boldlambda,\gamma)$,
\[
\psi_k(\xvec,\theta)\coloneqq\log\!\big(\gamma-\langle \boldlambda,g(\xvec)\rangle\big),
\qquad
b_k(\theta)\coloneqq1-\gamma+\inf_{\cvec\in\calC_k}\langle \boldlambda,\cvec\rangle,
\qquad
H_k(\theta,P_k)\coloneqq \mathbb{E}_{P_k}[\psi_k(X,\theta)]+b_k(\theta).
\]
Then, these terms satisfy the conditions required by~\Cref{assump:general-dual-function}. 
\end{lemma}
\begin{proof} We will verify the following: (i) concavity of $H_k$, (ii) uniform boundedness of $H_k$, (iii) equicontinuity of $H_k$, and (iv) uniform Lipschitz property of $\psi_k$.

\emph{Concavity.} For each fixed $\xvec$, the map $\theta\mapsto \log(\gamma-\langle \boldlambda,g(\xvec)\rangle)$ is concave on its domain. Since expectation preserves concavity, the term $1-\gamma$ is affine in $(\boldlambda, \gamma)$, and
$\boldlambda\mapsto \inf_{\cvec\in\calC_k}\langle \boldlambda,\cvec\rangle$ is concave (infimum of linear maps), we can conclude that the objective $H_k(\cdot,P_k)$ is concave on $\Theta_k$ for all $k \geq 1$.

To verify the other three conditions, we need to ensure that on the retracted domain, the argument of $\log(\cdot)$ is strictly bounded away from $0$.  For $t\in(0,1)$ define the retracted domain $\Theta_k^{(t)}\coloneqq t(\boldsymbol{0},1)+(1-t)\Theta_k$.
Take any $\theta_t=(\boldlambda_t,\gamma_t)\in\Theta_k^{(t)}$. Then $\theta_t=t(\boldsymbol{0},1)+(1-t)\theta$ for some $\theta=(\boldlambda,\gamma)\in\Theta_k$, and for all $\vvec\in V_k$,
\begin{align}
\gamma_t-\langle \boldlambda_t,g(\vvec)\rangle
= t+(1-t)\big(\gamma-\langle \boldlambda,g(\vvec)\rangle\big)\ge t.
\end{align}
Now fix $\xvec\in\calX$ and choose $\vvec\in V_k$ with $\|\xvec-\vvec\|_\infty\le \Delta_k$. Using the $L_g$-Lipschitz assumption on $g$,
\begin{align}
\gamma_t-\langle \boldlambda_t,g(\xvec)\rangle \ge \gamma_t-\langle \boldlambda_t,g(\vvec)\rangle - \|\boldlambda_t\|_1\,\|g(\xvec)-g(\vvec)\|_\infty \ge t - \|\boldlambda_t\|_1 L_g\Delta_k.
\end{align}
Since $\|\boldlambda_t\|_1\le \sqrt{J}\|\theta_t\|_2\le \sqrt{J}R$, for all sufficiently large $k$ we have $\sqrt{J}R L_g\Delta_k\le t/2$, and hence
\begin{align}
\gamma_t-\langle \boldlambda_t,g(\xvec)\rangle \ge t/2, \qtext{for all} \xvec\in\calX,\text{ and } \theta_t\in\Theta_k^{(t)}.
\end{align}
We now proceed to the verification of the remaining three properties. 

\emph{Uniform boundedness.}
Observe that  $t/2 \; \leq \;\gamma_t-\langle \boldlambda_t,g(\xvec)\rangle \leq |\gamma_t|+\|\boldlambda_t\|_1\|g(\xvec)\|_\infty \leq R+\sqrt{J}R G_\infty$, which implies 
 $\sup_{\xvec,\theta_t}|\psi_k(\xvec,\theta_t)|<\infty$ uniformly over large $k$. Moreover,  since $\calC$ is compact and $\calC_k=\calC+B_\infty(0,\eta_k)$ are uniformly bounded for all $k \geq 1$, so
$\sup_{\theta_t\in\Theta_k^{(t)}}|b_k(\theta_t)|<\infty$ uniformly over large $k$.
Therefore $H_k(\cdot,P_k)$ is uniformly bounded on $\Theta_k^{(t)}$.

\emph{Equicontinuity in $\theta$.}
On $\Theta_k^{(t)}$, the log-argument is at least $t/2$, so for any $\xvec$ and $\theta,\theta'\in\Theta_k^{(t)}$, we have 
\begin{align}
|\psi_k(\xvec,\theta)-\psi_k(\xvec,\theta')| &\leq \frac{2}{t}\Big(|\gamma-\gamma'|+|\langle \boldlambda-\boldlambda',g(\xvec)\rangle|\Big) \leq \frac{2}{t}\Big(|\gamma-\gamma'|+\sqrt{J}G_\infty\|\boldlambda-\boldlambda'\|_2\Big) \\
|b_k(\theta)-b_k(\theta')| &\leq |\gamma-\gamma'|+\sup_{\cvec\in\calC_k}\|\cvec\|_2\,\|\boldlambda-\boldlambda'\|_2 \leq |\gamma-\gamma'|+B\,\|\boldlambda-\boldlambda'\|_2,
\end{align}
for some $B<\infty$ independent of large $k$. Combining and taking expectations yields a Lipschitz modulus $\omega_t(r)=L_t^{(\theta)}r$ such that
\begin{align}
|H_k(\theta,P_k)-H_k(\theta',P_k)|\leq \omega_t(\|\theta-\theta'\|_2),
\qquad \text{for all } \theta,\theta'\in\Theta_k^{(t)},
\end{align}
uniformly over sufficiently large $k$.

\emph{Uniform Lipschitz in $\xvec$.}
For any $\theta_t\in\Theta_k^{(t)}$ and $\xvec,\xvec'\in\calX$,
\begin{align}
|\psi_k(\xvec,\theta_t)-\psi_k(\xvec',\theta_t)|
\leq \frac{2}{t}\,|\langle \boldlambda_t,g(\xvec)-g(\xvec')\rangle|
\leq \frac{2}{t}\,\|\boldlambda_t\|_1\,\|g(\xvec)-g(\xvec')\|_\infty
\leq \frac{2}{t}\,\sqrt{J}R\,L_g\,\|\xvec-\xvec'\|_\infty.
\end{align}
Thus the Lipschitz constant required by Assumption~B.4 can be taken as $L_t = \frac{2}{t}\sqrt{J}R L_g$, uniformly over all
sufficiently large $k$. This completes the verification of~\Cref{assump:general-dual-function}. 
\end{proof}

Next, we verify the conditions required by~\Cref{assump:general-dual-limits}. 
\begin{lemma}\label{lemma:KL-general-limit-B5}
Fix $R<\infty$ as in Lemma~\ref{lemma:KL-general-limit-1} and define the compact convex domains
\[
\Theta_k \coloneqq  \overline{\widetilde\Theta}_k\cap\{\theta:\|\theta\|_2\leq R\},
\qquad
\Theta \coloneqq  \overline{\widetilde\Theta}\cap\{\theta:\|\theta\|_2\leq R\},
\qquad \theta=(\boldlambda,\gamma).
\]
Let $\theta_0=(\boldsymbol{0},1)$. For $t\in(0,1)$, define $\Theta_k^{(t)}\coloneqq t\theta_0+(1-t)\Theta_k$ and
$\Theta^{(t)}\coloneqq t\theta_0+(1-t)\Theta$. Let $\Pi_{\Theta_k}$ and $\Pi_\Theta$ denote Euclidean projections onto $\Theta_k$ and $\Theta$,
and let $\tau_{k,t}$ denote the ``identification map'' from~\Cref{assump:general-dual-limits}. 
Finally, define
\[
H(\theta,P)\coloneqq \mathbb{E}_P\!\Big[\log\big(\gamma-\langle \boldlambda,g(X)\rangle\big)\Big]+1-\gamma+\inf_{\cvec\in\calC}\langle \boldlambda,\cvec\rangle.
\]
Then the triple $(\Theta_k,H_k,P_k)$ satisfies Assumption~\ref{assump:general-dual-limits} with the above $(\Theta,\theta_0,\tau_{k,t})$.
In particular, for each $t\in(0,1)$ and each $\theta\in\Theta^{(t)}$,
\[
F_k(\theta) \coloneqq H_k(\tau_{k,t}(\theta),P_k) \quad \overset{k \to \infty}{\longrightarrow}\quad H(\theta,P),
\]
so the required limit exists and is finite on any countable dense subset $D_t\subset\Theta^{(t)}$.
\end{lemma}

\begin{proof}
We will verify the (i) Hausdorff convergence of the domains, and (ii)  the pointwise convergence of the dual objectives.

\emph{Hausdorff convergence of dual domains.} First, $\Theta\subseteq \Theta_k$ for all $k$, since positivity on all $\xvec\in\calX$ implies positivity on all $\vvec\in V_k$ and both
sets are intersected with the same Euclidean ball.
For the reverse direction, take any $\theta=(\boldlambda,\gamma)\in\Theta_k$ and any $\xvec\in\calX$. Choose $\vvec\in V_k$ with $\|\xvec-\vvec\|_\infty\le \Delta_k$, and using the $L_g$-Lipschitz assumption on $g$, note that 
\begin{align}
\gamma-\langle \boldlambda,g(\xvec)\rangle \ge \gamma-\langle \boldlambda,g(\vvec)\rangle - \|\boldlambda\|_1\,\|g(\xvec)-g(\vvec)\|_\infty \ge 0 - \|\boldlambda\|_1 L_g\Delta_k.
\end{align}
Since $\|\boldlambda\|_1\le \sqrt{J}\|\theta\|_2\le \sqrt{J}R$, we have
\begin{align}
\gamma-\langle \boldlambda,g(\xvec)\rangle \ge -\delta_k \quad\text{for all }\xvec\in\calX, \qquad \delta_k\coloneqq \sqrt{J}R L_g\Delta_k.
\end{align}
Hence $\theta'\coloneqq (\boldlambda,\gamma+\delta_k)\in \overline{\widetilde\Theta}$, and if  $\|\theta'\|_2\le R$ then $\theta'\in\Theta$ and $\|\theta-\theta'\|_2=\delta_k$. Otherwise, set $\alpha\coloneqq R/\|\theta'\|_2\in(0,1)$ and $\theta''\coloneqq \alpha\theta'$, and since the constraint $\gamma-\langle \boldlambda,g(\xvec)\rangle\ge 0$ is homogeneous under scaling, $\theta''\in\overline{\widetilde\Theta}$ and $\|\theta''\|_2=R$, hence $\theta''\in\Theta$. Moreover,
\begin{align}
\|\theta-\theta''\|_2 \le \|\theta-\theta'\|_2 + \|\theta'-\theta''\|_2
= \delta_k + (\|\theta'\|_2-R) \le 2\delta_k, \qtext{$\implies$}
d_H(\Theta_k,\Theta)\le 2\delta_k \longrightarrow 0.
\end{align}
This concludes the verification of the Hausdorff convergence of the dual domains. 

\emph{Pointwise convergence of dual objectives along $\tau_{k,t}$.}
First, fix a  $t\in(0,1)$ and $\theta\in\Theta^{(t)}$, and set $z\coloneqq (\theta-\theta_0)/(1-t)\in\Theta$. Then
$\mathrm{dist}(z,\Theta_k) \coloneqq \inf_{\theta \in \Theta_k} \|z - \theta\|_2 \le d_H(\Theta_k,\Theta)$, so
\[
\|\tau_{k,t}(\theta)-\theta\|_2
=(1-t)\|\Pi_{\Theta_k}(z)-z\|_2
=(1-t)\,\mathrm{dist}(z,\Theta_k)
\le (1-t)\,d_H(\Theta_k,\Theta)\to 0.
\]
Now,  write $\theta_k \coloneqq \tau_{k,t}(\theta)=(\boldlambda_k,\gamma_k)$, and observe that  $\theta_k\to\theta=(\boldlambda,\gamma)$ as $k \to \infty$.
Since $\theta_k\in\Theta_k^{(t)}$, the positivity margin from the proof of~\Cref{lemma:general-KL-2} yields that for all sufficiently large $k$, we have  $\gamma_k-\langle \boldlambda_k,g(\xvec)\rangle \ge t/2$, for all $\xvec\in\calX$.
Define $f_k(\xvec)\coloneqq \log(\gamma_k-\langle \boldlambda_k,g(\xvec)\rangle)$ and $f(\xvec)\coloneqq \log(\gamma-\langle \boldlambda,g(\xvec)\rangle)$.
Then $f_k$ and $f$ are bounded and continuous on $\calX$, and
\begin{align}
\big|\mathbb{E}_{P_k}[f_k(X)]-\mathbb{E}_P[f(X)]\big|
\leq \big|\mathbb{E}_{P_k}[f_k(X)-f(X)]\big| + \big|\mathbb{E}_{P_k}[f(X)]-\mathbb{E}_P[f(X)]\big|.
\end{align}
The second term tends to $0$ since $P_k\Rightarrow P$ and $f$ is bounded continuous. For the first term, the derivative bound $|\log a-\log b|\leq (2/t)|a-b|$ and boundedness of $\|g(\cdot)\|_\infty$ imply $\sup_{\xvec\in\calX}|f_k(\xvec)-f(\xvec)| \leq C_t\|\theta_k-\theta\|_2 \to 0$, hence $\mathbb{E}_{P_k}[f_k-f]\to 0$. Therefore we can conclude that $\mathbb{E}_{P_k}[f_k]\to \mathbb{E}_P[f]$.

Next, since $\gamma_k\to\gamma$, it remains to show
$\inf_{\cvec\in\calC_k}\langle \boldlambda_k,\cvec\rangle \to \inf_{\cvec\in\calC}\langle \boldlambda,\cvec\rangle$.
Decompose
\[
\inf_{\cvec\in\calC_k}\langle \boldlambda_k,\cvec\rangle - \inf_{\cvec\in\calC}\langle \boldlambda,\cvec\rangle
=
\Big(\inf_{\calC_k}\langle \boldlambda_k,\cvec\rangle-\inf_{\calC_k}\langle \boldlambda,\cvec\rangle\Big)
+
\Big(\inf_{\calC_k}\langle \boldlambda,\cvec\rangle-\inf_{\calC}\langle \boldlambda,\cvec\rangle\Big).
\]
The first term tends to $0$ because $\calC_k$ are uniformly bounded and $\boldlambda_k\to\boldlambda$.
For the second term, using $\calC_k=\calC+B_\infty(0,\eta_k)$ we have
$\inf_{\cvec\in\calC_k}\langle \boldlambda,\cvec\rangle=\inf_{\cvec\in\calC}\langle \boldlambda,\cvec\rangle-\eta_k\|\boldlambda\|_1$,
which tends to $\inf_{\cvec\in\calC}\langle \boldlambda,\cvec\rangle$ since $\eta_k\to 0$.
Hence $b_k(\theta_k)\to 1-\gamma+\inf_{\cvec\in\calC}\langle \boldlambda,\cvec\rangle$, and therefore
\begin{align}
H_k(\tau_{k,t}(\theta),P_k)\to H(\theta,P).
\end{align}
This limit is finite since the log-argument is bounded away from $0$ on $\Theta^{(t)}$. Finally, for each $t\in(0,1)$ choose any countable dense set $D_t\subset\Theta^{(t)}$ (e.g.,\ rational points in $\Theta^{(t)}$). Then the above convergence holds for every $\theta\in D_t$, verifying the last clause of Assumption~\ref{assump:general-dual-limits}.
\end{proof}

\section{Deferred Proofs from~\Cref{sec:statistical-applications}}
\label{proof:statistical-applications}

\subsection{Proof of~\Cref{theorem:sequential-testing}}
\label{proof:sequential-testing}

For the lower bound, first, recall that for any $\alpha\in(0,1)$ and any test $\tau'_\alpha \in \mathcal T(\boldmu, \alpha)$, 
\begin{align}
    \frac{\mathbb{E}[\tau'_\alpha]}{\log(1/\alpha)} \geq \frac{1}{\KLinf(P_X, \boldmu)} \quad \implies \quad \inf_{\tau'_\alpha \in \calT(\boldmu, \alpha)}\frac{\mathbb{E}[\tau'_\alpha]}{\log(1/\alpha)}  \geq \frac{1}{\KLinf(P_X, \boldmu)}, 
\end{align}
which immediately implies that $\liminf_{\alpha \downarrow 0} \inf_{\tau'_\alpha \in \calT(\boldmu, \alpha)} \mathbb{E}[\tau'_\alpha]/\log(1/\alpha) \geq 1/\KLinf(P_X, \boldmu)$. The first lower bound follows directly from an application of data-processing inequality. We refer the reader to \citet[Theorem 3.1]{agrawal2025stopping} for a proof of the first inequality.

Next, let $\calP_0 = \{P \in \calP(\calX): \mathbb{E}_{P}[X] = \boldmu\}$ and $\calP_1 = \{P \in \calP(\calX): \mathbb{E}_{P}[X] \not = \boldmu\}$ denote the null and alternative classes of distributions. The $\alpha$-correctness (under null) of the test introduced in~\Cref{def:sequential-test} follows from \citet[Lemma F.1]{agrawal2021optimal} and the dual formulation for $\KLinf$. \\ 

Finally, we now establish the sample complexity upper bound for this test.\\

\noindent{\bf Expected stopping time under the alternative.} Fix an alternative distribution $P_X \in \calP_1$, and let $\boldlambda^*$ denote the optimal dual variable attaining the maximum in $\KLinf(P_X, \boldmu) = \sup_{\boldlambda\in \dualdomain} \mathbb{E}_{P_X}[\log(1 - \boldlambda^T(X-\boldmu))] = \mathbb{E}_{P_X}[\log(1+(\boldlambda^*)^T(X-\boldmu))]$. Let $S^*_n = \sum_{i=1}^n \log(1+ (\boldlambda^*)^T(X-\boldmu))$, and observe that 
\begin{align}
    n \KLinf(\widehat{P}_n, \boldmu) = \sup_{\boldlambda \in \dualdomain} \sum_{i=1}^n \log (1 + \boldlambda^T(X_i-\boldmu)) \; \geq \; S_n^*. 
\end{align}
This implies that for any $\alpha \in(0,1)$, we have 
\begin{align}
    \tau_\alpha \leq \tau^*_\alpha \coloneqq \inf \{n \geq 1: S_n^* \geq \log(n^K/\alpha)\} \; \implies \; \mathbb{E}_{P_X}[\tau_\alpha] \leq  \mathbb{E}_{P_X}[\tau_\alpha^*]. 
\end{align}
Thus, it suffices to get an upper bound on $\mathbb{E}_{P_X}[\tau^*_\alpha]$. Since we have assumed that $P_X$ is such that $\boldlambda^*$ lies in the interior of its domain, it means that $S_n$ is a random walk with positive drift $\KLinf(P_X, \boldmu)$, and bounded increments~(say $C<\infty$). Furthermore, since $\lim_{n \to \infty} S_n^*/n$ converges almost surely to $\KLinf(P_X, \boldmu) > 0$ due to the strong law of large numbers, and $\log(n^K/\alpha)/n = 0$, we must have $\tau_\alpha^* < \infty$ almost surely. Thus, an application of Wald's identity implies that 
\begin{align}
    \KLinf(P_X, \boldmu) \mathbb{E}_{P_X}[\tau_\alpha^*] = \mathbb{E}_{P_X}[S_{\tau^*_\alpha}^*]  \leq \mathbb{E}_{P_X}\left[ K \log(\tau_\alpha^*) + \log(1/\alpha) + C \right]. 
\end{align}
This leads to the bound 
\begin{align}
    \mathbb{E}_{P_X}[\tau_\alpha] \leq \mathbb{E}_{P_X}[\tau_\alpha^*] = \frac{\log(1/\alpha)}{\KLinf(P_X, \boldmu)}\left(1 + o(1) \right), 
\end{align}
where $o(1)$ term converges to $0$ as $\log(1/\alpha)$ goes to $\infty$.  This completes the proof.

\subsection{Proof of~\Cref{prop:change-detection}}
\label{proof:change-detection}
We first prove the bounds on the ARL and detection delay for the scheme in~\Cref{def:change-detection}.
\paragraph{ARL bound for $N_\alpha$.} One way of establishing the ARL bound is to observe that 
\begin{align}
    N_\alpha = \inf \{n \geq 1: M_n \geq 1/\alpha\}, \qtext{where} M_n = \sum_{k \leq n} E_n^{(k)}, \quad E_n^{(k)} = \frac 1{\alpha} \boldsymbol{1}_{\tau_\alpha^{(k)} \leq n}. 
\end{align}
Hence, $\{M_n: n \geq 0\}$ is a  Shirayev-Roberts~(SR) type e-detector in the terminology of~\citet[Definition 2.6]{shin2023detectors}, which means that for any stopping time $N^*$, it satisfies the inequality $\mathbb{E}_{\infty, P}[N^*] \geq \mathbb{E}_{\infty, P}[M_{N^*}]$. Using this condition in particular with the stopping time $N_\alpha$, and observing that by definition we must have $M_{N_\alpha} \geq 1/\alpha$ almost surely, we obtain 
\begin{align}
    \mathbb{E}_{\infty, P}[N_\alpha] \geq \mathbb{E}_{\infty, P}[M_{N_\alpha}] \geq \frac 1{\alpha}, 
\end{align}
which is the required lower bound on the ARL under $T=\infty$.

\paragraph{Detection Delay.} Suppose the change in distribution from $P \in \calP_0$ to some $Q \in \calP_1$ happens at time $T \in \mathbb{N}$. Then, for $\mathcal{F}_T = \sigma(X_1, \ldots, X_T)$, note that 
\begin{align}
    \mathbb{E}_{T, P, Q}[(N_\alpha - T)^+ \mid \calF_T] & \leq \mathbb{E}_{T,P,Q}[(\tau_\alpha^{(T+1)} - T)^+ \mid \calF_T] = \mathbb{E}_{0, P, Q}[(\tau_\alpha^{(1)} - 0)^+ \mid \calF_0]. 
\end{align}
Since this is simply the expected value of the stopping time of~\Cref{def:sequential-test} under the alternative, we know by~\Cref{theorem:sequential-testing} that
\begin{align}
    \sup_{P \in \calP_0} \frac{\mathbb{E}_{T, P, Q}[(N_\alpha - T)^+ \mid \calF_T]}{\log(1/\alpha)} &\leq \frac{1}{\KLinf(Q, \boldmu_0)} \left( 1 + o(1) \right). 
\end{align}
On taking the $\mathrm{ess}\sup$ and the supremum over all $T \in \mathbb{N}$ and taking $\alpha \downarrow 0$, we get the required bound on $J_L(N_\alpha, P, Q)$. 

We now prove the lower bound on the detection delay for any test with the specified ARL constraint.
\paragraph{Lower Bound.} Let us introduce the notation 
\begin{align}
L(Q, \boldmu_0, \alpha) \coloneqq \inf_{N'_\alpha \in \calC(\boldmu_0, \alpha)} \sup_{P \in \calP_0} J_L(N'_\alpha, P, Q). 
\end{align}
Now, for an arbitrary $\epsilon >0$, let $P_\epsilon \in  \calP_0$ be a distribution with $\KL(Q, P_\epsilon) \leq \KLinf(Q, \boldmu_0) + \epsilon$. By the lower bound derived by~\citet[Theorem~3]{lorden1971procedures}, we know that 
\begin{align}
    L(Q, \boldmu_0, \alpha) \geq \inf_{N_\alpha' \in \calC(\boldmu_0, \alpha)} J_L(N_\alpha', P_\epsilon, Q) \geq \frac{\log(1/\alpha)}{\KL(Q, P_\epsilon)} \left(1 + o(1) \right) \geq \frac{\log(1/\alpha)}{\KLinf(Q, \boldmu_0) + \epsilon} \left(1 + o(1) \right). 
\end{align}
Thus, on dividing by $\log(1/\alpha)$ and taking $\alpha \to 0$, we get 
\begin{align}
    \liminf_{\alpha \downarrow 0} \frac{L(Q, \boldmu_0, \alpha)} {\log(1/\alpha)} \geq \frac{1}{\KLinf(Q, \boldmu_0) + \epsilon}. 
\end{align}
Since $\epsilon > 0$ was arbitrary, the result follows.

\end{document}